%% file: main.tex
\documentclass[11pt]{article}

\input{preamble}
\usepackage{times}

\usepackage{esvect}
\usepackage{thmtools} 
\usepackage{thm-restate}

\begin{document}

\title{Directed Isoperimetric Theorems for Boolean Functions on the Hypergrid
and an $\otilde(n\sqrt{d})$ Monotonicity Tester}

\author{Hadley Black\\
University of California, Los Angeles\\
{\tt hablack@ucla.edu}
\and
Deeparnab  Chakrabarty\thanks{Supported by NSF-CAREER award 2041920} \\
Dartmouth\\
{\tt deeparnab@dartmouth.edu}
\and
C. Seshadhri\thanks{Supported by NSF DMS-2023495, CCF-1740850, 1839317, 1813165, 1908384, 1909790, and ARO Award W911NF1910294.} \\
University of California, Santa Cruz\\
{\tt sesh@ucsc.edu}
}

\date{}

\maketitle

\def\Inf{I}

\def\tI{\Phi}

\def\sort{\mathtt{sort}}
\def\bA{\vv{A^{\bot}}}
\def\bB{\vv{B^{\bot}}}
\def\bbA{\vv{A^{\parallel}}}
\def\bbB{\vv{B^{\parallel}}}
\def\bT{\vv{\phi}}
\def\vu{\vv{u}}
\def\vvv{\vv{v}}
\def\bU{\vv{U}}
\def\bV{\vv{V}}
\def\bS{\vv{S}}
\def\vw{\vv{w}}
\def\vz{\vv{z}}
\newcommand{\sortdown}[1]{\left(#1\right)^{\downarrow}}
\newcommand{\sortup}[1]{\left(#1\right)^{\uparrow}}

\def\LHS{\mathsf{LHS}}
\def\RHS{\mathsf{RHS}}
\def\colchi{\color{red} \chi'}
\def\coordom{\succeq_{\mathsf{coor}}}
\def\majorizes{\succeq_{\mathsf{maj}}}
\def\majorizedby{\preceq_{\mathsf{maj}}}
\def\dominates{\trianglerighteq}

\begin{abstract}
The problem of testing monotonicity for Boolean functions on the hypergrid, $f:[n]^d \to \{0,1\}$ is a classic topic in property testing. When $n=2$, the domain is the hypercube. For the hypercube case, a breakthrough result of Khot-Minzer-Safra (FOCS 2015) gave a non-adaptive, one-sided tester making $\otilde(\eps^{-2}\sqrt{d})$ queries. Up to polylog $d$ and $\eps$ factors, this bound matches the $\widetilde{\Omega}(\sqrt{d})$-query non-adaptive lower bound (Chen-De-Servedio-Tan (STOC 2015), Chen-Waingarten-Xie (STOC 2017)). For any $n > 2$, the optimal non-adaptive complexity was unknown. A previous result of the authors achieves a $\otilde(d^{5/6})$-query upper bound (SODA 2020), quite far from the $\sqrt{d}$ bound for the hypercube.

In this paper, we resolve the non-adaptive complexity of monotonicity testing for all constant $n$, up to $\poly(\eps^{-1}\log d)$ factors. Specifically, we give a non-adaptive, one-sided monotonicity tester making $\otilde(\eps^{-2}n\sqrt{d})$ queries. From a technical standpoint, we prove new directed isoperimetric theorems over the hypergrid $[n]^d$. These results generalize the celebrated directed Talagrand inequalities that were only known for the hypercube.

\end{abstract}
\thispagestyle{empty}
\setcounter{page}{0}
\newpage

\input{intro.tex}

\input{prelims.tex}


%

\input{threshold_talagrand_setup.tex}

\input{semisorting_reduction.tex}

\input{telescoping_argument.tex}


\input{semisorted_threshold_talagrand.tex}


\input{tester_analysis.tex}


\input{thoughts-on-psi.tex}

\bibliography{derivative-testing.bib}
\bibliographystyle{alpha}

\end{document}

%% file: preamble.tex
	\usepackage{a4,geometry}

\usepackage{graphicx}
\usepackage{amsmath,amssymb,amsthm,mathtools}
\usepackage{paralist}
\usepackage{bm}
\usepackage{bbm}
\usepackage{xspace}
\usepackage{url}
\usepackage{fullpage, prettyref}
\usepackage{boxedminipage}
\usepackage{wrapfig}
\usepackage{ifthen}
\usepackage{color}
\usepackage{xcolor}
\usepackage{framed}
\usepackage{algorithmic,algorithm}
\usepackage[pagebackref,letterpaper=true,colorlinks=true,pdfpagemode=none,urlcolor=blue,linkcolor=blue,citecolor=violet,pdfstartview=FitH]{hyperref}
\usepackage[nameinlink]{cleveref}
\usepackage{fullpage}
\usepackage{esvect}
\usepackage{mdframed}
\usepackage{thmtools}
\usepackage{thm-restate}

\newtheorem{theorem}{Theorem}[section]
\newtheorem{conjecture}[theorem]{Conjecture}
\newtheorem{lemma}[theorem]{Lemma}
\newtheorem{claim}[theorem]{Claim}

\newtheorem{definition}[theorem]{Definition}

\newtheorem{fact}[theorem]{Fact}

\newtheorem{observation}[theorem]{Observation}

\newcommand{\ignore}[1]{}

\newcommand{\cC}{{\cal C}}
\newcommand{\cD}{\mathcal{D}}
\newcommand{\cE}{{\cal E}}

\newcommand{\cH}{{\cal H}}

\newcommand{\cR}{{\cal R}}

\newcommand{\eps}{\varepsilon}

\newcommand{\var}{\mathrm{var}}

\newcommand{\poly}{\mathrm{poly}}

\newcommand{\calH}{{\cal H}}

\newcommand{\bone}{{\bf 1}}
\newcommand{\ba}{\mathbf{a}}
\newcommand{\bb}{\mathbf{b}}

\newcommand{\be}{\mathbf{e}}

\newcommand{\bu}{\mathbf{u}}
\newcommand{\bv}{\mathbf{v}}
\newcommand{\bw}{\mathbf{w}}
\newcommand{\bx}{\mathbf{x}}
\newcommand{\by}{\mathbf{y}}
\newcommand{\bz}{\mathbf{z}}

\newcommand{\bA}{\boldsymbol{A}}

\newcommand{\bfS}{\boldsymbol{S}}

\newcommand{\RR}{\mathbb{R}}
\newcommand{\mbone}{\mathbbm{1}}

\newcommand{\norm}[1]{\left\lVert #1 \right\rVert}
\newcommand{\ceil}[1]{\lceil#1\rceil}
\newcommand{\Exp}{\EX}

\newcommand{\EX}{\mathbf{E}}

\newcommand{\eqdef}{:=}

\newcommand{\otilde}{\widetilde{O}}

\newcommand{\hyp}[1]{\{0,1\}^{#1}}
\newcommand{\dtal}[1]{T_{\Phi_\chi}(#1)}

\newcommand{\infl}[2]{I_#1(#2)}

\newcommand{\Phiinfl}[2]{\Phi_{#1}(#2)}

\newcommand{\sortline}[1]{\mathtt{sort}(#1)}
\newcommand{\sorti}[2]{\mathtt{sort}_#2(#1)}

\newcommand{\cper}{C_{per}}
\newcommand{\clarge}{C_{lar}}

\newcommand{\Sec}[1]{\hyperref[sec:#1]{\Cref*{sec:#1}}} 
\newcommand{\Eqn}[1]{\hyperref[eq:#1]{(\ref*{eq:#1})}} 
\newcommand{\Fig}[1]{\hyperref[fig:#1]{Fig.\,\ref*{fig:#1}}} 
\newcommand{\Tab}[1]{\hyperref[tab:#1]{Tab.\,\ref*{tab:#1}}} 
\newcommand{\Thm}[1]{\hyperref[thm:#1]{Theorem\,\ref*{thm:#1}}} 
\newcommand{\Fact}[1]{\hyperref[fact:#1]{Fact\,\ref*{fact:#1}}} 
\newcommand{\Lem}[1]{\hyperref[lem:#1]{Lemma\,\ref*{lem:#1}}} 
\newcommand{\Prop}[1]{\hyperref[prop:#1]{Prop.~\ref*{prop:#1}}} 
\newcommand{\Cor}[1]{\hyperref[cor:#1]{Corollary~\ref*{cor:#1}}} 
\newcommand{\Conj}[1]{\hyperref[conj:#1]{Conjecture~\ref*{conj:#1}}} 
\newcommand{\Def}[1]{\hyperref[def:#1]{Definition~\ref*{def:#1}}} 
\newcommand{\Alg}[1]{\hyperref[alg:#1]{Alg.~\ref*{alg:#1}}} 
\newcommand{\Obs}[1]{\hyperref[obs:#1]{Obs.~\ref*{obs:#1}}} 
\newcommand{\Ex}[1]{\hyperref[ex:#1]{Ex.~\ref*{ex:#1}}} 
\newcommand{\Clm}[1]{\hyperref[clm:#1]{Claim~\ref*{clm:#1}}} 
\newcommand{\Step}[1]{\hyperref[step:#1]{Step~\ref*{step:#1}}} 

\newcommand{\DeepC}[1]{{\color{red} DeepC: #1}}

\newcommand{\Omgt}{\widetilde{\Omega}}

%% file: intro.tex
\section{Introduction}

Monotonicity testing, especially over hypergrid domains, is one of the most well studied problems in property testing. 
We use $[n]$ to denote the set $\{1,2,\ldots, n\}$. The set $[n]^d$ is the $d$-dimensional hypergrid where $\bx \in [n]^d$ is a $d$-dimensional vector with $\bx_i \in [n]$.
The hypergrid is equipped with the natural partial order $\bx \preceq \by$ iff $\bx_i \leq \by_i$ for all $i\in [d]$. Note that when $n=2$, the hypergrid $[n]^d$ is isomorphic to the
hypercube $\{0,1\}^d$. 

Let $f:[n]^d \to \{0,1\}$ be a Boolean function defined on the hypergrid. The function $f$ is monotone if $f(\bx) \leq f(\by)$ whenever $\bx \preceq \by$.
The Hamming distance between two Boolean functions $f$ and $g$,
denoted as $\Delta(f,g)$, is the fraction of points where they differ.
The {\em distance to monotonicity} of a function $f:[n]^d \to \{0,1\}$ is defined as
$\eps_f := \min_{g~\textrm{monotone}} \Delta(f,g)$. 
The Boolean monotonicity testing problem on the hypergrid takes parameter $\eps$ and oracle access to $f:[n]^d \to \{0,1\}$. The objective is to design a randomized algorithm, called the tester,
that accepts a monotone function with probability $\geq 2/3$ and rejects a function $f$ with $\eps_f \geq \eps$ with probability $\geq 2/3$. 
A tester is one-sided if it accepts a monotone function with probability $1$. A tester is non-adaptive if all its queries are made in one round before seeing any responses. 

There has been a rich history of results on monotonicity testing over hypergrids,
with a significant focus on hypercubes~\cite{GGLRS00,DGLRRS99,ChSe13,ChSe13-j,BeRaYa14,ChenST14,ChDi+15,ChenDST15,KMS15,BeBl16,Chen17,BlackCS18,BlackCS20,BKR20,HY22}.
We discuss the history more in \Cref{sec:related}, but for now, we give the state of the art.
For hypercubes, after a long line of work, the breakthrough result~\cite{KMS15} of Khot, Minzer, and Safra gave
an $\otilde_\eps(\sqrt{d})$-query non-adaptive, one-sided tester. This result is tight due to a nearly matching $\widetilde{\Omega}(\sqrt{d})$-query lower bound for non-adaptive testers due to Chen, Waingarten, and Xie~\cite{Chen17}.
For general hypergrids, the best upper bound is the $\otilde_\eps(d^{5/6})$-query 
tester of the authors~\cite{BlackCS18,BlackCS20}. 

This $\widetilde{\Omega}(\sqrt{d})$ vs $\otilde(d^{5/6})$ gap for non-adaptive testers is a tantalizing and important open question in property testing.
Even for the domain $[3]^d$, the optimal non-adaptive monotonicity testing bound is unknown.
One of the main questions driving our work is: 

\begin{center}
\emph{Are there $\otilde_\eps(\sqrt{d})$-query monotonicity testers for domains beyond the hypercube?}
\end{center}

\paragraph{Directed isoperimetric theorems.} The initial seminal work on monotonicity testing,
by Goldreich, Goldwasser, Lehman, Ron, and Samorodnitsky~\cite{GGLRS00}
and Dodis, Goldreich, Lehman, Ron, Raskhodnikova and Samorodnitsky~\cite{DGLRRS99}
prove the existence of $\otilde_\eps(d)$-query testers. 
For almost a decade, it was not clear whether $o(d)$-query testers were possible. 
In~\cite{ChSe13-j}, the last two authors gave the first such tester via an exciting connection
with \emph{robust directed isoperimetric theorems}. Indeed, all $o(d)$-query testers
are achieved through such theorems.

Think of a Boolean function $f$ as the indicator for a subset of the domain.
The variance of $f$, $\var(f)$, is a measure of the volume of the indicated subset.
An isoperimetric theorem for Boolean functions relates the variance of $f$ to the ``boundary'' of the function which corresponds
to the sensitive edges and/or their endpoints.
The deep insight of these theorems comes from sophisticated ways of measuring boundary size, involving both the vertex and edge boundary.
A \emph{directed} isoperimetric theorem is an analog where we only measure ``up-boundary'' formed by monotonicity violations.
Rather surprisingly, in the directed case, one can replace the variance as a measure of volume by the distance to monotonicity.

In \Cref{tab:dir-iso}, we list some classic isoperimetric results and their directed analogues for the hypercube. 
For a point $\bx$, $I_f(\bx)$ is the number of sensitive edges incident to $\bx$.
We use $I_f$ to denote $\Exp_{\bx}[I_f(\bx)]$, the total influence of $f$, which the number of sensitive edges in $f$ divided by the domain size $2^d$.
The quantity $\Gamma_f$ is the vertex boundary size divided by $2^d$. The directed analogues of these, $I^-_f, \Gamma^-_f, I^-_f(\bx)$,
only consider sensitive edges that violate monotonicity.

\begin{table}[ht!]
\begin{center}
	\def\arraystretch{1.5}
	\begin{tabular} {| c | c | }
		\hline
		 Undirected Isoperimetry & Directed Isoperimetry \\ \hline
		$I_f \geq \Omega(\var(f))$ ~~~~ (\emph{Poincar\'{e} inequality, Folklore})& $I^-_f \geq \Omega(\eps_f)$ ~~~ (\emph{Goldreich et al.\cite{GGLRS00}})
		\\ \hline
		 $I_f\cdot \Gamma_f \geq \Omega(\var(f)^2)$ ~~~(\emph{Margulis~\cite{Mar74}})
		& $I^-_f \cdot \Gamma^-_f \geq \Omega(\eps^2_f)$ ~~(\emph{Chakrabarty, Seshadhri~\cite{ChSe13-j}})
		 \\ \hline
		 $\Exp_\bx\left[\sqrt{I_f(\bx)}\right] \geq \Omega(\var(f))$  (\emph{Talagrand ~\cite{Tal93}}) 
		& $\Exp_\bx\left[\sqrt{I^-_f(\bx)}\right] = \Omega(\frac{\eps_f}{\log d})$ ~~(\emph{Khot, Minzer, Safra~\cite{KMS15}})
		\\ \hline
	\end{tabular}
\end{center}
\caption{\em Boolean hypercube isoperimetry results and their directed analogues.
Pallavoor, Raskhodnikova, and Waingarten~\cite{PRW22} removed the $\log d$-dependence in the directed Talagrand inequality.}\label{tab:dir-iso}
\end{table}

Observe the remarkable parallel between the standard isoperimetric results and their directed versions.
The Talagrand inequality is the strongest statement, and implies all other bounds. The directed versions
imply the undirected versions, using standard inequalities regarding monotone functions. The~\cite{KMS15}
$\otilde_\eps(\sqrt{d})$-query tester is based on the directed Talagrand inequality. 

The story for hypergrids is much more complicated. From an isoperimetric perspective, a common approach is to consider the \emph{augmented hypergrid},
wherein we add edges between pairs in the same line. 
The dimension reduction technique in~\cite{DGLRRS99} used to prove the $\otilde_\eps(d)$ testers can be thought of as establishing a directed Poincar\'{e} inequality~.
In previous work~\cite{BlackCS18}, the authors proved a directed Margulis inequality, which led to the $\otilde_\eps(d^{5/6})$ query tester.
Another motivating question for our work is:

\begin{center}
\emph{Can the directed Talagrand inequality be generalized beyond the hypercube?}
\end{center}

\subsection{Main results} \label{sec:results}
We answer both questions mentioned above in the affirmative. To state our results more formally, we begin with some notation.
For any $i\in [d]$, we use $\be_i$ to denote the $d$-dimensional vector which has $1$ on the $i$th coordinate and zero everywhere else.
For a dimension $i$, a pair $(\bx,\by)$ is called \emph{$i$-aligned} if $\bx$ and $\by$ only differ on their
$i$-coordinate. An \emph{$i$-line} is a 1D line of $n$ points obtained by fixing all but the $i$th coordinate.

We define a notion of directed influence of Boolean functions on hypergrids, which generalizes the notion for Boolean functions on hypercubes.
In plain English, for a point $\bx$ we count the number of {\em dimensions} in which $\bx$ takes part in a violation. 
We call this the {\em thresholded negative influence} of $\bx$. Note that $\bx$ could participate in multiple violations along the same dimension.
Throughout this paper, we will be only talking about negative influences of functions on the hypergrid, and thus will often refer to the above
as just thresholded influence, and for brevity's sake we also don't use the superscript ``$-$'' in the notation below to denote the negative aspect.

%


\begin{definition}[Thresholded Influence]\label{def:phi-f}
	Fix $f:[n]^d \to \{0,1\}$ and a dimension $i\in [d]$. Fix a point $\bx \in [n]^d$. The thresholded influence of $\bx$ along coordinate $i$ is
    denoted $\Phi_f(\bx;i)$, and has value $1$ if there exists an $i$-aligned violation $(\bx,\by)$.
	The thresholded influence of $\bx$ is $\Phi_f(\bx) = \sum_{i=1}^d \Phi_f(\bx;i)$.
\end{definition}

Note that the thresholded influence coincides with the hypercube directed influence when $n=2$. 
Also note that for any $\bx$, $\Phi_f(\bx) \in \{0,1,\ldots, d\}$ and is independent of $n$.
We prove the following theorem, a directed Talagrand theorem for hypergrids, which generalizes the~\cite{KMS15} result.

\begin{theorem}\label{thm:dir-tal-uncolored}~
		Let $f:[n]^d \to \{0,1\}$ be $\eps$-far from monotone.
		\[
		\Exp_{\bx\in [n]^d}~\left[\sqrt{\Phi_{f}(\bx)}  \right] = \Omega\left(\frac{\eps}{\log n}\right)
		\]
\end{theorem}
\paragraph{Robust isoperimetric theorems and monotonicity testing.} For the application to monotonicity testing,
as~\cite{KMS15} showed, a significant strengthening of \Thm{dir-tal-uncolored} is required. 
The weakness of \Thm{dir-tal-uncolored}, as stated, is that the same violation/influence
is ``double-counted" at both its endpoints. 
The LHS can significantly vary  depending on whether we choose to only ``count" influences at zero-valued or one-valued points, and this is true even on the hypercube.
As a simple illustration, 
consider the function $f$ that is $1$ at the all zeros point and $0$ everywhere else.
Suppose we only count influences at one-valued points.
Then the only vertex with any $I_f^-(\bx)$ is the all $0$'s point, and this value is $d$. Therefore, the Talagrand objective is $\frac{\sqrt{d}}{2^d}$. 
On the other hand, if we count influences at zero-valued points, then $\Inf_f^-(\bx) = 1$ for the $d$ points $\be_1$ to $\be_d$,
and $0$ everywhere else. The Talagrand objective counted from zero-valued points is now much larger: $\frac{d}{2^d}$. 
Therefore, depending on how we count, one can potentially reduce the Talagrand objective, $\EX_\bx[\sqrt{\Inf_f^-(\bx)}]$.

~\cite{KMS15} define a general way of deciding which endpoint ``pays'' for a violated edge. Consider a {\em coloring}\footnote{\cite{KMS15} considered the colorings to be red/blue, but we find the $0,1$-coloring more natural.} $\chi: E \to \{0,1\}$ of every edge $(\bx,\by) \in E$ of the hypercube to either $0$ or $1$. Now, given a violated edge $(\bx, \by)$, we use this coloring to decide whose influence this edge contributes towards. More precisely, given this coloring $\chi$, the {\em colored} directed influence $\Inf^-_{f,\chi}(\bx)$ of $\bx$ is defined as 
the number of violated edges $(\bx, \by)$ incident on $\bx$ which have the same color as $f(\bx)$. 
Given a coloring, the {\em colorful} Talagrand objective equals the expected root colored directed influence. 
What~\cite{KMS15} prove is that no matter what coloring $\chi$ one chooses, 
the Talagrand objective is still large, and in particular $\Exp_\bx\left[\sqrt{I^-_{f,\chi}(\bx)}\right] = \Omega(\frac{\eps_f}{\log d})$.

We define the robust/colorful generalizations of the thresholded negative influence on hypergrids.
Consider the \emph{fully augmented hypergrid}, where we put the edge $(\bx,\by)$ if $\bx$
and $\by$ differ on only one coordinate.
Let $E$ be the set of edges in the fully augmented hypergrid.

\begin{definition}[Colorful Thresholded Influence]\label{def:phi-f-chi}
	Fix $f:[n]^d \to \{0,1\}$ and $\chi:E\to \{0,1\}$. Fix a dimension $i\in [d]$ and a point $\bx \in [n]^d$. The colorful thresholded negative influence of $\bx$ along coordinate $i$ is denoted $\Phi_{f,\chi}(\bx;i)$, and has value $1$ if there exists an $i$-aligned violation $(\bx,\by)$
    such that $\chi(\bx,\by) = f(\bx)$, and has value $0$ otherwise.
	The colorful thresholded negative influence of $\bx$ is $\Phi_{f,\chi}(\bx) = \sum_{i=1}^d \Phi_{f,\chi}(\bx;i)$.
\end{definition}

The main result of our paper is a robust directed Talagrand isoperimetry theorem for Boolean functions on the hypergrid.
It is a strict generalization of the KMS Talagrand theorem for hypercubes.
\begin{mdframed}[backgroundcolor=gray!20,topline=false,bottomline=false,leftline=false,rightline=false] 
	\begin{theorem}\label{thm:dir-tal}~
		Let $f:[n]^d \to \{0,1\}$ be $\eps$-far from monotone, and let $\chi:E\to \{0,1\}$ be an arbitrary coloring of the edges of the augmented hypergrid.
		\[
		\Exp_{\bx\in [n]^d}~\left[\sqrt{\Phi_{f,\chi}(\bx)}  \right] = \Omega\left(\frac{\eps}{\log n}\right)
		\]
	\end{theorem}
\end{mdframed}

As a consequence of this theorem, we can (up to log factors) resolve the question of non-adaptive
monotonicity testing on hypergrids with constant $n$. We note that the best bound for any $n > 2$
was $\otilde(d^{5/6})$. Even for the simplest non-hypercube case of $n=3$, it was open whether the optimal non-adaptive complexity of monotonicity testing is $\sqrt{d}$.

\begin{mdframed}[backgroundcolor=gray!20,topline=false,bottomline=false,leftline=false,rightline=false] 
\begin{theorem}\label{thm:mono-testing}~
    Consider Boolean functions over the hypergrid, $f:[n]^d \to \{0,1\}$.
    There is a one-sided, non-adaptive tester for monotonicity that makes $O(\eps^{-2} n \sqrt{d} \log^5(nd))$ queries.
\end{theorem}
\end{mdframed}

\paragraph{The importance of being robust.} We briefly explain why the 
robust Talagrand version is central
to the monotonicity testing application. All testers that have a $o(d)$-query complexity are versions of a \emph{path tester},
which can be thought of as querying endpoints of a directed random walk in the hypercube. Consider a function $f$
as the indicator for a set $\bone_f$, where the violating edges
form the ``up-boundary" between $\bone_f$ and its complement. To analyze the random walk, 
we would like to lower bound the probability that a random walk starts in $\bone_f$, crosses over the boundary, and stays
in $\overline{\bone_f}$, that is, the set of $0$'s. To analyze this, one needs some structural properties in the graph induced by the boundary edges, which~\cite{KMS15} 
express via their notion of a ``good subgraph''. In particular, one needs that there be a large number of edges, but also that 
they are regularly spread out among the vertices. It doesn't seem that the ``uncolored'' Talagrand versions (like~\Cref{thm:dir-tal-uncolored}) 
are strong enough to prove this regularity, but the robust version can ``weed out'' high-degree vertices via a definition of a suitable coloring function $\chi$.
In short, the robust version of the Talagrand-style isoperimetric theorem is much more expressive. Indeed, these style of robust results
have found other applications in distribution testing~\cite{CaChGa+21} as well.

\paragraph{The dependence on $n$.} Given \Thm{mono-testing}, it is natural to ask whether the dependence
on $n$ is necessary. Previous \emph{domain reduction} theorems have
shown that one can reduce $n$ to $\poly(d)$ in a black box manner~\cite{BlackCS20,HY22}. The monotonicity tester
based on the directed Margulis inequality for hypergrids has a logarithmic dependence on $n$~\cite{BlackCS18}.
Combining with domain reduction, we get a $\otilde(\poly(\eps^{-1}) d^{5/6})$-query tester.
It is an outstanding open problem to remove the dependence on $n$ from \Thm{mono-testing}.
In \Cref{sec:no-n}, we outline an approach to do so using the directed Talagrand inequality of \Thm{dir-tal}.

\subsection{Challenges} \label{sec:challenges}

We explain the challenges faced in proving \Thm{dir-tal} and \Thm{mono-testing}. The KMS proof of the directed Talagrand inequality for the hypercube is a
tour-de-force~\cite{KMS15}, and there are many parts of their proof that do not generalize for $n > 2$. 
We begin by giving an overview of the KMS proof for the hypercube case. 

For the time being, let us focus on the uncolored case. 
For convenience, let $T(f) = \EX_\bx[\sqrt{I^-_f(\bx)}]$ denote the hypercube directed Talagrand objective for a  $f:\hyp{d} \to \hyp{}$.
To lower bound $T(f)$,~\cite{KMS15} transform the function $f$ to a function $g$ using a sequence of what they call {\em split} operators.
The $i$th split operator applied to $f$ replaces the $i$th coordinate/dimension by two new coordinates $(i,+)$ and $(i,-)$. 
One way to think of the split operator is that takes the $\left((0,\bx_{-i}), (1, \bx_{-i})\right)$ edge and converts it into a square.
(Here, $\bx_{-i}$ denotes the collection of coordinates in $\bx$ skipping $\bx_i$.) The ``bottom" and ``top" corners
of the square store the original values of the edge, while the ``diagonal" corners store the min and max values (of the edge).
The definition of this remarkably ingenious operator ensures that the split function
is monotone in $(i,+)$ and anti-monotone in $(i,-)$. The final function $g:\{0,1\}^{2d} \to \{0,1\}$ obtained by splitting
on all coordinates has the property that it is either monotone or anti-monotone on all coordinates. That is, $g$ is unate (or pure, as~\cite{KMS15} call them), 
and for such functions the directed Talagrand inequality can be proved via a short reduction to the undirected case.

The utility of the split operator comes from the main technical contribution of~\cite{KMS15} (Section 3.4),
where it is shown that splitting cannot increase the directed Talagrand objective. This is a ``roll-your-sleeve-and-calculate'' argument that follows a case-by-case analysis. 
So, we can lower bound $T(f) \geq T(g)$. Since $g$ is unate, one can prove $T(g) = \Omega(\eps_g)$ (the distance of $g$ to monotonicity).
But how does one handle $\eps_g$, or $g$ more generally?
This is done by relating splitting to the classic {\em switch operator} in monotonicity testing, introduced
in \cite{GGLRS00}.
The switch operator for the $i$th coordinate can be thought of as modifying the edges along the $i$-dimension:
for any $i$-edge violation $(\bx,\by)$, this operator switches the values, thereby fixing the violation.
The switching operator has the remarkable property of never increasing monotonicity violations in other dimensions;
hence, switching in all dimensions leads to a monotone function.
\cite{KMS15} observe that the function $g$ basically ``embeds" disjoint variations of $f$, wherein
each variation is obtained by performing a distinct sequence of switches on $f$. 
The function $g$ contains all possible such variations of $f$, stored
cleverly so that $g$ is unate. One can then use properties of the switch
operators to relate $\eps_g$ to $\eps_f$. (The truth is more complicated; we will come 
back to this point later.)\smallskip

\noindent
{\bf Challenge \#1, splitting on hypergrids?} The biggest challenge in trying to generalize the~\cite{KMS15} argument
is to generalize the split operator. One natural 
starting point would be to consider the \emph{sort} operator, defined in~\cite{DGLRRS99}, which generalizes the switch operator:
the sort operator in the $i$th coordinate sorts the function along all $i$-lines.
But it is not at all clear how to split the $i$th coordinate into a set of coordinates that
contains the information about the sort operator thereby leading to a pure/unate function.
In short, sorting is a much more complicated operation than switching, and it is not clear how to succinctly encode this information 
using a single operator. 
%
%
%

We address this challenge by a reorientation of the KMS proof. Instead of looking at operators on dimensions to understand effects of switching/sorting, 
we do this via what we call ``tracker functions'' which are $n^d$ different Boolean functions tracking the changes in $f$. We discuss this more in~\Cref{sec:main-ideas}. \smallskip

\noindent
{\bf Challenge \#2, the case analysis for decreasing Talagrand objective.} As mentioned earlier,
the central calculation of KMS is in showing that splitting does not increase the directed
Talagrand objective. This is related (not quite, but close enough) to showing
that the switch operator does not increase the Talagrand objective. 
A statement like this is proven in KMS by case analysis; there are $4$ cases, for the possible values
a Boolean function takes on an edge. One immediately sees that such an approach
cannot scale for general $n$, since the number of possible Boolean
functions on a line is $2^n$. Even with our new idea of tracking functions, 
we cannot escape this complexity of arguing how the Talagrand-style objective decreases upon a sorting 
operation, and a case-by-case analysis depending on the values of the function is infeasible.

We address this challenge by a connection to the theory of majorization. We show
that the sort operator is (roughly) a majorizing operator on the vector of influences.
The concavity of the square root function implies that sorting along lines cannot increase
the Talagrand objective. More details are given in the next section. \smallskip

\noindent
{\bf Challenge \#3, the colorings.} Even if we circumvented the above
issues, the robust colored Talagrand objective brings a new set of issues.
Roughly speaking, colorings decide which points ``pay" for violations of the Talagrand objective, the switching/sorting operator
move points around by changing values, and the high-level argument to prove $T(f)$ drops is showing that these violations ``pay'' for the moves.
In the hypercube, a switch either changes the values on all the points
of the edge or none of the points, and this binary nature makes the handling of colors in the KMS proof fairly 
easy, merely introducing a few extra cases in their argument. 
Sorting, on the other hand, can change an arbitrary set of points, and in particular,
even in the case of $n=3$, a point participating in a violation may not change value in a sort.

%
%

To address this challenge, as we apply the sort operators to obtain a handle on our function, 
we also need to {\em recolor} the edges such that we obtain the drop in the $T$-objective. 
Once again, the theory of majorization is the guide. This part of the proof is perhaps the most technical portion of our paper. \smallskip

\noindent
{\bf Other minor challenges: the telescoping argument and tester analysis:} The issues
detailed here are not really conceptual challenges, but they do require some work
to handle the richer hypergrid domain.

Recall that the KMS analysis proves the chain of inequalities, $T(f) \geq T(g) = \Omega(\eps_g)$.
Unfortunately, it can happen that $\eps_g \ll \eps_f$. In this case, KMS observe
that one could redo the entire argument on random restrictions of $f$ to half the coordinates.
If the corresponding $\eps_g$ is still too small, then one restricts on one-fourth of the coordinates, 
so on and so forth. One can prove that somewhere along these $\log d$ restrictions, one must have $\eps_g = \Omega(\eps_f)$.
Pallavoor, Raskhodnikova, and Waingarten~\cite{PRW22} improve this analysis to remove a $\log d$ loss from the final bound.
We face the same problems in our analysis, and have to adapt the analysis to our setting. 

Finally, the tester analysis of KMS for the hypercube can be ported to the hypergrid path tester,
with some suitable adaptations of their argument. It is convenient to think of the \emph{fully augmented hypergrid},
where all pairs that lie along a line are connected by an edge. We can essentially view the hypergrid tester
as sampling a random hypercube from the fully augmented hypergrid, and then performing
a directed random walk on this hypercube. We can then piggyback on various tools from KMS for the hypercube tester,
to bound the rejection probability of the path tester for hypergrids.

\subsection{Main Ideas}\label{sec:main-ideas}
We sketch some key ideas needed to prove~\Cref{thm:dir-tal} and address the challenges detailed earlier. 
We begin with a key conceptual contribution of this paper. Given a function $f:[n]^d \to \hyp{}$, we define a collection of Boolean functions on the hypercube
called {\em tracker functions}. 
We will lower bound the directed Talagrand objective on the hypergrid by the undirected Talagrand objective on these tracker functions.
Indeed, the inspiration of these tracker functions arose out of understanding the analysis in~\cite{KMS15}, in particular, the intermediate ``$g$'' function in their Section $4$.
As an homage, we also denote our tracker functions with the same Roman letter, even though it is different from their function.

\subsubsection{Tracker functions $g_\bx$ for all $\bx \in [n]^d$}

Let us begin with the sort operator discussed earlier.
Without loss of generality, fix the ordering of coordinates in $[d]$ to be $(1,2,\ldots,d)$.
The operator $\sort_i$ for $i\in [d]$ sorts the function on every $i$-line. Given a subset $S\subseteq [d]$ of coordinates,
the function $(S\circ f)$ is obtained by sorting $f$ on the coordinates in $S$ in that order. 

Sorting along any dimension 
cannot increase the number of violations along any other dimension, and therefore upon sorting on all dimensions, the result is a monotone function~\cite{DGLRRS99}.
Suppose $f$ is $\eps$-far from monotone.
Clearly, the total number of points changed by sorting along all dimensions must be at least $\eps n^d$.
While this is not obvious here, it will be useful to 
to {\em track} how the function value changes when we sort along a 
certain subset $S$ of coordinates. The intuitive idea is: if the function value changes for most such partial sortings, then perhaps the function is far from being monotone.
To this end, for every point $\bx \in [n]^d$, we define a Boolean function $g_\bx : 2^{[d]} \to \{0,1\}$ that tracks how the function value $f$ changes 
as we apply the sort operator a subset $S$ of the coordinates. It is best to think of the domain of $g_\bx$ as subsets $S\subseteq [d]$.

\begin{definition}[Tracker Functions $g_\bx$]\label{def:gx}
	Fix an $\bx\in [n]^d$. The tracker function $g_\bx : \{0,1\}^d \to \{0,1\}$ is defined as
	\[
	\forall S\subseteq [d], ~~~~ g_\bx(S) := \left(S\circ f\right) (\bx)
	\]
\end{definition}
\noindent
We provide an illustration of this definition in~\Cref{fig:tracker-illus}. 
\begin{figure}[ht!]
	\begin{center}
		\includegraphics[trim = 200 150 200 10, clip, scale=0.5]{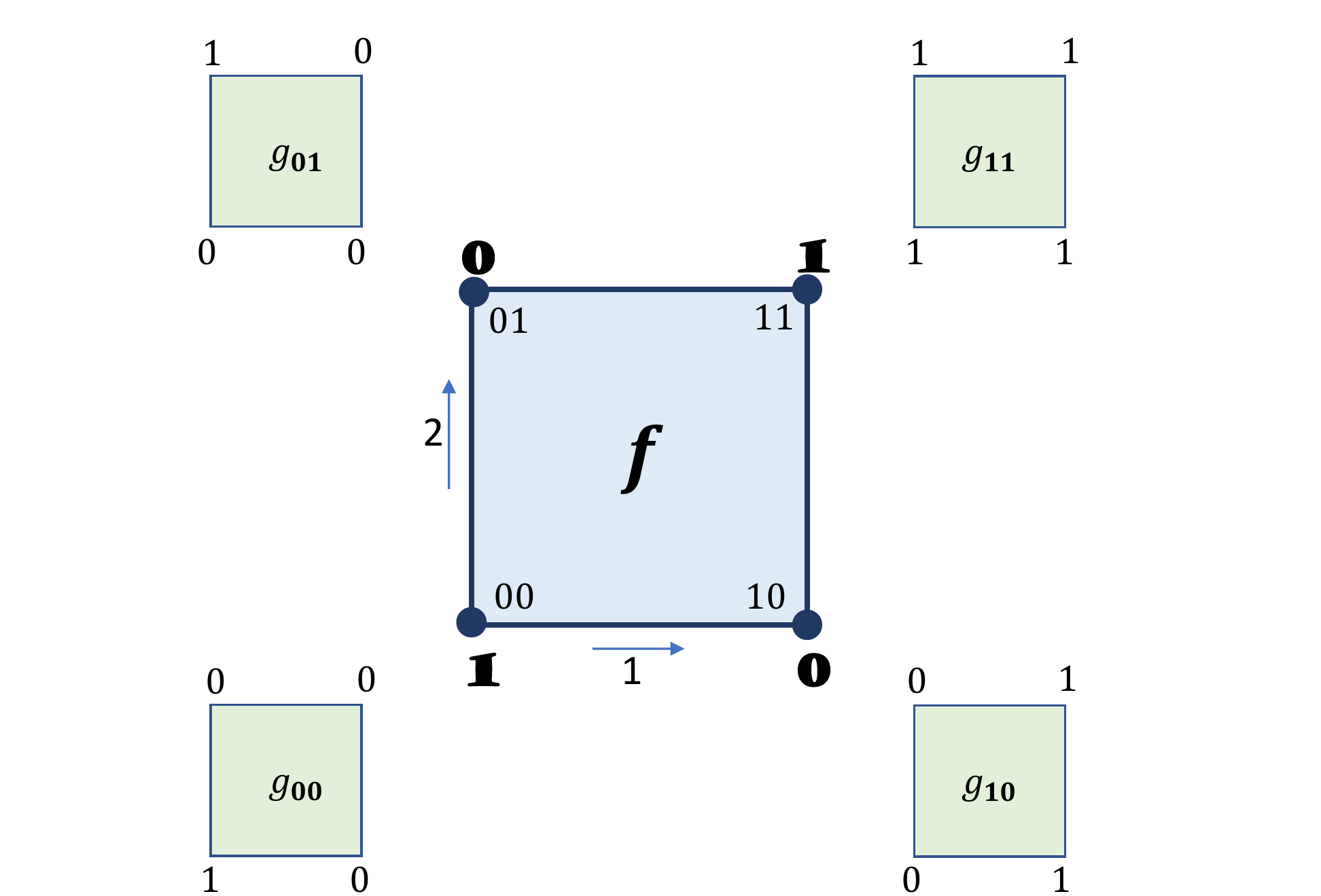}
	\end{center}
	\caption{\em The blue function $f:[n]^d\to\{0,1\}$ is defined in the middle using bold, gothic characters. We have $d=2$ and $n=2$.
		For each of the $4$ points of this square, we have four different $g_\bx:\{0,1\}^2 \to \{0,1\}$ and they are described in the 
		four green squares.
		For any $S\subseteq \{1,2\}$, if we focus on the corresponding corners of the four squares, then we get the function $(S\circ f)$.
		For instance, if $S = \{2\}$, then if we focus on the top left corners, then starting from $g_{00}$ and moving clockwise we get $(0, 1, 1, 0)$.
		These will precisely the function $f$ (read clockwise from $00$) after we sort along dimension $2$. 
	}\label{fig:tracker-illus}	
\end{figure}

\noindent
Note that when $f$ is a monotone function, all the functions $g_\bx$ are constants.
Sorting does not change any values, so $g_\bx(S)$ is always $f(\bx)$.
On the other hand, if $f$ is not monotone along dimension $i$, then there are points such that $g_\bx(\{i\}) \neq f(\bx)$. 
Indeed, one would expect the typical variance of these $g_\bx$ functions
to be related to the distance to monotonicity of $f$ (technically not true, but we come to this point later).

The tracker functions help us lower bound the (colorful) Talagrand objective for thresholded influence, in particular, the LHS in~\Cref{thm:dir-tal}. 
Recall that the Talagrand objective is the expected square root of the colorful thresholded influences on the hypergrid function $f$.
We lower bound this quantity by the expected Talagrand objective on the {\em undirected} (colorful, however) influence of the various $g_\bx$ functions.
Note that $g_\bx$ functions are defined on hypercubes.
So we reduce the robust directed Talagrand inequality on hypergrids to
robust undirected Talagrand inequalities on hypercubes.
This is the main technical contribution of our paper. Let us define the (colored) influences of these $g_\bx$ functions.

\begin{definition}[Influence of the Tracking Functions]
	Fix a $\bx \in [n]^d$ and consider the tracking function $g_\bx : \{0,1\}^d \to \{0,1\}$.
	Fix a coordinate $j\in [d]$. The influence of $g_\bx$ at a subset $S$ along the $j$th coordinate is defined as
	\[
	\Inf^{= j}_{g_\bx}(S) = 1 ~~\textrm{iff}~~g_\bx(S) \neq g_\bx(S\oplus j)~~~~\textrm{that is}~~~ (S\circ f)(\bx) \neq (S\oplus j~~\circ f)(\bx)
	\]
\end{definition}
\noindent
In plain English, the influence of the $j$th coordinate at a subset $S$ is $1$ if the function value (the hypergrid function) changes when we include the dimension $j$
to be sorted.
Once again, note that the same sensitive edge $(S, S\oplus j)$ is contributing towards both $\Inf^{=j}_{g_\bx}(S)$ and $\Inf^{=j}_{g_\bx}(S\oplus j)$.
We define a robust, colored version of these influences.

\begin{definition}[Colorful Influence of the Tracking Functions]
	Fix a $\bx \in [n]^d$ and consider the tracking function $g_\bx : \{0,1\}^d \to \{0,1\}$. 
	Fix any arbitrary coloring $\xi_\bx : E(2^{[d]}) \to \{0,1\}$ of the Boolean {\em hypercube}. 
	Fix a coordinate $j\in [d]$. The influence of $g_\bx$ at a subset $S$ along the $j$th coordinate is defined as
	\[
	\Inf^{= j}_{g_\bx, \xi_\bx}(S) = 1 ~~\textrm{iff}~~g_\bx(S) \neq g_\bx(S\oplus j)~~~~\textbf{and}~~~ g_\bx(S) = \xi_\bx(S, S\oplus j)
	\]
	The colorful total influence at the point $S$ in $g_\bx$ is defined as 
	\begin{equation}\label{eq:def-colorful-total-infl-gx}
		I_{g_\bx, \xi_\bx}(S) := \sum_{j=1}^d I^{=j}_{g_\bx, \xi_\bx} (S) \notag
	\end{equation}
\end{definition}
\noindent
As before, for a sensitive edge $(S, S\oplus j)$ of $g_\bx$, we count it towards the influence of the endpoint whose value equals the color $\xi_\bx(S, S\oplus j)$. 
The main technical contribution of this paper is proving that for any function $f:[n]^d \to \{0,1\}$ and any arbitrary coloring $\chi: E\to \{0,1\}$ of the 
hypergrid edges, for every $\bx\in [n]^d$ there \underline{exists} a coloring $\xi_\bx:E(2^{[d]}) \to \{0,1\}$ of the Boolean {\em hypercube} edges, such that 
\begin{equation}\label{eq:hope2}
T_{\Phi_\chi}(f) := \Exp_{\bx\in [n]^d}~\left[\sqrt{\Phi_{f,\chi}(\bx)}  \right] ~~\gtrapprox~~ \Exp_{\bx\in [n]^d} \Exp_{S\subseteq [d]}~~[\sqrt{I_{g_\bx, \xi_\bx}(S)}] \tag{H1}
\end{equation}
We explain the $\approx$ in the above inequality in the next subsection.
\smallskip

Why is a statement like~\eqref{eq:hope2} useful? Because the RHS terms are Talagrand objectives on colored influences on the usual undirected hypercube.
Therefore, we can apply undirected Talagrand bounds (known from KMS, \Cref{thm:kms-und}) to get an upper bound on the variance.
\begin{restatable}[Corollary of Theorem 1.8 in~\cite{KMS15}]{corollary}{corkms}
	\label{cor:kms}
	Fix $f:[n]^d \to \{0,1\}$. Fix an $\bx \in [n]^d$ and consider the tracking function $g_\bx : \{0,1\}^d \to \{0,1\}$. 
	Consider any {\em arbitrary} coloring $\xi_\bx : E(2^{[d]}) \to \{0,1\}$ of the Boolean {\em hypercube}. Then, for every $\bx \in [n]^d$, we have
	\[
	\Exp_{S\subseteq [d]}~~[\sqrt{I_{g_\bx, \xi_\bx}(S)}] = \Omega(\var(g_\bx))
	\]
\end{restatable}
\noindent
The final piece of the puzzle connects $\var(g_\bx)$'s with the distance to monotonicity. Ideally, we would have liked to have a statement such as the following true.
\begin{equation}\label{eq:hope1}
	\Exp_{\bx\in [n]^d} \left[\var(g_\bx)\right] \approx \Omega(\eps_f) \tag{H2} 
\end{equation}
We now see that \eqref{eq:hope2}, \Cref{cor:kms}, and \eqref{eq:hope1} together implies~\Cref{thm:dir-tal} (indeed without the $\log n$). 

\subsubsection{High level description of our approaches}

\paragraph{Addressing the $\approx$ in \eqref{eq:hope2} via semisorting.} As stated, we do not know if \eqref{eq:hope2} is true. However, we establish \eqref{eq:hope2} for 
{\em semisorted} functions $f:[n]^d \to \{0,1\}$. A function $f$ is semisorted if on any line $\ell$, the restriction of the function on the first half is sorted and the restriction 
on the second half is sorted. This may seem like a simple subclass of functions, but note that all functions on the Boolean hypercube ($n=2$)
are vacuously semisorted. Thus, proving \Cref{thm:dir-tal} on semi-sorted functions is already a generalization of the~\cite{KMS15} result. \Cref{thm:semisorted-reduce-to-g} is the formal 
restatement of \eqref{eq:hope2}.

We reduce \Thm{dir-tal} on general functions to the same bound for semisorted functions.
Consider semisorting $f$, which means we sort $f$ on each half of every line. 
Suppose the Talagrand objective did not increase \emph{and} the distance to monotonicity
did not decrease. Then \Thm{dir-tal} on the semisorted version of $f$ implies
\Thm{dir-tal} on $f$. 
What we can prove is that: given the semisorted function, one can find a {\em recoloring} of the hypergrid edges such that the Talagrand objective doesn't increase. The precise statement is given in~\Cref{lem:semisorting-decreases}. We comment on our techniques to prove such a statement in a later paragraph.

Although semisorting can't increase the Talagrand objective, it can clearly reduce the distance to monotonicity. However, a relatively simple inductive
argument proves \Thm{dir-tal} with a $\log n$ loss.
Any function can be turned into a completely sorted (aka monotone) function by performing ``$\log n$ semisorting steps'' at varying scales.
In each scale, we consider many disjoint small hypergrids, and 
convert a semisorted function defined over a small hypergrid to another semisorted function over a hypergrid of double the size (the next scale).
In one of these scales, we will find a semisorted function that has $\Omega(\eps/\log n)$ distance from its sorted version.
One can average \Thm{dir-tal} over all the small hypergrids at this scale to bound the Talagrand objective of the whole function by $\Omega(\eps/\log n)$. This is the step where we incur the $\log n$-factor loss.
This argument is not complicated, and we provide illustrated details in~\Cref{sec:semisorted}.

The real work happens in proving \Cref{thm:dir-tal-semisorted}, that is,~\eqref{eq:hope2} for semisorted functions.

\paragraph{Approach to proving \eqref{eq:hope2} for semisorted functions.} 
Recall, we have a fixed adversarial coloring $\chi:E \to \{0,1\}$. 
The proof follows a ``hybrid argument'' where we define a potential that is modified over $d+1$ rounds. 
At the beginning of round $0$ it takes the value $\Exp_{\bx \in [n]^d}[\sqrt{\Phi_{f,\chi}(\bx)}]$ which is the LHS of~\eqref{eq:hope2}. At the end of round $d$ it takes the value $\Exp_{\bx \in [n]^d}\Exp_{S\subseteq [d]}[\sqrt{I_{g_\bx, \xi_\bx}(S)}]$ which is the RHS of \eqref{eq:hope2}.
The proof follows by showing that the potential decreases in each round.

Let us describe the potential. Let us first write this without any reference to the colorings (so no $\chi$'s and $\xi_\bx$'s), and then subsequently address the colorings.
At stage $i$, fix a subset $S \subseteq [i]$. Define
\begin{equation}\label{eq:hybrid}
R_{i}(S) := \Exp_{\bx \in [n]^d} \left[\sqrt{~\sum_{j=1}^{i} I^{=j}_{g_\bx} (S) ~~+~~ \sum_{j=i+1}^d \Phi_{S\circ f}(\bx; j) }~\right] \tag{Hybrid}
\end{equation}
We remind the reader that $S\circ f$ is the function $f$ after the dimensions corresponding to $i\in S$ have been sorted. 
Thus, $R_i(S)$ is a ``hybrid" Talagrand objective, with two different kinds of influences being summed.
Consider point $\bx \in [n]^d$. On the first $i$ coordinates,
we sum the undirected influence (along these coordinates) of $S$ on the function $g_\bx$. On the coordinates $i+1$ to $d$, 
we sum to directed influence along these coordinates in the function $S \circ f$. 
The potential is $\Lambda_i := \Exp_{S\subseteq [i]} [R_i(S)]$. 

To make some sense of this, consider the extreme cases of $i=0$ and $i=d$. When $i=0$,
we only have the second $\Phi_{S \circ f}$ term. Furthermore, $S$ is empty since $S \subseteq [i]$.
So $\Lambda_0$ is precisely the original directed Talagrand objective, the LHS of \eqref{eq:hope2}.
When $i = d$, we only have the $I^{=j}_{g_\bx}$ terms. Taking expectation
over $S \subseteq [d]$ to get $\Lambda_d$, we deduce that $\Lambda_d$ is the RHS of \eqref{eq:hope2}.

%
%

We will prove $\Lambda_{i-1} \geq \Lambda_i$ for all $1\leq i\leq d$.
To choose a uar set in $[i]$, we can choose a uar subset of $[i-1]$
and then add $i$ with $1/2$ probability.
Hence, $\Lambda_i = (\EX_{S \subseteq [i-1]} [R_i(S) + R_i(S+i)])/2$,
while $\Lambda_{i-1} = \EX_{S \subseteq [i-1]} [R_{i-1}(S)]$.
So, if we prove that $R_{i-1}(S)$ is at least both $R_i(S)$ and $R_i(S+i)$, then $\Lambda_{i-1} \geq \Lambda_i$. 
The bulk of the technical work in this paper is involved in proving these two inequalities, so let us spend a little time explaining what proving this entails.

Let's take the inequality $R_{i-1}(S) \geq R_i(S)$. Refer again to \eqref{eq:hybrid}. When we go from $R_{i-1}(S)$ to $R_i(S)$, 
under the square root, the term $\Phi_{S\circ f}(\bx;i)$ is replaced by $I^{=i}_{g_\bx} (S)$. 
To remind the reader, the former term is the indicator of whether $\bx$ participates in a $i$-violation after the coordinates in $S \subseteq [i-1]$ have been sorted. 
The latter term is whether $g_\bx(S+i)$ equals $g_\bx(S)$, that is, whether the (hypergrid) function value at $\bx$ changes between sorting on coordinates in $S$ and $S+i$.
Just by parsing the definitions, one can observe that $\Phi_{S \circ f}(\bx;i) \geq I^{=i}_{g_\bx}(S)$; if a point is modified on sorting in the $i$-coordinate, 
then it must be participating in some $i$-violation (note that vice-versa may not be true and thus we have an inequality and not an equality). The quantity under the square-root {\em point-wise} dominates (ie, for every $\bx$) when we move from $R_{i-1}(S)$ to $R_i(S)$.  Thus, $R_{i-1}(S) \geq R_i(S)$.

The other inequality $R_{i-1}(S) \geq R_i(S+i)$, however, is much trickier to establish.  
In $R_i(S+i)$, the second summation under the square-root, the $\Phi$ terms, are actually on a {\em different} function. 
The $\Phi_{S\circ f}(\bx; j)$ terms in $R_{i-1}(S)$ are the thresholded influences of the function after sorting on coordinates in $S$.
But in $R_i(S+i)$, these terms are $\Phi_{(S+i)\circ f} (\bx;j)$, the thresholded influences of $\bx$ for the function after sorting on $S+i$.
Although, it is true that sorting on more coordinates cannot increase the total number of violations along any dimension, 
this fact is {\em not} true point-wise.  So, a point-wise argument as in the previous inequality is not possible. 

The argument for this inequality proceeds {\em line-by-line}. One fixes an $i$-line $\ell$ and considers the vector of ``hybrid function'' values on this line. 
We then consider this vector when moving from $R_{i-1}(S)$ to $R_i(S+i)$, and we need to show that the {\em sum of square roots} can get only smaller. 
This is where one of our key insights comes in: the theory of majorization can be used to assert these bounds.
Roughly speaking, a vector $\ba$ (weakly) majorizes a vector $\bb$ if the sum of the $k$-largest coordinates of $\ba$ 
dominates the sum of the $k$-largest coordinates of $\bb$, for every $k$. 
A less balanced vector majorizes a more balanced vector.
If the $\ell_1$-norms of these vectors are the same, then the sum of square roots 
of the entries of $\ba$ is at most the sum of square roots of that of $\bb$. 
This follows from concavity of the square-root function.

Our overarching mantra throughout this paper is this: whenever we perform an operation and the hybrid-influence-vector induced by a line changes, the new vector majorizes the old vector.
Specifically, these vectors are generated by look at the terms of $R_{i-1}(S)$ and $R_i(S+i)$
restricted to $i$-lines.

To prove this vector-after-operation majorizes vector-before-operation, we need some structural assumptions on the function. Otherwise, it's not hard to construct examples where this just fails.
The structure we need is precisely the {\em semisortedness} of $f$. When a function is semisorted, 
the majorization argument goes through. At a high level, when $f$ is semisorted, 
the vector of influences (along a line) satisfy various monotonicity properties.
In particular, when we (fully) sort on some coordinate $i$, we can show
the points losing violations had low violations to begin with. That is, 
the vector of violations becomes less balanced, and the majorization follows.

The above discussion disregarded the colors. With colors, the situation is noticeably more difficult. Although the function $f$ is assumed to be semisorted, the coloring $\chi:E \to \{0,1\}$ is adversarial. So even though the vector
of influences may have monotonicity properties, the colored influences may not have this structure.
So a point with high influence could have much lower colored influence. Note that the sort operator
is insensitive to the coloring. So the majorization argument discussed above might not hold when 
looking at colored influences. 

With colors, \eqref{eq:hybrid} is replaced by the actual quantity~\eqref{eq:rhs-quantity} described in~\Cref{sec:mainworkhorse}. 
To carry out the majorization argument, we need to construct a family of colorings $\xi_\bx$ on the $n^d$
different hypercubes. We also need 
$2^d$ many different auxiliary colorings $\chi_S$ of the hypergrid, constructed after every sort operation.
The argument is highly technical. But all colorings are
chosen to follow our mantra: vector after operation should majorize vector before operation. 
The same principle is also used to prove~\Cref{lem:semisorting-decreases} which claims that semisorting an interval can only decrease the Talagrand objective, after a recoloring. 

The details of the actual $R_i(S)$ hybrid function and the strategy to use them is presented in~\Cref{sec:mainworkhorse}. The most technical part of the paper 
is in~\Cref{sec:proofoflemma6}, which proves that 
the potential decreases in each round. 

\paragraph{Addressing the $\approx$ in \eqref{eq:hope1} via random sorts.} To finally complete the argument, we need \eqref{eq:hope1} 
that relates the average variance of the $g_\bx$ functions to the distance to monotonicity of $f$.
As discussed earlier, \eqref{eq:hope1} is false, even for the case of hypercubes. 
Nevertheless, one can use \eqref{eq:hope2} and \Cref{cor:kms} to prove a lower bound on $T_{\Phi_{\chi}}(f)$ with respect to $\eps_f$. 
This is the telescoping argument of KMS, refined in~\cite{PRW22}.  We describe the main ideas below.
The first observation (see~\Cref{thm:semisorted-reduce-conv}) is that $	\Exp_{\bx\in [n]^d} \left[\var(g_\bx)\right]$ is roughly 
$\EX_S[\Delta(S \circ f, \overline{S} \circ f]$
where $S$ is a uniform random subset of coordinates.
The distance to monotonicity $\eps_f$ is approximated by $\Delta\left(f, S\circ \overline{S}\circ f\right)$ which, by the triangle inequality, is at most 
$\Delta(f, S\circ f) + \Delta(S\circ f, \overline{S}\circ f)$. Thus, we get a relation between $\eps_f$, the expected $\var(g_\bx)$, and the distance between $f$ and a ``random sort'' of $f$.
Therefore, if \eqref{eq:hope1} is not true, then a random sort of $f$ must be still far from being monotone, and then one can repeat the whole argument on just this random sort itself.
In one of these $\log d$ ``repetitions'', the \eqref{eq:hope1} must be true since in the end we get a monotone function (which can't be far from being monotone). 
And this suffices to establish \Cref{thm:dir-tal}. We re-assert that the main ideas are already present in~\cite{KMS15,PRW22}. However, we require a more general presentation
to make things work for hypergrids. These details can be found in~\Cref{sec:semisorted-tal-dist}.

\subsection{Related Work} \label{sec:related}

Monotonicity testing has seen much activity since its introduction around 25 years ago~\cite{Ras99,EKK+00,GGLRS00,DGLRRS99,LR01,FLNRRS02,HK03,AC04,HK04,ACCL04,E04,SS08,Bha08,BCG+10,FR,BBM11,RRSW11,BGJ+12,ChSe13,ChSe13-j,ChenST14,BeRaYa14,BlRY14,ChenDST15,ChDi+15,KMS15,BeBl16,Chen17,BlackCS18,BlackCS20,BKR20,HY22}.

We have already covered much of the previous work on Boolean monotonicity testing
over the hypercube, but give a short recap. For convenience of presentation, in some results, we subsume $\eps$-dependencies
using the notation $O_\eps$.
The problem was introduced by Goldreich et al.~\cite{GGLRS00} and Raskhodnikova \cite{Ras99}, who described an $O(d/\eps)$-query tester.
Chakrabarty and Seshadhri~\cite{ChSe13-j} achieved the first sublinear in dimension query complexity of $\otilde_\eps(d^{7/8})$ using directed isoperimetric inequalities.
Chen, Servedio, and Tan~\cite{ChenST14} improved the analysis to $\otilde_\eps(d^{5/6})$ queries. Fischer et al.~\cite{FLNRRS02} had first shown an $\Omega(\sqrt{d})$-query
lower bound for non-adaptive, one-sided testers, by a short and neat construction.
The non-adaptive, two-sided $\widetilde{\Omega}(\sqrt{d})$ lower bound is much harder to attain, and was done by Chen, Waingarten, and Xie~\cite{Chen17},
improving on the $\Omega(d^{1/2-c})$ bound from~\cite{ChenDST15}, which itself improved on the $\widetilde{\Omega}(d^{1/5})$ bound of \cite{ChenST14}. ~\cite{KMS15} gave an $\otilde_\eps(\sqrt{d})$-query tester, via the
robust directed Talagrand inequality.

While this resolves the non-adaptive testing complexity (up to $\poly(\eps^{-1}\log d)$ factors) for the hypercube,
the adaptive complexity is still open. 
The first polynomial lower bound of $\widetilde{\Omega}(d^{1/4})$ for adaptive testers was given
by Belovs and Blais~\cite{BeBl16} and has since been improved to $\widetilde{\Omega}(d^{1/3})$ by Chen, Waingarten, and Xie~\cite{Chen17}.
Chakrabarty and Seshadhri~\cite{ChSe19} gave an adaptive $\otilde_\eps(I_f)$-query tester, thereby showing that adaptivity
can help in monotonicity testing. The $d^{1/3}$ vs $\sqrt{d}$ query complexity gap is an outstanding open question
in property testing.

There has been work on approximating the distance to monotonicity in $\poly(d,\eps_f)$-queries.
Fattal and Ron~\cite{FR} gave the first non-trivial result of an $O(d)$-approximation, and Pallavoor, Raskhodnikova, and Waingarten~\cite{PRW22}
gave a non-adaptive $O(\sqrt{d})$-approximation (all running in $\poly(d,\eps_f)$ time). They also
show that non-adaptive $\poly(d)$-time algorithms cannot beat this approximation factor.

The above discussion is only for Boolean valued functions on the hypercube. For arbitrary ranges,
the original results on monotonicity testing gave an $O(d^2/\eps)$-query tester~\cite{GGLRS00,DGLRRS99}. 
Chakrabarty and Seshadhri~\cite{ChSe13} proved that $O(d/\eps)$-queries suffices for monotonicity testing,
matching the lower bound of $\Omega(d)$ of Blais, Brody, and Matulef~\cite{BBM11}. The latter bound
holds even when the range size is $\sqrt{d}$.
A recent result of Black, Kalemaj, and Raskhodnikova showed a smooth trade-off between the $\sqrt{d}$ bound
for the Boolean range and the $d$ bound for arbitrary ranges (\cite{BKR20}). Consider functions $f:\hyp{d} \to [r]$.
They gave a tester with query complexity $\otilde_\eps(r\sqrt{d})$, achieved by extending
the directed Talagrand inequality to arbitrary range functions. Their techniques are quite black-box
and carry over to other posets. We note that their techniques
can also be ported to our setting, so we can get an $\otilde_\eps(rn\sqrt{d})$-query monotonicity
tester for functions $f:[n]^d \to [r]$.

We now discuss monotonicity testing on the hypergrid. We discuss more about the $\eps$-dependencies,
since there have been interesting relevant discoveries. As mentioned above, \cite{DGLRRS99}
gives a non-adaptive, one-sided $O((d/\eps)\log^2(d/\eps))$-query tester. This was improved to $O((d/\eps)\log(d/\eps))$ by Berman, Raskhodnikova, and Yaroslavtsev~\cite{BeRaYa14}.
This paper also showed an interesting adaptivity gap for 2D functions $f:[n]^2 \to \hyp{}$:
there exists an $O(1/\eps)$-query adaptive tester (in fact, for any constant dimension $d$), and they show 
an $\Omega(\log(1/\eps)/\eps)$ lower bound for non-adaptive testers.
Previous work~\cite{BlackCS18} by the authors gave an $\otilde_\eps(d^{5/6}\log n)$-query tester, by proving
a directed Margulis inequality on augmented hypergrids. Another work~\cite{BlackCS20}
of the authors, and subsequently a work~\cite{HY22} by Harms and Yoshida, designed domain reduction methods for monotonicity testing, showing how $n$ can be reduced to $\poly(\eps^{-1},d)$ by subsampling
the hypergrid.

For hypergrid functions with arbitrary ranges, the optimal complexity is known to be $\Theta(d\log n)$~\cite{ChSe13,ChSe14}.
When the range is $[r]$ and $d=1$, one can get $O(\log r)$-query testers~\cite{PaRaVa18}. 

%% file: prelims.tex
\section{Preliminaries} \label{sec:prelims}

A central construct in our proof is the \emph{sort} operator. 

\begin{definition} \label{def:sortline} Consider a Boolean function on the line $h:[n] \to \hyp{}$.  
The sort operator $\sortline{}$ is defined as follows.
$$ \sortline{h}(b) =
\begin{cases}
    0 & \textrm{if} \ b < n-\|h\|_1 \\
    1 & \textrm{if} \ b \geq n-\|h\|_1
\end{cases}$$
\end{definition}

Thus, the sort operator ``moves" the values on a line to ensure that it is sorted.
Note that $\sortline{h}$ and $h$ have exactly the same number of zero/one valued points.
We can now define the sort operator for any dimension $i$. This operator
takes a hypergrid function and applies the sort operator on every $i$-line.

\begin{definition} \label{def:sortop} Let $i$ be a dimension
and $f:[n]^d \to \hyp{}$. The sort operator for dimension $i$, $\sorti{}{i}$, is defined as follows.
For every $i$-line $\ell$, $\sorti{f}{i}|_\ell = \sortline{f|_\ell}$.

Let $S$ be an ordered list of dimensions, denoted $(i_1, i_2, \ldots, i_k)$. 
The function $S \circ f$
is obtained by applying the $\sorti{}{i}$ operator in the order given by $S$. Namely,
$$ S \circ f = \sorti{\sorti{\ldots \sorti{f}{{i_1}}}{{i_{k-1}}}}{{i_k}} $$
\end{definition}

Somewhat abusing notation, we will treat the ordered list of dimensions $S$
as a set, with respect to containing elements. The key property of the sort
operator is that it preserves the sortedness of \emph{other} dimensions. 

\begin{claim} \label{clm:sortS} The function $S \circ f$ is monotone
along all dimensions in $S$.
\end{claim}

\begin{proof} We will prove the following statement: if $f$
is monotone along dimension $i$, then $\sorti{f}{j}$ is monotone
along both dimensions $i$ and $j$. A straightforward induction (which we omit)
proves the claim.

By construction, the function $\sorti{f}{j}$ is monotone along dimension $j$.
Consider two arbitrary points $\bx \preceq \bx'$ that are $i$-aligned (meaning
that they only differ in their $i$-coordinates).
We will prove that $\sorti{f}{j}(\bx) \leq \sorti{f}{j}(\bx')$, which will prove
that $\sorti{f}{j}$ is monotone along dimension $i$.

For convenience, let the $j$-lines containing $\bx$ and $\bx'$
be $\ell$ and $\ell'$, respectively. Note that these $j$-lines
only differ in their $i$-coordinates. Let $c$ denote the $j$-coordinate
of $\bx$ (and $\bx'$). Observe that $\sorti{f}{j}{\bx} = \sorti{f}{j}|_\ell(c)$
(analogously for $\bx'$). 

Note that, $\forall c \in [n]$, $f|_\ell(c) \leq f|_{\ell'}(c)$.
This is because $f$ is monotone along dimension $i$, and $\ell$
has a lower $i$-coordinate than that of $\ell'$.
Hence, $\|f|_{\ell}\|_1 \leq \|f|_{\ell'}\|_1$.
By the definition of the sort operator, $\forall c \in [n], \sortline{f|_{\ell}}(c) \leq \sortline{f|_{\ell'}}(c)$.
Thus, $\sorti{f}{j}|_\ell(c) \leq \sorti{f}{j}|_{\ell'}(c)$, implying
$\sorti{f}{j}(\bx) \leq \sorti{f}{j}(\bx')$.
\end{proof}

A crucial property of the sort operator is that it can never increase 
the distance between functions. This property, which was first established in~\cite{DGLRRS99} (Lemma 4), will be used in
\Cref{sec:semisorted-tal-dist}, where we apply our main isoperimetric theorem
on random restrictions.
We provide a proof for completeness.
\begin{claim} \label{clm:sort-hamm} Let $f, f': [n]^d \to \hyp{}$
be two Boolean functions. For any ordered set $S \subseteq [d]$,
$$ \Delta(S \circ f, S \circ f') \leq \Delta(f,f')$$
\end{claim}

\begin{proof} It suffices to prove this bound when $S$
is a singleton. We prove that for any $i \in [d]$, $\Delta(\sorti{f}{i}, \sorti{f'}{i}) \leq \Delta(f,f')$. In the following, we will use the simple fact that for monotone functions $h, h':[n] \to \hyp{}$, $\Delta(h,h') = \Big| \|h\|_1 - \|h'\|_1 \Big|$.
Also, we use the equality $\|\sortline{h}\|_1 = \|h\|_1$.
\begin{eqnarray*}
\Delta(\sorti{f}{i}, \sorti{f'}{i}) & = & \sum_{\ell \ \textrm{$i$-line}} \Delta(\sorti{f}{i}|_\ell,
    \sorti{f'}{i}|_\ell) = \sum_{\ell} \Big| \| \sorti{f}{i}|_\ell\|_1 - \| \sorti{f'}{i}|_{\ell}\|_1 \Big| \nonumber \\
    & = & \sum_{\ell} \Big| \|f|_\ell\|_1 - \|f'|_{\ell}\|_1  \Big| \\
    & = & \sum_{\ell} \Big| \sum_{c \in [n]} f|_\ell(c) - \sum_{c \in [n]} f|_{\ell}(c)  \Big| \\
    & \leq & \sum_{\ell} \sum_{c \in [n]} \Big|f|_\ell(c) - f'|_\ell(c) \Big| = \Delta(f,f')
\end{eqnarray*}
\end{proof}

The method of obtaining a monotone function via repeated sorting is close to being optimal.
For hypercubes, this result was established by~\cite{FR} (Lemma 4.3) and also present in~\cite{KMS15} (Lemma 3.5).
The proofs goes through word-for-word applied to hypergrids.

\begin{claim}\label{clm:2appx}
	For any function $f:[n]^d \to \{0,1\}$, 
	\[
		\eps_f \leq \Delta(f, [d]\circ f) \leq 2\eps_f
	\]
\end{claim}
\begin{proof}
	The first inequality is obvious since $[d]\circ f$ is monotone as established in~\Cref{clm:sortS}.	
	Let $h$ be the monotone function closest to $f$, that is, $\eps_f =\Delta(f, h)$.
So,
	\[
	\Delta(f, [d]\circ f) \underbrace{\leq}_{\text{triangle ineq}} \Delta(f, h) + \Delta([d]\circ f, h) 
	\underbrace{=}_{\text{since}~h = [d]\circ h} \Delta(f, h) + \underbrace{\Delta([d]\circ f, [d]\circ h)}_{\leq \Delta(f,h) ~	\text{by~\Cref{clm:sort-hamm}}}~ \leq 2\Delta(f,h) = 2\eps_f
	\]	
\end{proof}

We provide one more simple claim about the sort operator that will be used throughout \Cref{sec:proofoflemma6}. Given $h,h' \colon [n] \to \{0,1\}$, define 
\[
\Delta^-(h,h') = |\{c \in [n] \colon h(c) > h'(c)\}| \text{ and } \Delta^+(h,h') = |\{c \in [n] \colon h(c) < h'(c)\}| \text{.}
\]

\begin{claim} \label{clm:sort-violations} Let $h,h' \colon [n] \to \{0,1\}$ be any two functions. Then, $\Delta^-(\sortline{h},\sortline{h'}) \leq \Delta^-(h,h')$. \end{claim}

\begin{proof} Observe that if $\norm{h}_1 \leq \norm{h'}_1$, then $\Delta^-(\sortline{h},\sortline{h'}) = 0$ and so we are done. On the other hand if $\norm{h}_1 \geq \norm{h'}_1$, then we have
\[
\Delta^-(\sortline{h},\sortline{h'}) = \norm{h}_1 - \norm{h'}_1 = \sum_{c\in[n]} h(c) - h'(c) = \Delta^-(h,h') - \Delta^+(h,h') \leq \Delta^-(h,h') \text{.}
\]
\end{proof}


\subsection{Colorful Influences and the Talagrand Objective} \label{sec:tal-obj}

We will need undirected, colorful Talagrand inequalities for proving \Thm{dir-tal}. 
For the sake of completeness, we explicitly define the undirected colored influence.

\begin{definition} \label{def:col-inf} Consider a function $g:\hyp{d} \to \hyp{}$
and a $0$-$1$ coloring $\xi$ of the edges of the hypercube $\hyp{d}$.
The influence of $\bz \in \hyp{d}$, denoted $I_{g,\xi}(\bz)$,
is the number of sensitive edges incident to $\bz$ whose color has value $f(\bz)$.

(An edge is sensitive if both endpoints have different values.)
\end{definition}

Talagrand's theorem asserts that $\EX_\bz[\sqrt{I_g(\bz)}] = \Omega(\var(g))$~\cite{Tal93}. 
The robust/colored version proven by KMS asserts this to be true for arbitrary colored
influences.

\begin{theorem}[Paraphrasing Theorem 1.8 of~\cite{KMS15}] (Colored Talagrand Theorem on the Undirected Hypercube)\label{thm:kms-und}
There exists an absolute constant $C > 0$ such that 
for any function $g:\{0,1\}^d \to \{0,1\}$ and any $0$-$1$ coloring 
$\xi$ of the edges of the hypercube,
\[
\Exp_{\bz\in \{0,1\}^d} \left[\sqrt{\Inf_\xi(\bz)} \right] \geq C \cdot \var(g)
\]
\end{theorem}

It will be convenient in our analysis to formally define the Talagrand objective for colored, thresholded influences 
on the hypergrid.
 
\begin{definition}[Colored Thresholded Talagrand Objective]
	Given any Boolean function $f:[n]^d \to \{0,1\}$ and $\chi:E \to \{0,1\}$, 
	we define the Talagrand objective with respect to the colorful thresholded influence as
	\[
	T_{\Phi_\chi}(f) := \Exp_\bx~\left[\sqrt{\Phi_{f,\chi}(\bx)}  \right] 
	\]
	where, $\Phi_{f,\chi}$ is defined in~\Cref{def:phi-f-chi}.
\end{definition}

\subsection{Majorization} \label{sec:major}

It is convenient to think of the Talagrand objective as a ``norm'' of a vector. Throughout the paper, we (ab)use the following notation:
\[
\textrm{given a vector $\bv \in \RR^t_{\geq 0}$},~~~ \norm{\bv}_{1/2} := \sum_{i=1}^t \sqrt{\bv_i}\text{.}
\]
If we imagine an $n^d$-dimensional vector indexed by the points of the hypergrid, we see that the Talagrand objective is precisely the norm of 
the vector whose $\bx$'th entry is $\Phi_{f,\chi}(\bx)$. Most often, however, we would be considering the Talagrand objective line-by-line, with the 
natural ordering of the line defining a natural ordering on the vector. To be more precise, fix a dimension $i\in [d]$ and fix an $i$-line $\ell$. 
An $i$-line is a set of $n$ points which only differ in the $i$th coordinate. This line $\ell$ defines a vector $\vv{\Phi_\ell(f)} \in \RR_{\geq 0}^n$ 
whose $j$th coordinate, for $1\leq j\leq n$ is precisely $\Phi_{f,\chi}(\bx)$ where $\bx \in \ell$ has $\bx_i = j$. Note that
\[
\forall i\in [d],~~~	T_{\Phi_{\chi}}(f) = \frac{1}{n^d} \sum_{i\text{-lines}~\ell} \norm{\vv{\Phi_\ell(f)}}_{1/2}\text{.}
\] 
Our proof to establish (the correct version of)~\eqref{eq:hope2} proceeds via a hybrid argument that modifies the function and the coloring in various stages.
In each stage, we prove that the norm decreases. We use the following facts from the theory of majorization.

In the rest of this subsection all vectors, unless explicitly mentioned, live in $\RR^t_{\geq 0}$ for some positive integer $t$.
Given a vector $\ba$, we use $\sortdown{\ba}$ and $\sortup{\ba}$ to denote the vectors obtained by sorting $\ba$ in decreasing and increasing order, respectively. Given two  vectors 
$\ba$ and $\bb$ with the same $\ell_1$ norm,
we say $\ba \majorizes \bb$ if for all $1\leq k \leq t$, $\sum_{i\leq k} \sortdown{\ba}_i \geq \sum_{i\leq k} \sortdown{\bb}_i$. 

Throughout this paper, when we apply majorization the LHS vector would be sorted (either increasing or decreasing) while the RHS vector would be unsorted.
To be absolutely clear which is which, when $\ba$ is sorted decreasing, we use the notation $\ba \majorizes \sortdown{\bb}$ and when $\ba$ is sorted increasing we use the notation
$\ba \majorizes \sortup{\bb}$. Here is a simple standard fact that connects majorization to the Talagrand objective; 
it uses the fact that the sum of square roots is a symmetric concave function, and is thus Schur-concave.

\begin{fact}[Chapter 3,~\cite{MarshallOA11}]
	Let $\ba$ and $\bb$ be two vectors such that $\ba \majorizes \bb$. Then, $\norm{\ba}_{1/2} \leq \norm{\bb}_{1/2}$.
\end{fact}

\noindent
Next, we state and prove a simple but key lemma repeatedly used throughout the analysis.
\begin{mdframed}[backgroundcolor=blue!10,topline=false,bottomline=false,leftline=false,rightline=false]
	\begin{lemma}\label{lem:sum-of-vectors}
		Let $\bU = \sum_i \bw_i$ be a finite sum of $t$-dimensional non-negative vectors.
		Let $\bS := \sum_i \sortdown{\bw_i}$. Then, $\bS \majorizes \sortdown{\bU}$.
		Analogously, if $\bS := \sum_i \sortup{\bw_i}$, then $\bS \majorizes \sortup{\bU}$.
	\end{lemma}
\end{mdframed}
\begin{proof}
	We prove the first statement; the second analogous statement has an absolutely analogous proof.
	We begin by noting $\bS$ is a sorted decreasing vector since it is a sum of sorted decreasing vectors.
	For brevity, let's use $\bV := \sortdown{\bU}$. Next, we note that 
	$\norm{\bS}_1 = \norm{\bV}_1 = \sum_i \norm{\bw_i}_1$.
	
	Now fix a $1\leq \tau \leq t$. We need to show $\sum_{j=1}^\tau \bS_j\geq \sum_{j=1}^\tau \bV_j$. 
	Consider the $\tau$ largest coordinates of $\bU$, and let them comprise $T\subseteq [t]$ where $|T| = \tau$.
	Consider the $\tau$-dimensional vectors $\bw_i[T]$ where we restrict our attention to only these coordinates.
	Let $\bS'$ be the $\tau$-dimensional vector formed by the sum of the sorted versions $\sortdown{\bw_i[T]}$.
	Note that $\sum_{j=1}^{\tau} \bS'_j = \sum_{j=1}^\tau \bV_j$. 
	Also note that for any $1\le j\leq \tau$, the number $\bS'_j$ equals $\sum_i(\text{$j$th max of $\bw_i[T]$})$ and $\bS_j$ equals $\sum_{i} (\text{$j$th max of $\bw_i$})$.
	Thus, $\bS_j \geq \bS'_j$, proving that $\sum_{j=1}^\tau \bS_j\geq \sum_{j=1}^\tau \bV_j$. \end{proof} 
	

%% file: threshold_talagrand_setup.tex
\section{Semisorting and Reduction to Semisorted Functions}\label{sec:semisorted}

As we mentioned earlier when we stated~\eqref{eq:hope2}, we do not know if this is a true statement for an arbitrary function.
It is true for what we call semisorted functions, and proving this would be the bulk of the work. In this section, we define what semisorted
functions are, we prove that the Talagrand objective can only decrease when one moves to a semisorted function, and therefore how one can reduce to proving~\Cref{thm:dir-tal}
only for semisorted functions. \smallskip

Fix a function $f:[n]^d \to \{0,1\}$. Fix a coordinate $i$ and fix an interval $I = [a,b]$.
Semisorting $f$ on this interval in dimension $i$ leads to a function $h:[n]^d \to \{0,1\}$ as follows. We take every $i$-line $\ell$
and consider the function restricted on the interval $I$ on this line, and we sort it. The following lemma shows that semisorting on any $(i,I)$ pair
can only reduce the Talagrand objective. We defer its proof to~\Cref{sec:semisorting-can-only-reduce}.

\begin{mdframed}[backgroundcolor=blue!10,topline=false,bottomline=false,leftline=false,rightline=false] 
\begin{lemma}[Semisorting only decreases $T_\Phi$.]\label{lem:semisorting-decreases}
	Let $f$ be any hypergrid function and let $\chi$ be any bicoloring of the augmented hypergrid edges. Let $i\in [d]$ be any dimension and $I$ be any interval $[a,b]$.
	There exists a (re)-coloring $\chi'$ of the edges of the augmented hypergrid such that
	\[
		T_{\Phi_\chi}(f) \geq T_{\Phi_{\chi'}}(h)
	\]
	where $h$ is the function obtained upon semisorting $f$ in dimension $i$ on the interval $I$.
\end{lemma}
\end{mdframed}
\noindent
A function $f:[n]^d \to \{0,1\}$ is called {\em semisorted} if for any $i\in [d]$ and any $i$-line $\ell$, the function restricted to the first $n/2$ points is sorted increasing
and the function restricted to the second half is also sorted increasing. It is instructive to note that when $n=2$, that is when the domain is the hypercube, every function is semisorted.
This shows that semisorted functions form a non-trivial family. However, the semisortedness is a property that allows us to prove that
\eqref{eq:hope2} holds. In particular, we prove this theorem.
\begin{mdframed}[backgroundcolor=gray!20,topline=false,bottomline=false,leftline=false,rightline=false] 
\begin{restatable}[Connecting Talagrand Objectives of $f$ and Tracker Functions]{theorem}{semisorted}
\label{thm:semisorted-reduce-to-g}
	Let $f \colon [n]^d \to \{0,1\}$ be a {\bf \em semisorted} function and let $\chi \colon E\to \{0,1\}$ be an arbitrary coloring of the edges of the fully augmented hypergrid.
Then for every $\bx \in [n]^d$, one can find a coloring $\xi_\bx$ of the edges of the Boolean hypercube such that
\[
T_{\Phi_\chi}(f) := \Exp_{\bx \in [n]^d}~\left[\sqrt{\Phi_{f,\chi}(\bx)}  \right] ~~\geq~~ \Exp_{\bx \in [n]^d} ~\Exp_{S\subseteq [d]}~~[\sqrt{I_{g_\bx, \xi_\bx}(S)}] \text{.}
\]
\end{restatable}
\end{mdframed}
\noindent
We can use the above theorem to get set the intuition behind~\eqref{eq:hope1} correct, and prove~\Cref{thm:dir-tal} for semisorted functions. We state this below, but we defer the proof 
of this to~\Cref{sec:semisorted-tal-dist}. At this point we remind the reader again that this is not at all trivial, but the proof ideas are generalizations of those present in~\cite{KMS15,PRW22}
for the hypercube case.
\begin{mdframed}[backgroundcolor=gray!20,topline=false,bottomline=false,leftline=false,rightline=false] 
	\begin{restatable}[\Cref{thm:dir-tal} for semisorted functions.]{theorem}{dirtalsemisorted}\label{thm:dir-tal-semisorted}
		Let $f:[n]^d \to \{0,1\}$ be a {\bf \em semisorted} function that is $\eps$-far from monotone. Let $\chi:E\to \{0,1\}$ be an arbitrary coloring of the edges of the augmented hypergrid.
		Then there is a constant $C''$ such that 
		\[
		T_{\Phi_\chi}(f) := \Exp_\bx~\left[\sqrt{\Phi_{f,\chi}(\bx)}  \right]  \geq C''\eps
		\]
	\end{restatable}
\end{mdframed}

\Cref{lem:semisorting-decreases} shows that the Talagrand objective can't rise on semisorting. The distance to monotonicty, however, can fall. In the remainder of the section we show 
how we can reduce to the semisorted case with a loss of $\log n$, and in particular, we use~\Cref{thm:dir-tal-semisorted} to prove~\Cref{thm:dir-tal}.

\paragraph{Sequence of Semisorted Functions and Reduction to the Semisorted Case.}
We now describe a semi-sorting process which gives a way of getting from $f$ to a monotone function. Without much loss of generality, let us assume $n=2^k$ which we can assume by padding.
Iteratively coarsen the domain $[n]^d = [2^k]^d$ as follows.	
First ``chop'' this hypergrid into $2^d$ many $[n/2]^d = [2^{k - 1}]^d$ hypergrids by slicing through the ``middle'' in each of the $d$-coordinates.
More precisely, these $2^d$ hypergrids can be indexed via $\bv \in \{0,1\}^d$, where given such a vector, the corresponding hypergrid is
\[
H_\bv = \prod_{i=1}^d \{\bv_i \cdot \frac{n}{2} + 1, \bv_i \cdot \frac{n}{2} + 2, \cdots, \bv_i \cdot \frac{n}{2} + \frac{n}{2} \}
\]
Each hypergrid $H_\bv$ is an $[n/2]^d = [2^{k-1}]^d$ hypergrid. Let us denote the collection of all these hypergrids as the set $\cH_1$.
So, $\cH_1$ has $2^d$ many hypergrids and each hypergrid has dimension $[n/2]^d = [2^{k-1}]^d$.
Repeat the above operation on each hypergrid in $\cH_1$. More precisely, each hypergrid $H_\bv$ in $\cH_1$ will lead to $2^d$ hypergrids each with dimension $[n/4]^d = [2^{k-2}]^d$.
The total number of such hypergrids, which we collect in the collection $\cH_2$, is $2^d \times 2^d = (2^2)^d$. More generally, we have a family $\cH_i$ consisting of $\left(2^i\right)^d$ many
hypergrids of dimension $[n/2^i]^d = [2^{k - i}]^d$. The collection $\cH_{k-1}$ consists of $(2^{k-1})^d$ many $d$-dimensional hypercubes.

\begin{figure}[ht!]
	\begin{center}
		\includegraphics[trim = 0 150 0 60, clip, scale=0.4]{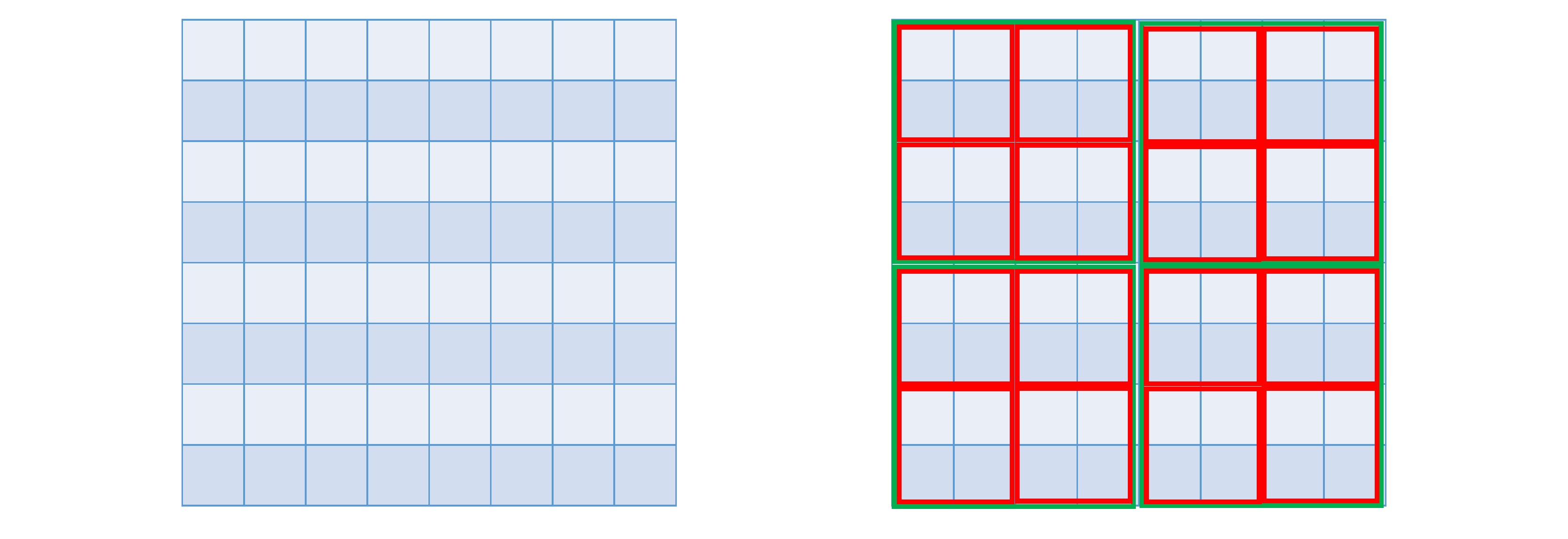}
		\caption{\em In the figure, we see an example with $d = 2$ and $n = 8 =2^{3}$. There are $2^2$ many $4\times 4$ green (hyper)-grids, 
			and $4^2$ many $2\times 2$ red squares.}
	\end{center}
\end{figure}
\noindent
Note that in any family $\cH_i$ for $1\leq i\leq k-1$, each $H\in \cH_i$ is a sub-hypergrid of $[n]^d$. We let $f_H$ denote the restriction of $f$ only to this subset $H$ of the domain.
Also, let $\cH_0$ denote the singleton set containing only one hypergrid, $[n]^d$.
Define the function $f_1 : [n]^d \to \{0,1\}$ as follows: consider every hypergrid\footnote{these will be hypercubes} $H$ in $\cH_{k-1}$ and  apply the sort operator on $f_H$ for all these hypergrids. Note that $f_1$ is a monotone function when restricted to $H\in \cH_{k-1}$. Recursively define $f_i$ as follows: consider every hypergrid $H\in \cH_{k-i}$ and apply the sort operator on $(f_{i-1})_H$ for all these hypergrids. 
\Cref{fig:semisortwithvals}~is an illustration for $d = 2$ and $k = 3$, i.e. $n=8$.

\begin{figure}[ht!]
	\begin{center}
		\includegraphics[trim = 0 45 0 30, clip, scale=0.5]{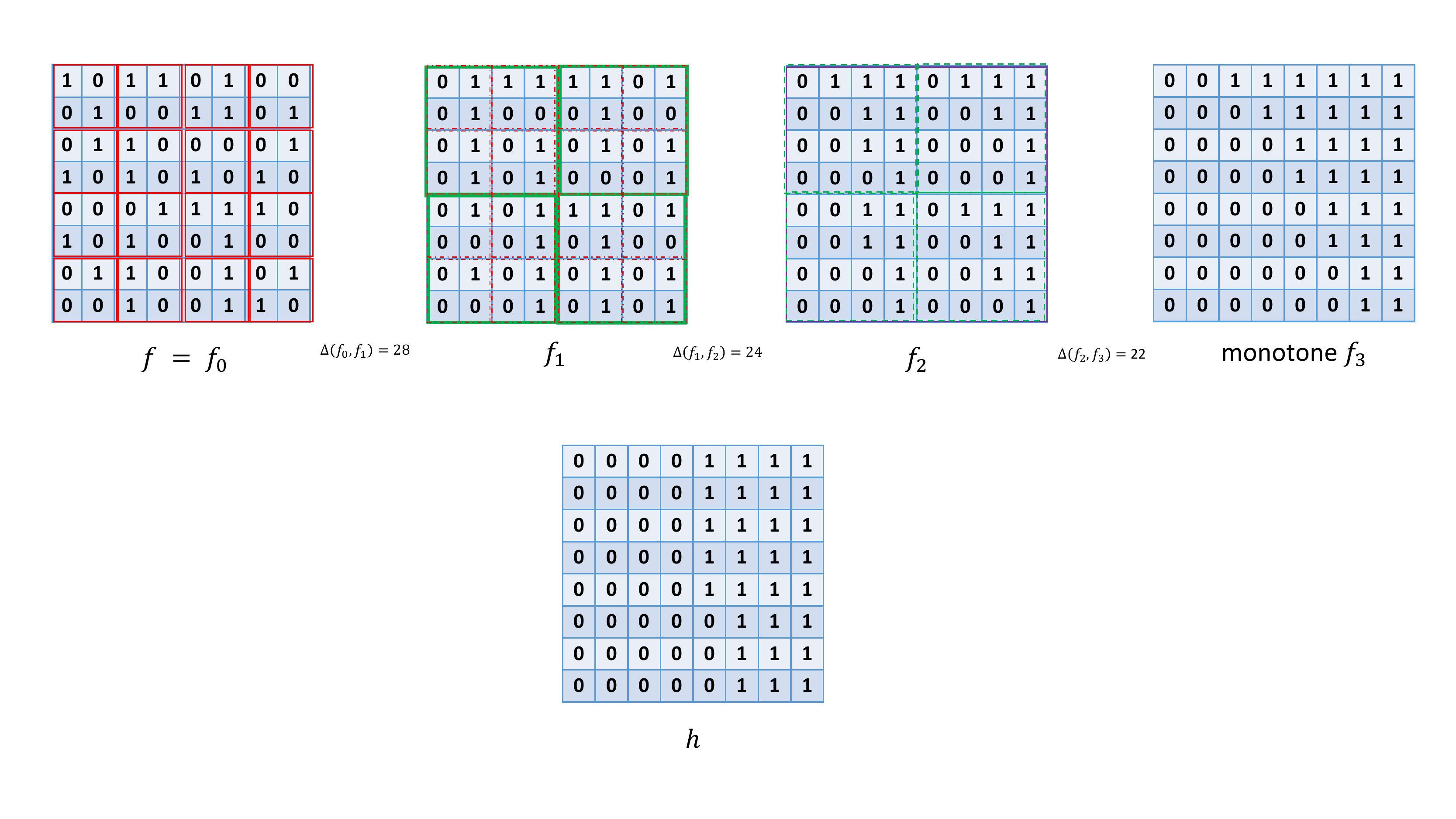}
		\caption{\em The function $f=f_0$ is described to the left, and then one obtains $f_1, f_2$ and $f_3$. The function $h$ which is obtained doing sort on the whole of $f$
			is described below. Note $h\neq f_3$.}\label{fig:semisortwithvals}
	\end{center}
\end{figure}

\begin{claim}\label{clm:triangle-ineq}
	There must exist an $0\leq j \leq k-1$ such that $\Delta(f_{j}, f_{j+1}) \geq \eps_f/k$.
\end{claim}
\begin{proof}
	This follows from triangle inequality and the fact that $\Delta(f_0, f_k) \geq \eps_f$.
\end{proof}
\begin{proof}[\bf Proof of~\Cref{thm:dir-tal}]
	We now show how~\Cref{thm:dir-tal} follows from~\Cref{lem:semisorting-decreases} and~\Cref{thm:dir-tal-semisorted} via an averaging argument.
	We fix the $j$ as in~\Cref{clm:triangle-ineq}.
	By~\Cref{lem:semisorting-decreases} we get that for any function $f$ and any coloring $\chi$, there exists a recoloring $\chi'$ such that
$T_{\Phi_{\chi}}(f) \geq T_{\Phi_{\chi'}}(f_{j})$. Now consider the hypergrids in $H\in \calH_{k-j-1}$. Let $f_j|_H$ be the function restricted to this sub-domain $H$.
Note that the function $f_j|_H$ is indeed semisorted by construction. Therefore, by~\Cref{thm:dir-tal-semisorted} (on the coloring $\chi'$) we know that for all $H\in \calH_{k-j-1}$, 
\[
T_{\Phi_{\chi'}}({f_j|_H}) \geq C''\cdot \eps_{f_{j}|_H}	
\]
By~\Cref{clm:2appx}, we know that $2\eps_{f_{j}|_H} \geq \Delta(f_{j}|_H, f_{j+1}|_H)$. Taking expectation over $H\in \cH_{k-j-1}$, 
we see that the LHS is at most (at most since we only consider violations staying in $H$) $T_{\Phi_{\chi'}}(f_j)$, while the RHS is precisely $\Delta(f_j,f_{j+1})/2 \geq \eps_f/2k$.
Putting everything together, we get
$
T_{\Phi_\chi}(f)   \geq \frac{C''\eps_f}{2\log n}
$
proving~\Cref{thm:dir-tal}.
\end{proof}

%% file: semisorting_reduction.tex
\subsection{Semisorting only decreases the Talagrand objective: Proof of~\Cref{lem:semisorting-decreases}}\label{sec:semisorting-can-only-reduce}

Let us first describe the coloring $\chi'$. 
\begin{itemize}
\item First let us describe the recoloring of pairs of points $(\bx, \bx')$ which differ only in some coordinate $j\neq i$
and $\bx_i = \bx'_i$ lies in the interval $[a,b]$. We go over all these edges by considering pairs of $i$-lines which differ on a single coordinate $j\neq i$.
More precisely, if $\ell = \bx \pm t\be_i$ then $\ell' = \bx' \pm t\be_i$ for some $\bx' = \bx + a\be_j$ with $a > 0$. We now consider re-coloring the pairs
$(\bx, \bx' = \bx + a\be_j)$ as follows.

Let $V$ denote the points $\bx \in \ell$ such that (a) $\bx_i \in I$, (b) $f(\bx) = 1$, but (c) $f(\bx+ a\be_j) = 0$. That is $(\bx, \bx+a\be_j)$ is a violation.
Consider all edges $E_V := \{(\bx, \bx+a\be_j)~:~\bx\in V\}$ and let $\vv{\chi}$ be the $|E_V|$ dimensional $0,1$-vector which are the $\chi$ values of edges in $E_V$ going left to right.

Now consider the function $h$ where $I$ has been sorted on both $\ell$ and $\ell'$. Let $U$ denote the points $\bx \in \ell$ such that (a) $\bx_i \in I$, (b) $h(\bx) = 1$, but (c) $h(\bx+ a\be_j) = 0$. That is $(\bx, \bx+a\be_j)$ is a violation in $h$. Firstly note that $|U| \leq |V|$ and furthermore, these $|U|$ points form a contiguous interval of $I$.
We now describe the recoloring $\chi'$  of the edges in $E_U := \{(\bx, \bx+a\be_j)~:~\bx\in V\}$; all the other recolorings are immaterial since they don't contribute to $T_{\Phi_{\chi'}}(h)$ since the edges are not violating. We take the $|V|$-dimensional vector $\vv{\chi}$, sort in {\em decreasing} order, and then take the first $|U|$ coordinates and use them to define $\chi'(e)$ for $e\in E_U$, left to right. See~\Cref{fig:semisorting-reduces} for an illustration.

\begin{figure}[ht!]
	\includegraphics[trim = 0 300 0 50, clip, scale=0.5]{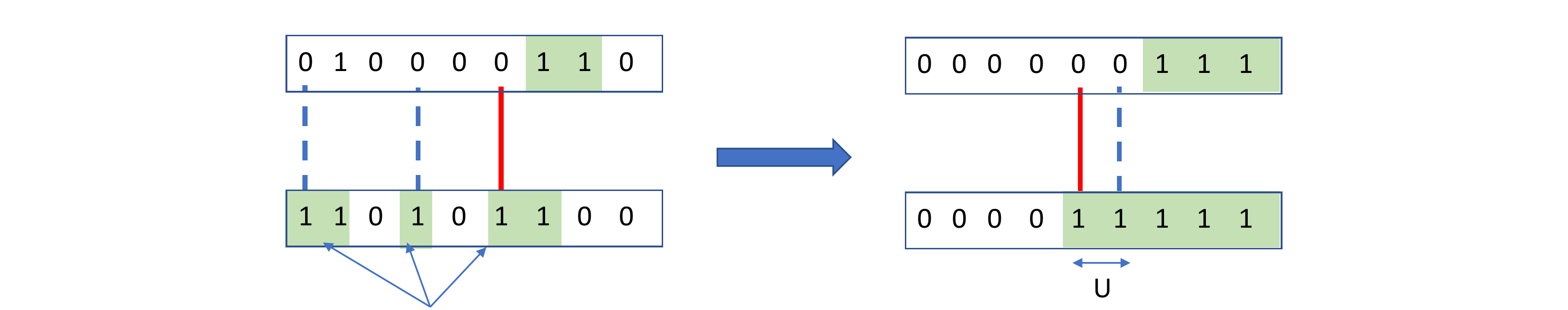}
	\caption{\em We are considering only the interval $I$. The line below is $\ell$ and the line above is $\ell'$. The green shaded zones correspond to where the function evaluates to $1$s.
		The situation to the right is after sorting. Only the violating edges are marked. On the left,
		the red solid edges are colored $\chi(e) = 1$ while the blue dashed 
	are colored $\chi(e) = 0$. On the right, the color-coding is the same but for $\chi'$. All other unmarked edges inherit the same colors as $\chi$. }\label{fig:semisorting-reduces}
\end{figure}

\item Now we describe recoloring of pairs of points $(\bx, \by)$ which only differ in coordinate $i$. First, if both $\bx_i$ and $\by_i$ lie in $I$, or if they both lie outside $I$, then we leave their colors unchanged. Furthermore, if $(\bx, \by)$ is {\em not} a violating pair in $f$, then we leave its color unchanged. 
Now consider a $\by$ to the right of $I$, that is, $\by_i > b$ and $f(\by) = 0$. Consider the $\bx$'s with $\bx_i$ in $I$ with $f(\bx) = 1$, each of which forms a violation with $\by$. Suppose there are $k$ many of them, of which $k_0$ of them are colored $0$
and $k_1$ of them are colored $1$. We now consider the picture in $h$, and once again there are exactly $k$ (possibly different) points in the interval which are violating with $\by$ in $h$.
Going from left to right, we color the first $k_1$ of them $1$ and the next $k_0$ of them $0$, in $\chi'$. 
We now do a similar thing for a $\bz$ to the left of $I$, that is, $\bz_i < a$ and $f(\bz) = 1$. We now consider the $\bx$'s with $\bx_i \in I$ with $f(\bx) = 0$, each of which forms a violation with $\bz$.
As before, suppose there are $k$ many of them $k_1$ of them colored $1$ and $k_0$ of them colored $0$. In $g$ also there are $k$ locations with which $\bz$ is a violation. 
We, once again, going from left to right, color the first $k_1$ of them $1$ and the next $k_0$ of them $0$, in $\chi'$.	
See~\Cref{fig:semisorting-reduces-2} for an illustration.

\begin{figure}[ht!]
	\includegraphics[trim = 0 300 0 0, clip, scale=0.5]{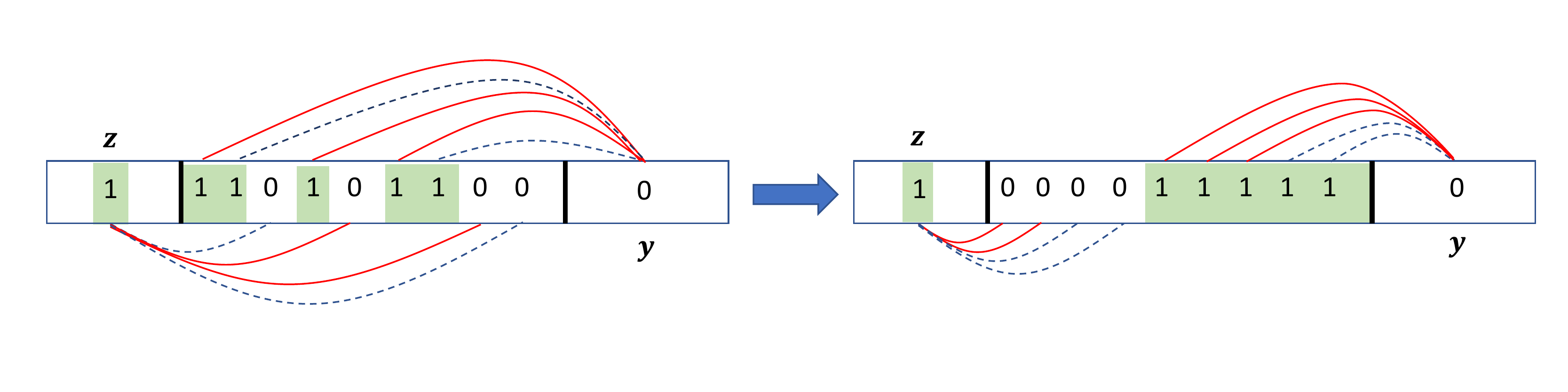}
	\caption{\em The two vertical black lines demarcate $I$. 
		The green shaded zones correspond to where the function evaluates to $1$s. The situation to the right is after sorting. 
		$\by$ is a point with $f(\by) = 0$ to the right of $I$; $\bz$ is a point with $f(\bz) = 1$ to the left of $I$.
		Only the violating edges incident to $\by$ and $\bz$ are marked. On the left,
		the red solid edges are colored $\chi(e) = 1$ while the blue dashed 
		are colored $\chi(e) = 0$. On the right, the color-coding is the same but for the recoloring $\chi'$. All other unmarked edges incident of $\by$ or $\bz$ inherit the same colors as $\chi$. 
	Edges with both endpoints in $I$ or both endpoints outside $I$ also inherit the same color.}\label{fig:semisorting-reduces-2}
\end{figure}

\end{itemize}

Now we prove the lemma ``line-by-line''. In particular, we want to prove for any $i$-line $\ell$, we have 
\[
\sum_{\bx\in \ell} \sqrt{\Phi_{f,\chi}(\bx)} \geq 	\sum_{\bx\in \ell} \sqrt{\Phi_{h,\chi'}(\bx)}
\]
Note that it suffices to prove the above for $\bx$ whose $\bx_i \in I$.

To prove the above inequality, it is best to consider the two vectors $\vv{\Phi_\chi(f)}$ and $\vv{\Phi_{\chi'}(h)}$ which are $|I|$-dimensional whose  $\bx$th coordinate is precisely $\Phi_{f,\chi}(\bx)$
and $\Phi_{h,\chi'}(\bx)$ respectively. We want to prove 
\begin{equation}\label{eq:line}
	\norm{\vv{\Phi_\chi(f)}}_{1/2} \geq \norm{\vv{\Phi_{\chi'}(h)}}_{1/2}
\end{equation}

First we divide the $|I|$ coordinates of $\vv{\tI_\chi(f)}$ into $O \cup Z$ corresponding to when $f(\bx) = 1$ and $f(\bx) = 0$. Let's call these two vectors $\vv{\tI^{(1)}_\chi(f)}$ and $\vv{\tI^{(0)}_\chi(f)}$. The former vector is $|O|$ dimensional, the latter is $|Z|$ dimensional, and $\vv{\tI_\chi(f)}$ is obtained by some splicing of these two vectors.
We will do the same for the coordinates of $\vv{\tI_{\chi'}(h)}$ to obtain $\vv{\tI^{(1)}_{\chi'}(h)}$ and $\vv{\tI^{(0)}_{\chi'}(h)}$. Note that since sorting doesn't change the number of $0$s or $1$,
both these vectors are $|O|$ and $|Z|$ dimensional, respectively.
We now set to prove
\begin{equation}\label{eq:2cases}
	\norm{\vv{\tI^{(1)}_\chi(f)}}_{1/2} \geq \norm{\vv{\tI^{(1)}_{\chi'}(h)}}_{1/2}  ~~~\textrm{and}~~~ \norm{\vv{\tI^{(0)}_\chi(f)}}_{1/2} \geq \norm{\vv{\tI^{(0)}_{\chi'}(h)}}_{1/2}
\end{equation}
and this will prove \eqref{eq:line}. We prove the first inequality; the proof of the second is analogous. 
For brevity's sake, for the rest of the section we drop the superscript $(1)$ from $\vv{\tI^{(1)}}$. \medskip

\noindent
The plan is to write $\vv{\tI_\chi(f)}$ as a sum of (Boolean) vectors, and then show that $\vv{\tI_{\chi'}(h)}$ is dominated by the sum of sorts of those Boolean vectors. Then we invoke~\Cref{lem:sum-of-vectors}.

We write $\vv{\tI_\chi(f)}$ as a sum of Boolean vectors as follows.
Fix any other $i$-line $\ell' := \ell + a\be_j$ for some $j\neq i$ and $a > 0$. Define the following 
$(0,1)$-vector also indexed by elements of $O$.
\[
\bu_{\ell'}(\bx) = 1 ~~\text{if $f(\bx + a\be_j) = 0$ and $\chi(\bx, \bx+a\be_j) = 1$}
\]
That is, $\bu_{\ell'}(\bx) = 1$ if the projection of $\bx$ onto $\ell'$, $(\bx, \bx' := \bx + a\be_j)$, is a violating edge in $f$ with $\chi$-color $1$.

Define the following vector $\bA$ as follows. 

\begin{definition} \label{eq:def1}
	For any $\bx\in  O$, 
	\[
	\bA(\bx) = \sum_{j\neq i} \underbrace{\min\left(1, \sum_{\ell' = \ell + a\be_j} \bu_{\ell'}(\bx) \right)}_{\text{let's call this}~\bw_j(\bx) \in \{0,1\}} ~~=:\sum_{j\neq i} \bw_j(\bx)
	\]
\end{definition}

Finally, for $\bx \in O$, define
\[
\bbA(\bx) = 1~~\textrm{if there is some $\by$ to its right, potentially outside the interval $I$ with $f(\by) = 0$ and $\chi(\bx,\by) = 1$}
\]
Using the vectors, we can write
\begin{observation}\label{obs:phiA}
	For any $\bx \in O$,
	\[
	\vv{\tI_\chi(f)}(\bx) =  \bA(\bx) + \bbA(\bx)
	\]
\end{observation}

Now let's consider the situation after $I$ is sorted. The ones of $O$ now ``shift around''; indeed, they are the $|O|$ many right most points.
Let's call these locations $O'$ and note $|O'| = |O|$.

Now define the $|O'|=|O|$ dimensional vector $\bv_{\ell'}$ where for $\bx \in O'$ 
\[
\bv_{\ell'}(\bx) = 1 ~~\text{if $h(\bx + a\be_j) = 0$ and $\chi'(\bx, \bx+a\be_j) = 1$}
\]

Now we will use the property of the recoloring we performed. We claim two things:


\begin{claim}\label{clm:crucial-1}
	The number of $1$s in $\bv_{\ell'}$ is at most the number of $1$s in $\bu_{\ell'}$, and $\bv_{\ell'}$ is sorted decreasing.
\end{claim}
\begin{proof}
	The number of $1$s in $\bu_{\ell'}$ is precisely the number of violating edges of the form $(\bx, \bx')$ in $f$, where $\bx_i \in I$ and $\bx' = \bx + a\be_j$ and $\chi(\bx,\bx') = 1$.
	Similarly, 	the number of $1$s in $\bu_{\ell'}$ are precisely the number of violating edges of the form $(\bx, \bx')$ in $h$, where $\bx_i \in I$ and $\bx' = \bx + a\be_j$ and $\chi'(\bx,\bx') = 1$.
	When we recolored to get $\chi'$ we made sure by property (a) that the latter number is smaller.
	
	Take $\bx$ and $\by$ in $O$, with $\bx_i < \by_i$, but suppose, for the sake of contradiction, $\bv_{\ell'}(\bx) = 0$ and $\bv_{\ell'}(\by) = 1$.
	The latter implies $h(\by' := \by + a\be_j) = 0$ and $\chi'(\by,\by') = 1$. Since $h$ is sorted on $\ell'$, $h(\bx' := \bx + a\be_j) = 0$ as well. 
	Since $\bx \in O$, $h(\bx) = 1$ which means $(\bx, \bx')$ is a violating edge in $h$. $\bv_{\ell'}(\bx) = 0$ implies $\chi'(\bx,\bx') = 0$. 
	But this violates property (b) of $\chi'$.
\end{proof}

%

What we need is the following corollary.
\begin{equation}\label{eq:coordom-sort}
	\textrm{For any $\ell' = \ell + a\be_j$}, ~~~ \bv_{\ell'} ~\leq_{\mathrm{coor}} \sortdown{\bu_{\ell'}}
\end{equation}
where recall that $\sortdown{z}$ is the sorted-decreasing version of $z$.

Just as we defined $\bA$, define the $|O|$-dimensional vector $\bB$ as follows.
\begin{definition}\label{eq:def2}
	For any $\bx\in O'$, 
	\[
	\bB(\bx) = \sum_{j\neq i} \underbrace{\min\left(1, \sum_{\ell' = \ell + a\be_j} \bv_{\ell'}(\bx) \right)}_{\text{let's call this}~\bz_j(\bx) \in \{0,1\}} ~~=:\sum_{j\neq i} \bz_j(\bx)
	\]
\end{definition}
Note that for every $j \neq i$, $\bw_j$ and $\bz_j$ are $|O|=|O'|$ dimensional Boolean vectors which we index by $\bx \in O$ and $\bx \in O'$, respectively.
\begin{claim}\label{clm:z-vs-w}
	For all $j$, $\bz_j \leq_{\mathrm{coor}} \sortdown{\bw_j}$.
\end{claim}
\begin{proof}
Follows from \eqref{eq:coordom-sort}, and the defintions of $\bz_j$ and $\bw_j$ as described in~\Cref{eq:def1} and~\Cref{eq:def2}.
\end{proof}
%

Finally, for $\bx \in O'$, define the $|O'| = |O|$ dimensional vector $\bbB$ as
\[
\bbB(\bx) = 1~~\textrm{if there is some $\by$ to its right, outside the interval $I$ with $h(\by) = f(\by) = 0$ and $\chi'(\bx,\by) = 1$\text{.}}
\]
Just as in~\Cref{obs:phiA}, note that
\begin{observation}\label{obs:Phib}
	For any $\bx \in O'$,
	\[
	\vv{\Phi_{\chi'}(h)}(\bx) =  \bB(\bx) + \bbB(\bx)
	\]
\end{observation}
We now connect $\bbA$ and $\bbB$ as follows.
\begin{claim}\label{eq:googoo}
	$\bbB \leq_{\mathrm{coor}} \sortdown{\bbA}$	
\end{claim}
\begin{proof}
	Similar to~\Cref{clm:crucial-1}, this follows from the following claim.
	\begin{claim}\label{clm:crucial-2}
		The number of $1$s in $\bbB$ is at most that in $\bbA$, and 	$\bbB$ is sorted decreasing.
	\end{claim}
	\begin{proof}
		This also follows from the way we recolor $\chi'$ the pairs of the form $(\bx, \by)$ with $\by$ lying to the right of $I$ and $f(\by) = 0$.
		First let's show $\bbB$ is sorted decreasing. Take two points $\bx$ and $\bz$ with $a < \bx_i < \bz_i < b$ both evaluating to $1$ in $g$.
		Say, $\bbB(\bz) = 1$ implying there is some $\by$ with $g(\by) = f(\by) = 0$ to the right of $I$ s.t. $\chi'(\bz,\by) = 1$.
		However, the way we recolor the edges incident on $\by$, this implies $\chi'(\bx, \by) = 1$ as well. But that would imply $\bbB(\bx) = 1$.
		
		The first part of the claim also follows from the way we recolor. Suppose the number of ones in $\bbA$ is $t$.
		That is, only $t$ of the points in $O$ have $1$-colored edges going to the right of the interval. Consider the subset $W$ of these outer endpoints.
		The function value, both $f$ and $g$, are $0$ here. Note that none of these points in $W$ have more than $t$ edges incident on them which are colored $1$ in $\chi$.
		Now note that in $\chi'$, this number of $1$-edges are conserved, and so for every $\bw \in W$, the number of $1$-colored violating edges is still $\leq t$.
		Now suppose for contradiction $\bbB$ has $(t+1)$ ones. Take the right most point $\bx$ and consider the violating edge $(\bx, \by)$ which is colored $1$ in $\chi'$.
		By construction, this $\by$ must have $1$-colored edges to all the $(t+1)$ points (since we color them $1$ left-to-right). This contradicts the number of $1$-edges incident on $\by$.
\end{proof}\end{proof}

To summarize, we have from~\Cref{obs:phiA} and~\Cref{eq:def1}, 
\[
\vv{\tI_\chi(f)} = \sum_{j\neq i} \bw_j + \bbA
\]
that is, we have written the LHS as a sum of Boolean vectors.
And, we have from~\Cref{obs:Phib} and~\Cref{eq:def2}, followed by~\Cref{clm:z-vs-w} and~\eqref{eq:googoo} that
\[
\vv{\tI_{\chi'}(h)} = \sum_{j\neq i} \bz_j + \bbB ~~~\leq_{\mathrm{coor}}~~~ \underbrace{\sum_{j\neq i} \sortdown{\bw_j} + \sortdown{\bbA}}_{\text{call this $\vv{s\Phi}$}}
\]
Trivially, we have $\norm{\vv{\tI_{\chi'}(h)}}_{1/2} \leq \norm{\vv{s\Phi}}_{1/2}$, and from~\Cref{lem:sum-of-vectors}, we get $\norm{\vv{s\Phi}}_{1/2} \leq \norm{	\vv{\tI_\chi(f)}}_{1/2}$, completing the proof of the first part of~\eqref{eq:2cases}.

%% file: telescoping_argument.tex
\section{Connecting with the Distance to Monotonicity: Proof of~\Cref{thm:dir-tal-semisorted}} \label{sec:semisorted-tal-dist}

In this section, we set the intuition behind~\eqref{eq:hope1} straight. We show how the isoperimetric theorem~\Cref{thm:semisorted-reduce-to-g} 
on semisorted functions can be used to prove~\Cref{thm:dir-tal-semisorted}.
We begin by recalling the corollary of the undirected, colored Talagrand
objective on the hypercube.
\corkms*
%
%
\noindent
As mentioned earlier, one can't show \eqref{eq:hope1}, that is, 
$\Exp_\bx [\var(g_{\bx})] = \Omega(\eps_f)$. Indeed, there are examples
of functions even over the hypercube where the above bound does \emph{not} hold.
KMS deal with this problem by applying \Thm{semisorted-reduce-to-g}
to random restrictions of $f$. One can show that there
is some restriction where the corresponding $\Exp_\bx [\var(g_{\bx})]$ is large. They referred
to these calculations as the ``telescoping argument". This argument was quantitatively improved
by Pallavoor-Raskhodnikova-Waingarten~\cite{PRW22}.

In this section, we port that argument to the hypergrid setting. Our proof
is different in its presentation, though the key ideas are the same as KMS.
Our first step is to convert \Thm{semisorted-reduce-to-g} to a more convenient form,
using the undirected~\Thm{kms-und}.

\begin{theorem}\label{thm:semisorted-reduce-conv}
There exists a constant $C' > 0$ such that for any 
semisorted function $f:[n]^d \to \{0,1\}$ and any arbitrary coloring $\chi:E\to \{0,1\}$ of the augmented hypergrid, we have
		\[
				T_{\Phi_\chi}(f) \geq C' \cdot \EX_S[\Delta(S \circ f, \overline{S} \circ f)]\text{.}
		\]
\end{theorem}

\begin{proof} By \Thm{semisorted-reduce-to-g}, there exists some colorings $\xi_\bx$ such that $T_{\Phi_\chi}(f) \geq \EX_\bx\EX_S[\sqrt{\infl{{g_\bx,\xi_\bx}}{S}}]$.
By the undirected Talagrand bound \Thm{kms-und},
$\EX_S[\sqrt{\infl{{g_\bx,\xi_\bx}}{S}}] \geq C\cdot \var(g_\bx)$.

\begin{eqnarray}
    \EX_S[\Delta(S \circ f, \overline{S} \circ f)] & = & \EX_S \EX_\bx[\mbone((S \circ f)(\bx) \neq (\overline{S} \circ f)(\bx))] \nonumber \\
    & = & \EX_S \EX_\bx[\mbone(g_\bx(S) \neq g_\bx(\overline{S}))] \nonumber \\
    & = & \EX_\bx \EX_S[\mbone(g_\bx(S) \neq g_\bx(\overline{S}))] \leq 4\EX_\bx[\var(g_\bx)] \label{eq:delta-var}
\end{eqnarray}
(The final inequality uses \Clm{var-prob}, stated below.)
Hence, $\EX_\bx\EX_S[\sqrt{\infl{{g_\bx,\xi}}{S}}] \geq (C/4)\EX_S[\Delta(S \circ f, \overline{S} \circ f)]$.

\begin{claim} \label{clm:var-prob} For any Boolean function $h:\hyp{d} \to \hyp{}$, $\Pr_S[h(S) \neq h(\overline{S})] \leq 4\var(h)$.
\end{claim}

\begin{proof} Recall that $\var(h) = 4\Pr_S[h(S) = 0] \Pr_S[h(S) = 1]$.
Hence, $\var(h) = 4\max_{b \in \{0,1\}} \Pr_S[h(S) = b] \min_{b \in \{0,1\}} \Pr_S[h(S) = b]$.
Since one of the values is taken with probability at least $1/2$, $\var(h) \geq 2 \min_{b \in \{0,1\}} \Pr_S[h(S) = b]$.

Let $\bfS = \{S \ | \ h(S) \neq h(\overline{S})\}$. Observe that half the sets in $\bfS$ have an $h$-value of $1$,
and the other half have value zero. Hence, $\Pr_S[h(S) \neq h(\overline{S})] \leq 2 \min_{b \in \{0,1\}} \Pr_S[h(S) = b]$. 
Combining with the bound from the previous paragraph, $\Pr_S[h(S) \neq h(\overline{S})] \leq 4\var(h)$.
\end{proof}

\end{proof}

We now give some definitions and claim regarding the Talagrand objective of random restrictions
of functions.

\begin{definition} \label{def:restrict} Let $S \subseteq [d]$ be a subset of coordinates.
The \emph{distribution of restrictions on $S$}, denoted $\cR_S$, is supported over functions and generated
as follows. We pick a uar setting of the coordinates in $\overline{S}$, and output the function under this restriction.
(Hence, $h \sim \cR_S$ has domain $[n]^S$.)
\end{definition}

The isoperimetric theorem of \Thm{semisorted-reduce-to-g} holds for any
ordering of the coordinates. In this section, we will need to randomize the ordering
of the sort operators.
We will represent an ordering
as a permutation $\pi$ over $[d]$. Abusing notation, for any subset $S \subseteq [d]$,
$\pi(S)$ is the induced ordered list of $S$.

\begin{definition} \label{def:delta} For any function $h: [n]^k \to \hyp{}$, define $\delta(h)$
to be $\EX_\pi[\Delta(h, \pi([k]) \circ h)]$.
\end{definition}

By \Clm{sortS}, sorting on all coordinates leads to a monotone function.
Thus, $\delta(h)$ is at least the distance of $h$ to monotonicity. 
We will perform our analyses in terms of $\delta(f)$, since it is more
amenable to a proof by induction over domain size.

The following claim is central to the final induction, and relates $\delta(f)$
to $\EX_S[\Delta(S \circ f, \overline{S} \circ f)]$. This is the (only) claim
where we need to permute the coordinates. All other claims and theorems
hold for an arbitrary ordering of the coordinates (when defining $S \circ f$).

\begin{claim} \label{clm:triangle} $\delta(f) \leq \EX_S \EX_{h \sim \cR(S)}[\delta(h)] + \EX_\pi \EX_S[\Delta(\pi(S) \circ f, \pi(\overline{S}) \circ f)]$
\end{claim}

\begin{proof} Let us consider an arbitrary ordering of dimensions.
By triangle inequality,
\begin{eqnarray*}
\Delta(f, S \circ \overline{S} \circ f) \leq \Delta(f, S \circ f) + \Delta(S \circ f, S \circ \overline{S} \circ f)
\end{eqnarray*}
Observe that $S \circ S \circ f = S \circ f$, since sorting repeatedly on a dimension does not modify a function.
Hence, $\Delta(S \circ f, S \circ \overline{S} \circ f) = \Delta(S \circ S \circ f, S \circ \overline{S} \circ f)
\leq \Delta(S \circ f, \overline{S} \circ f)$. The latter inequality holds because sorting only reduces the Hamming distance
between functions (\Clm{sort-hamm}). Plugging this bound in and taking expectations over ordered subset $S$ of dimensions:
\begin{equation} \label{eq:restrict}
\EX_S[\Delta(f, S \circ \overline{S} \circ f)] \leq \EX_S[\Delta(f, S \circ f)] + \EX_S[\Delta(S \circ f, \overline{S} \circ f)]
\end{equation}
Observe that $S \circ f$ only changes the function in the dimensions in $S$, and can be thought
to act on the restrictions of $f$ (to $S$). Hence $\EX_S[\Delta(f, S \circ f)] = \EX_{h \sim \cR(S)}[\Delta(h, S \circ h)]$.
Roughly speaking, the quantity $\Delta(f, S \circ \overline{S} \circ f)$ is $\eps(f)$
and $\EX_{h \sim \cR(S)}[\Delta(h, S \circ h)]$ is $\EX_{h \sim \cR{S}} \eps(h)$. So we would
hope that \Eqn{restrict} implies $\eps(f) \leq \eps(h) + \EX_S[\Delta(S \circ f, \overline{S} \circ f)]$. 

Unfortunately, the quantities are only constant factor approximations of $\eps(f), \eps(h)$.
So by converting \Eqn{restrict} in terms of $\eps(f)$, we would potentially lose a constant factor
in \Eqn{restrict}.

To avoid this problem, we deal with $\delta(f)$ instead. By randomly permuting $S$ and taking expectations,
the quantities in \Eqn{restrict} can be replaced by $\delta(\cdot)$ terms.
Taking expectations over a uar $\pi$, \Eqn{restrict} implies
\begin{equation}
\EX_\pi\EX_S[\Delta(f, \pi(S) \circ \pi(\overline{S}) \circ f)] \leq \EX_\pi\EX_S[\Delta(f, \pi(S) \circ f)] + \EX_\pi\EX_S[\Delta(\pi(S) \circ f, \pi(\overline{S}) \circ f)]
\end{equation}

Note that the switching order in the LHS, $\pi(S) \circ \pi(\overline{S})$, is uniformly random.
Moreover,
$$ \EX_\pi \EX_S \EX_{h \sim \cR(S)}[\Delta(h, \pi(S) \circ h)] = \EX_S \EX_h \EX_\pi [\Delta(h, \pi(S) \circ h)] = \EX_S \EX_h[\delta(h)] $$
Combining all our bounds, we get that $\delta(f) \leq \EX_S \EX_{h \sim \cR(S)}[\delta(h)] + \EX_\pi \EX_S[\Delta(\pi(S) \circ f, \pi(\overline{S}) \circ f)]$.
\end{proof}

We prove a useful claim about the Talagrand objective of restrictions, made in~\cite{PRW22}.

\begin{claim} \label{clm:tal-restrict} Let $p \in (0,1)$, and $\cH(p)$ be the distribution
of subsets of $[d]$ generated by selecting each element with iid probability $p$. Then,
$\dtal{f} \geq (1/\sqrt{p}) \cdot \EX_{S \sim \cH(p)} \EX_{h \sim \cR_S} [\dtal{h}]$.
\end{claim}

\begin{proof} Fix a set $S$. For any subset $S$ of coordinates, let the define the influence in $S$ as $\Phiinfl{f,\chi}{\bx;S} \eqdef \sum_{i \in S} \Phiinfl{f,\chi}{\bx;i}$.
We are just summing the influences over the coordinates of $S$.

Consider the quantity $\EX_{h \sim \cR_S} [\dtal{h}] = \EX_{h \sim {\cR_S}} \EX_{\bz} [\sqrt{\Phiinfl{h,\chi}{\bz}}]$.
Note that $\bz$ denotes a uar setting of the coordinates in $S$. The colorings of $h$ are inherited from the coloring of $f$.
Each function $h$ is indexed
by a (uar) setting of $\overline{S}$. Hence, 

\begin{equation}
\EX_{h \sim {\cR_S}} \EX_{\bz} [\sqrt{\Phiinfl{h,\chi}{\bz}}] = \EX_{\bx} [\sqrt{\Phiinfl{f,\chi}{\bx; S}}]
\end{equation}

The point $\bx$ is uar in the entire domain $[n]^d$. 
Note that $\EX_{S \sim \cH(p)} [\Phiinfl{f,\chi}{\bx; S}]$
is precisely $p \cdot \Phiinfl{f,\chi}{\bx; S}$, since each coordinate is independently picked in $S$ with probability $p$. 
\begin{eqnarray*}
    \EX_{S \sim \cH(p)} \EX_{h \sim \cR_S} [\dtal{h}] & = & \EX_{S} \EX_\bx [\sqrt{\Phiinfl{f,\chi}{\bx; S}}] \\
    & = & \EX_\bx \EX_{S}[\sqrt{\Phiinfl{f,\chi}{\bx; S}}] \\
    & \leq & \EX_\bx \Big[\sqrt{\EX_{S}[\Phiinfl{f,\chi}{\bx; S}]}\Big]
    = \EX_\bx \Big[\sqrt{p \cdot \Phiinfl{f,\chi}{\bx; S}}\Big] = \sqrt{p} \cdot \dtal{f}
\end{eqnarray*}
The inequality above is a consequence of the concavity of the square root function and Jensen's inequality.
\end{proof}
Now we have all the ingredients to prove~\Cref{thm:dir-tal-semisorted} whice we restate below for convenience.
\begin{mdframed}[backgroundcolor=gray!20,topline=false,bottomline=false,leftline=false,rightline=false] 
\dirtalsemisorted*
\end{mdframed}
%
%

\begin{proof} The proof is by induction over the dimension $d$ of the domain. 
Formally, we will prove a lower bound of $(C'/10)\eps$, where $C'$ is the constat of \Thm{semisorted-reduce-conv}.

Let us first prove the base case, when $d \leq 10$. 
Note that $\Phi_{f,\chi}(\bx) = \sum_{i=1}^d \Phi_{f,\chi}(\bx;i)$,
where each term in the summation is 0-1 valued. Hence, 
by the $l_1$-$l_2$-inequality, $\sqrt{\Phi_{f,\chi}(\bx)} \geq \sum_{i=1}^d \Phi_{f,\chi}(\bx;i)/d = \Phi_{f,\chi}(\bx)/d$.

Thus, $T_{\Phi_\chi}(f) \geq \EX_\bx[\Phi_{f,\chi}(\bx)]/d$. Furthermore,
$\EX_\bx[\Phi_{f,\chi}(\bx)] = \sum_{i=1}^d \EX_{\bx}[\Phi_{f,\chi}(\bx;i)]$. We can 
break the expectation over $\bx$ into lines as follows.
$$\EX_\bx[\Phi_{f,\chi}(\bx)] = \sum_{i=1}^d \EX_{\ell \ \textrm{uar $i$-line}} \EX_c [\Phi_{{f|_\ell},\chi}(c)]$$
(The coordinate $c$ is uar in $[n]$.) Now, for a Boolean function $f|_\ell$ on a line, if the distance to monotonicity 
is $\eps$, then there are at least $\eps n$ violating pairs~\cite{EKK+00}, and thus for any coloring $\chi$, we have 
$\EX_c [\Phi_{{f|_\ell},\chi}(c)] \geq \eps(f|_\ell)$, 
and $\sum_{i=1}^d \EX_{\ell \ \textrm{uar $i$-line}} \eps(f|_\ell) = \Omega(\eps(f))$. 

Hence, $T_{\Phi_\chi}(f) = \Omega(\eps/d)$. For $d \leq 10$, the lemma holds, and so henceforth we assume $d\geq 10$.

Now for the induction step. We now break into cases.

\underline{Case 1, $\EX_\pi \EX_S[\Delta(\pi(S) \circ f, \pi(\overline{S}) \circ f)] \geq \delta(f)/10$:} By \Thm{semisorted-reduce-conv},
$\dtal{f} \geq c\cdot \EX_S[\Delta(S \circ f, \overline{S} \circ f)]$ (for any ordering of coordinates). So
$\dtal{f} \geq c\cdot \EX_\pi \EX_S[\Delta(\pi(S) \circ f, \pi(\overline{S}) \circ f)] \geq (c/10) \cdot \delta(f)$.

\medskip

\underline{Case 2, $\EX_\pi \EX_S[\Delta(S \circ f, \overline{S} \circ f)] < \delta(f)/10$:} By \Clm{triangle},
$\EX_S \EX_{h \sim \cR(S)}[\delta(h)] \geq  \delta(f) - \EX_\pi \EX_S[\Delta(S \circ f, \overline{S} \circ f)]$.
In this case, we can lower bound $\EX_S \EX_{h \sim \cR(S)}[\delta(h)] \geq (9/10)\delta(f)$.
Note that $S$ is drawn from the distribution $\cH(1/2)$.
When $S \neq [d]$, we can apply induction to $\dtal{h}$ for $h \sim \cR(S)$. Hence,
\begin{eqnarray}
    \EX_{S \sim \cH(1/2)} \EX_{h \sim \cR(S)} [\dtal{h}] & \geq & 2^{-d} \sum_{S \neq [d]} \EX_{h \sim \cR(S)}[\dtal{h}] \geq 2^{-d} \cdot (c/10) \cdot \sum_{S \neq [d]} \EX_{h \sim \cR(S)} [\delta(h)] \notag \\
    & = & 2^{-d} \cdot (c/10) \cdot \Big(\sum_{S \subseteq [d]} \EX_{h \sim \cR(S)} [\delta(h)] - \EX_{h \sim \cR([d])} [\delta(h)]\Big)\notag \\
    & = & (c/10) \Big(\EX_S \EX_{h \sim \cR(S)} [\delta(h)] - 2^{-d} \delta(f)\Big) \ \ \ \ \textrm{($h \sim \cR([d])$ is $f$)} \notag \\
    &\geq& (c/10)\cdot(9/10)\cdot \delta(f) - 2^{-d}\cdot(c/10)\cdot \delta(f) \ \ \ \textrm{(by case condition)}\notag \\
        & = & (9/10 - 2^{-d}) \cdot (c/10) \cdot \delta(f) \geq (4/5) \cdot (c/10) \cdot \delta(f) \label{eq:tal-restrict}
\end{eqnarray}
By \Clm{tal-restrict}, $\dtal{f} \geq \sqrt{2}\cdot \EX_{S \sim \cH(1/2)} \EX_{h \sim \cR(S)} [\dtal{h}]$.
Combining with the inequality of \Eqn{tal-restrict}, $\dtal{f} \geq (\sqrt{2}\cdot 4/5) \cdot (c/10) \cdot \delta(f) \geq (c/10) \cdot \delta(f)$.
\end{proof}

%% file: semisorted_threshold_talagrand.tex
\section{Connecting Talagrand Objectives of $f$ and the Tracker Functions}\label{sec:mainworkhorse}

In this section and the next, we establish our main technical result~\Cref{thm:semisorted-reduce-to-g} relating the Talagrand objectives 
on the colorful thresholded influence of the hypergrid function $f:[n]^d \to \{0,1\}$ and the Talagrand objectives on the undirected influence of the 
tracker functions. We restate the theorem below for convenience.
\begin{mdframed}[backgroundcolor=gray!20,topline=false,bottomline=false,leftline=false,rightline=false] 
\semisorted*
\end{mdframed}
\noindent
To prove \Cref{thm:semisorted-reduce-to-g} we need to describe the coloring $\xi_\bx$ for each $\bx$ in $[n]^d$. We proceed doing so in $d$ stages. 
\begin{itemize}
    \item For every $i \in \{0,1,\ldots,d\}$ and for every $\bx \in [n]^d$, we 	define a {\em partial} edge coloring $\xi^{(i)}_\bx$ of the hypercube which assigns a $\{0,1\}$ value to every hypercube edge of the form $(T, T\oplus j)$ for all $j \leq i$, and for all $T\subseteq [i]$. The process will begin with the null coloring, $\xi^{(0)}_{\bx}$, and end with a complete coloring, $\xi_\bx := \xi^{(d)}_\bx$, for every $\bx \in [n]^d$. 
    \item For every $i \in \{0,1,\ldots,d\}$ and every $S \subseteq [i]$ we will also define a coloring $\chi_S^{(i)}$ of the edges of the augmented hypergrid. We start with $\chi^{(0)}_{\emptyset} := \chi$ where $\chi$ is the original coloring which, recall, is adversarially chosen.
\end{itemize}


For every $i \in \{0,1,\ldots,d\}$ and $S \subseteq [i]$ we will use the above colorings to define the $(i,S)$-\emph{hybrid Talagrand objective}
\begin{equation}\label{eq:rhs-quantity}
R_{i}(S) := \Exp_{\bx \in [n]^d} \sqrt{~\sum_{j=1}^{i} I^{=j}_{g_\bx, \xi^{(i)}_\bx} (S) ~~+~~ \sum_{j=i+1}^d \Phi_{S\circ f,\chi_S^{(i)}}(\bx; j) }\text{.} \tag{Colorful Hybrid}
\end{equation}
Recall that $S\circ f$ is the function obtained after sorting $f$ on the coordinates in $S$. Note that $R_i(S)$ is well-defined given the partial colorings $\xi_\bx^{(i)}$ for each $\bx \in [n]^d$ as defined above. Also
observe that since $\chi_\emptyset^{(0)} := \chi$, the arbitrary coloring specified in the theorem statement, we have that
$R_0(\emptyset)$ is precisely the LHS in the statement of \Cref{thm:semisorted-reduce-to-g}, that is, $R_0(\emptyset) = \Exp_{\bx \in [n]^d}~\left[\sqrt{\Phi_{f,\chi}(\bx)} \right]$. Additionally, since we use $\xi_\bx := \xi_\bx^{(d)}$, observe that $\Exp_{S\subseteq [d]}[R_d(S)]$ is precisely the RHS in the statement of \Cref{thm:semisorted-reduce-to-g}.

With the above setup in mind, we show that the following \Cref{lem:lhs2-and-lhs3} suffices to prove \Cref{thm:semisorted-reduce-to-g}.

\begin{mdframed}[backgroundcolor=blue!10,topline=false,bottomline=false,leftline=false,rightline=false] 
	\begin{lemma}[Potential Drop Lemma]\label{lem:lhs2-and-lhs3}
		Fix $i \in \{1,\ldots,d\}$, $\xi^{(i-1)}_\bx$ for all $\bx \in [n]^d$, and $\chi_S^{(i-1)}$ for every $S \subseteq [i-1]$, which all satisfy the specifications described in the previous paragraph.
		There exists a choice of $\xi^{(i)}_\bx$ for every $\bx\in [n]^d$ and $\chi_S^{(i)}$, $\chi_{S+i}^{(i)}$ for every $S \subseteq [i-1]$ all  satisfying the specifications described in the previous paragraph, such that for all $S \subseteq [i-1]$, we have (a) $R_{i-1}(S) \geq R_i(S)$ and (b) $R_{i-1}(S)\geq R_i(S+i)$.
	\end{lemma}
\end{mdframed}

\begin{proof}[\bf Proof of \Cref{thm:semisorted-reduce-to-g}:]
	Consider the following binary tree with $d+1$ levels. Each level $i \in \{0,1,\ldots,d\}$ has $2^i$ nodes indexed by subsets $S\subseteq [i]$.
	Every such node is associated with a coloring $\chi_S^{(i)}$ of the augmented hypergrid edges. The level $i$ is also associated with a partial coloring $\xi^{(i)}_\bx$ for every $\bx\in [n]^d$.
	
	The $0$'th level contains a single node indexed by $\emptyset$. The associated augmented hypergrid coloring
	is $\chi_\emptyset^{(0)} := \chi$. The partial coloring $\xi^{(0)}_\bx$ is null for all $\bx\in [n]^d$.
	We associate the value $R_0(\emptyset) = T_{\Phi_{\chi}}(f)$ with the root. 
	
	For $1\leq i\leq d$, we describe the children of each node in level $i-1$. Each node in level $i-1$ is indexed by some $S\subseteq [i-1]$. We associate this node with the value $R_{i-1}(S)$. This node has two children at level $i$: one, the left child, indexed by $S$ and the other, the right child, indexed by $S+i$. The coloring of the hypergrid edges at the left child is defined as $\chi_S^{(i)}$ from the lemma, and that of the hypergrid edges at the right child is defined as $\chi_{S+i}^{(i)}$ from the lemma. The left and right children hold the quantites $R_{i}(S)$ and $R_i(S+i)$, respectively. At level $i$, the partial coloring $\xi_\bx^{(i-1)}$ is also extended to $\xi_\bx^{(i)}$ for every $\bx \in [n]^d$ as stated in the lemma. From the lemma, we have $R_{i-1}(S) \geq R_i(S)$ and $R_{i-1}(S) \geq R_{i}(S+i)$. This immediately implies the following:
	\[
	\text{For all } i \in \{1,\ldots,d\} \text{, we have } \Exp_{S \subseteq [i-1]}[R_{i-1}(S)] \geq \Exp_{S \subseteq [i]}[R_{i}(S)]
	\]
	and chaining these $d$ inequalities together yields $R_0(\emptyset) \geq \Exp_{S \subseteq [d]}[R_d(S)]$.
	
	Now consider the leaf nodes of this tree, which hold the values $R_d(S)$ for every $S \subseteq [d]$. 
	Observe that $R_d(S) = \Exp_{\bx \in [n]^d} \left[\sqrt{I_{g_\bx, \xi_\bx}(S)}\right]$ since $\xi_\bx := \xi_\bx^{(d)}$. Recalling that $R_0(\emptyset) = T_{\Phi_{\chi}}(f)$ yields
	\[
	T_{\Phi_{\chi}}(f) = R_0(\emptyset) \geq \Exp_{S\subseteq [d]} [R_d(S)] = \Exp_{S\subseteq [d]}\Exp_{\bx \in [n]^d} \left[\sqrt{I_{g_\bx, \xi_\bx}(S)}\right]
	\]
and this establishes the claim after exchanging the expectations. \end{proof}

\section{Proof of Potential Drop~\Cref{lem:lhs2-and-lhs3}}\label{sec:proofoflemma6}

Recall $i \in \{1,\ldots,d\}$ is fixed. For brevity's sake, we will fix a set $S \subseteq [i-1]$ and call $h := (S\circ f)$. Let's refer to $\chi^{(i-1)}_S$ as simply $\chi$ without confusing with the original $\chi$ in the theorem.
The two colorings $\chi^{(i)}_S$ and $\chi^{(i)}_{S+i}$ that we construct will be simply called $\chi'$ and $\chi''$, respectively.
Let's call the partial colorings $\xi^{(i-1)}_\bx$ as simply $\xi_\bx$.
We will call the coloring $\xi^{(i)}_\bx$ which we need to construct simply $\xi'_\bx$ in the latter. Recall that $\xi_{\bx}$ is defined on all edges $(T,T\oplus j)$ for $T \subseteq [i-1]$ and $j \leq i-1$ and in order to prove the lemma we will need to define $\xi_\bx'$ on all edges $(T \oplus j)$ for $T \subseteq [i]$ and $j \leq i$. 

Fix an $i$-line $\ell$. We prove the lemma line-by-line. To be precise, let us consider the following vectors. First,
\begin{equation}\label{eq:rhs-vector-def}
	\vv{L}_{\ell} := \left(~~ \underbrace{\sum_{j=1}^{i-1} I_{g_\bx,\xi_\bx}^{=j}(S)}_{\vv{L^{(1)}}_{\ell}} ~+~ \underbrace{\Phi_{h,\chi}(\bx;i)}_{_{\vv{L^{(2)}}_{\ell}}} ~+~\underbrace{\sum_{j=i+1}^d \Phi_{h,\chi}(\bx;j)}_{\vv{L^{(3)}}_{\ell}}~~~~:~\bx\in \ell\right)
\end{equation}
Observe that
\begin{equation}\label{eq:rhs}
	R_{i-1}(S) = \frac{1}{n^d} \sum_{i\text{-lines } \ell}\norm{\vv{L}_{\ell}}_{1/2} = \frac{1}{n^d} \sum_{i\text{-lines } \ell} \norm{\vv{L^{(1)}}_{\ell} + \vv{L^{(2)}}_{\ell} + \vv{L^{(3)}}_{\ell}}_{1/2}
\end{equation}
where, recall, we are (ab)using the notation $\norm{v}_{1/2} := \sum_i \sqrt{v_i}$.

\noindent
Define
\begin{equation}\label{eq:lhs2-vector-def}
	\vv{R}_{\ell} := \left(~~\underbrace{\sum_{j=1}^{i-1} I_{g_\bx,\color{red} \xi'_\bx}^{=j}(S)}_{\vv{R^{(1)}}_{\ell}} ~+~ \underbrace{ I_{g_\bx,{\color{red} \xi'_\bx}}^{=i}(S)}_{\vv{R^{(2)}}_{\ell}} ~+~
	\underbrace{\sum_{j=i+1}^d \Phi_{h,\color{red} \chi'}(\bx;j)}_{\vv{R^{(3)}}_{\ell}}~   ~~~:~\bx\in \ell\right)
\end{equation}
where we have denoted, in red, the recolorings that we need to define.
The ``first'' RHS term is
\begin{equation}\label{eq:rhs}
	R_{i}(S) := \frac{1}{n^d} \sum_{i\text{-lines } \ell} \norm{\vv{R}_{\ell}}_{1/2} = \frac{1}{n^d} \sum_{i\text{-lines } \ell} \norm{\vv{R^{(1)}}_{\ell} + \vv{R^{(2)}}_{\ell} + \vv{R^{(3)}}_{\ell}}_{1/2}
\end{equation}
Similarly, define
\begin{equation}\label{eq:lhs3-vector-def}
	\vv{M}_{\ell} := \left(~~\underbrace{\sum_{j=1}^{i-1} I^{=j}_{g_\bx,\color{red} \xi'_\bx}(S+i)}_{\vv{M^{(1)}}_{\ell}} ~+~ \underbrace{I^{=i}_{g_\bx,{\color{red} \xi'_\bx}}(S+i)}_{\vv{M^{(2)}}_{\ell}} ~+~
	\underbrace{\sum_{j=i+1}^d \Phi_{i\circ h,\color{red} \chi''}(\bx;j)}_{\vv{M^{(3)}}_{\ell}}~   ~~~:~\bx\in \ell\right)
\end{equation}
and notice that the ``second'' RHS term is
\begin{equation}\label{eq:rhs}
	R_{i}(S+i) := \frac{1}{n^d} \sum_{i\text{-lines } \ell}\norm{\vv{M}_{\ell}}_{1/2} = \frac{1}{n^d} \sum_{i\text{-lines } \ell} \norm{\vv{M^{(1)}}_{\ell} + \vv{M^{(2)}}_{\ell} + \vv{M^{(3)}}_{\ell}}_{1/2}
\end{equation}

Observe now that it suffices to prove that there exists colorings $\chi', \chi''$, and $\xi'_\bx$'s such that $\norm{\vv{L}_{\ell}}_{1/2} \geq \norm{\vv{R}_{\ell}}_{1/2}$ and $\norm{\vv{L}_{\ell}}_{1/2} \geq \norm{\vv{M}_{\ell}}_{1/2}$ for all $i$-lines $\ell$. Thus, we now fix an $i$-line $\ell$ and drop the subscript, $\ell$, from all the previously defined vectors for brevity. We define $\LHS := \norm{\vv{L}}_{1/2}$, $\RHS_1 := \norm{\vv{R}}_{1/2}$, $\RHS_2 := \norm{\vv{M}}_{1/2}$, and set out to prove $\LHS \geq \RHS_1$ and $\LHS \geq \RHS_2$.



\paragraph{A Picture of the Line.}
Since $h$ is semisorted, the picture of $h$ restricted to $\ell$ looks like this. The green zone is where the function is $1$. 
Without loss of generality we assume $\ell$ has more ones than zeros.
We use $A$ to denote the ones on the left and $C$ to denote the zeros on the right. We use $k := |C|$, and $B\subseteq A$ are the $k$ right most ones
in the left side.
\begin{figure}[h!]
	\includegraphics*[trim = 0 320 0 100, clip, scale = 0.5]{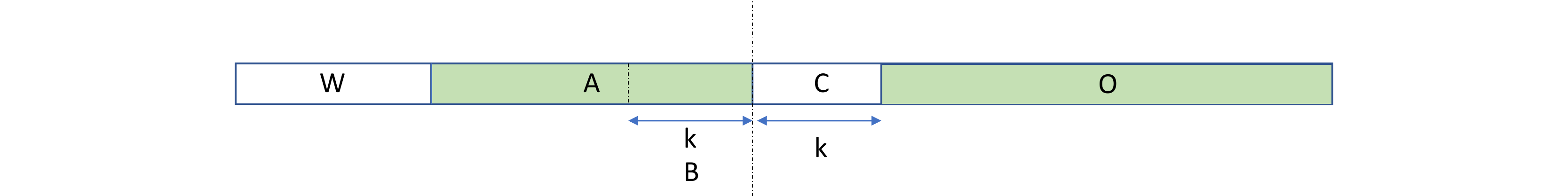}
\end{figure}
Throughout, we will use the notation $\vv{A}_X$ to denote the sub-vector of $\vv{A}$ defined on $\ell$ with coordinates restricted to $x\in X$; we will always use this notation when $X$ is a contiguous interval.  Indeed, these $X$'s will be always picked from $\{W, A, C, O, B, A\setminus B\}$ or unions of these, always making sure they form a contiguous interval. 

\paragraph{High Level Idea.} Before we venture into proving the inequalities, we would like to remind the reader again
of the proof strategy discussed in~\Cref{sec:main-ideas}.  We need to define the colorings $\chi'$, $\chi''$, and 
also $\xi_\bx^{(i)}$'s such that the objective after recoloring satisfy the inequality we desire to prove. This going to hinge upon showing that the vector 
obtained after operation either majorizes or is coordinate-wise dominated by a vector that majorizes the vector before the operation. 
In particular, these are the conditions (a)-(d) and (e)-(h) mentioned below in the grey boxes. 
To show these properties, we would be crucially using the property that the function $f$ is semi-sorted which leads to certain monotonicity properties that allows us to claim them. In particular, we would be using~\Cref{lem:sum-of-vectors} when establishing almost all the conditions mentioned above.
There is a certain sense of repetition in which these arguments are made, however, we have provided all the details for completeness.

\subsection{Proving $\LHS \geq \RHS_1$} 
During the proof of $\LHS \geq \RHS_1$, we will define the coloring $\chi'$ on all edges of the fully augmented hypergrid and $\xi'_\bx(S, S\oplus j)$ where $j\leq i$ for all $\bx \in [n]^d$.
We will not specify $\xi'_\bx(S + i, S + i \oplus j)$ since these won't be needed to prove this inequality; we will describe them when we prove $\LHS \geq \RHS_2$.

Before we describe the recolorings, it is useful to describe the plan of the proof. This will motivate why we recolor as we do.
We will actually consider 
\[
\LHS = \norm{\vv{L}_{W}}_{1/2} + \norm{\vv{L}_{A}}_{1/2} + \norm{\vv{L}_C}_{1/2} + \norm{\vv{L}_{O}}_{1/2} \]
and
\[
\RHS_1 = \norm{\vv{R}_{W}}_{1/2} + \norm{\vv{R}_{A}}_{1/2} + \norm{\vv{R}_C}_{1/2} + \norm{\vv{R}_{O}}_{1/2}
\] 
and argue domination term-by-term.

More precisely, we find recolorings $\chi', \xi'$ such that
\begin{mdframed}[backgroundcolor=gray!20,topline=false,bottomline=false,leftline=false,rightline=false]
	\begin{center}
		\begin{enumerate}
			\item[(a)] $\vv{R^{(q)}_A} \majorizes \sortdown{\vv{L^{(q)}_A}}$ and $\vv{R^{(q)}_O} \majorizes \sortdown{\vv{L^{(q)}_O}}$, for $q\in \{1,3\}$, 
				\item[(b)] $\exists \vv{L'^{(2)}_A}$ such that $\vv{L'^{(2)}_A} \majorizes \sortdown{\vv{L^{(2)}_A}}$ and $\vv{L'^{(2)}_A} \coordom \vv{R^{(2)}_A}$,
			\item[(c)] $\vv{R^{(q)}_W} \majorizes \sortup{\vv{L^{(q)}_W}}$ and $\vv{R^{(q)}_C} \majorizes \sortup{\vv{L^{(q)}_C}}$, for $q\in \{1,3\}$, 
		\item[(d)] $\exists \vv{L'^{(2)}_C}$ such that $\vv{L'^{(2)}_C} \majorizes \sortup{\vv{L^{(2)}_C}}$ and $\vv{L'^{(2)}_C} \coordom \vv{R^{(2)}_C}$.
		\end{enumerate}
	\end{center}
\end{mdframed}
Let us see why the above conditions suffice to prove the inequality. 
The second part of (b) implies that $\norm{\vv{R_A}}_{1/2} \leq \norm{\vv{R^{(1)}_A} + \vv{R^{(3)}_A} + \vv{L'^{(2)}_A}}_{1/2}$.
Part (a) and the first part of (b), along with~\Cref{lem:sum-of-vectors}, implies $\vv{R^{(1)}_A} + \vv{R^{(3)}_A} + \vv{L'^{(2)}_A} \majorizes \sortdown{\vv{L_A}}$.
And so, $\norm{L_A}_{1/2} \geq \norm{R_A}_{1/2}$. A similar argument using (c) and (d) implies $\norm{L_C}_{1/2} \geq \norm{R_C}_{1/2}$.

One last observation is needed to complete the proof. Note that $R^{(2)}_W$ is the {\bf zero} vector: the points $x\in W$ don't change value even when $\ell$ is sorted.
Also note that $L^{(2)}_W$ is the zero vector; the points $x\in W$ don't participate in a violation in direction $i$. 
And therefore, part (c) along with~\Cref{lem:sum-of-vectors} implies $\vv{R_W} \majorizes \sortup{\vv{L_W}}$ implying $\norm{L_W}_{1/2} \geq \norm{R_W}_{1/2}$.
Similarly, $R^{(2)}_O \equiv L^{(2)}_O \equiv \mathbf{0}$, and thus part (a) along with~\Cref{lem:sum-of-vectors} implies $\norm{L_O}_{1/2} \geq \norm{R_O}_{1/2}$.

\subsubsection{Proving (a) and (c) for $q=3$} \label{sec:ac3}

\paragraph{Defining the Coloring $\chi'$:}

We will now describe the coloring $\chi'$ on all edges of the form $(\bx,\bx + a\be_j)$ where $j \geq i+1$, $h(\bx) = 1$ and $h(\bx+a\be_j) = 0$. For all other edges $e$, we simply define $\chi'(e) = \chi(e)$ as these edges do not play a role in proving the inequality.

Given a pair of $i$-lines $\ell$ and $\ell' = \ell + a\be_j$ for $j\geq i+1$ and $a > 0$, we consider the set of violations from $\ell$ to $\ell'$ in $h$: 
\begin{align} \label{eq:violations}
    V := \{(\bx,\bx+a\be_j) \colon \bx \in \ell \text{, } h(\bx) = 1\text{, and } h(\bx+a\be_j) = 0\} \text{.}
\end{align}
Since $h$ is semi-sorted, it's clear that we can write $V = V_L \cup V_R$ as a union of two intervals, in the sense that $\{\bx \colon (\bx,\bx+a\be_j) \in V_L\}$ is an interval in the lower half of $\ell$ and $\{\bx \colon (\bx,\bx+a\be_j) \in V_R\}$ is an interval in the upper half of $\ell$. Similarly, the upper endpoints form two intervals in $\ell'$. We then obtain $\chi'$ by down-sorting $\chi$ on each of these intervals, moving left-to-right:
\[
(\chi'(e) \colon e \in V_L) = \sortdown{\chi(e) \colon e \in V_L} \text{ and } (\chi'(e) \colon e \in V_R) = \sortdown{\chi(e) \colon e \in V_R}\text{.}
\]
We provide the following illustration for clarity. The white and green intervals represent where $h = 0$ and $h=1$, respectively. The vertical arrows represent violated edges. Blue edges have color $0$ and red edges have color $1$. The left picture depicts the original coloring, $\chi$, and the right picture depicts the recoloring $\chi'$. 

\hspace*{-1cm}\includegraphics*[clip, scale = .5, trim = 0 300 0 40]{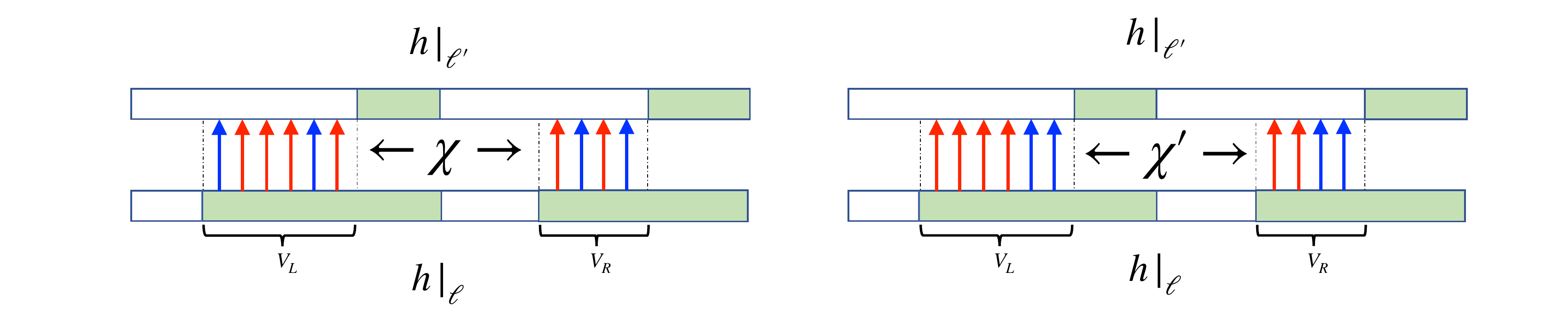}

We now return to our fixed $i$-line $\ell$ and set out to prove parts (a) and (c) for $q=3$, given this coloring $\chi'$. Let's recall our illustration of $\ell$ and our definition of the intervals $W,A,C,O$. 
\begin{figure}[h!]
	\includegraphics*[trim = 20 320 0 100, clip, scale = 0.5]{figs/semisorted-k-try2}
\end{figure}




\paragraph{Proving (a) for $q=3$:} 

Fix $j \geq i+1$ and a $i$-line $\ell' := \ell + a\be_j$. Let $A' := \{\bx\in A~:~ h(\bx+a\be_j) = 0\}$ and $O' := \{\bx\in O~:~ h(\bx+a\be_j) = 0\}$. 
Since $h$ is semi-sorted, it is not hard to see that $A'$ and $O'$ are prefixes of $A$ and $O$, respectively. 
\begin{claim}
	If $\bx_i < \bx'_i$ in $A$ such that $\bx' \in A'$, then $\bx\in A'$. The same is true for $O$ and $O'$.
\end{claim}
\begin{proof}
	Since $h$ is semisorted, $h(\bx' + a\be_j) = 0$ implies $h(\bx+a\be_j) = 0$.
\end{proof}
Moreover, observe that our definition of $\chi'$ gives us
\[
(\chi'(\bx,\bx+a\be_j) \colon \bx \in A') = \sortdown{\chi(\bx,\bx+a\be_j) \colon \bx \in A'}
\]
and
\[
(\chi'(\bx,\bx+a\be_j) \colon \bx \in O') = \sortdown{\chi(\bx,\bx+a\be_j) \colon \bx \in O'} \text{.}
\]

Let's investigate what this leads to. These are key properties.
\begin{definition}
	Fix $j \geq i + 1$ and fix an $i$-line $\ell' := \ell + a\be_j$ for $a > 0$. Define the following two boolean vectors
	\[
	\bv^{R}_{j,a} := \left(\bone{\left(h(\bx + a\be_j) = 0 ~~\textbf{and}~~~\chi'(\bx, \bx+a\be_j) = 1\right) }~~:~\bx\in A \right)
	\]
	and 
	\[
	\bv^{L}_{j,a} := \left(\bone{\left(h(\bx + a\be_j) = 0 ~~\textbf{and}~~~\chi(\bx, \bx+a\be_j) = 1\right) }~~:~\bx\in A\right)
	\]
	
\end{definition}
\noindent
Observe, for $\bx\in A$, 
\begin{equation}\label{eq:L3R3-piece1}
	\Phi_{h,\chi'}(\bx;j) = \min\left(1, \sum_a \bv^{R}_{j,a}(\bx)\right) ~~\textrm{and}~~~	\Phi_{h,\chi}(\bx; j) = \min\left(1, \sum_a \bv^{L}_{j,a}(\bx) \right)
\end{equation}

\begin{claim}
	Fix a $j \geq i + 1$ and $a > 0$. For any two $\bx_i < \bx'_i$ in $A$, we have $\bv^{R}_{j,a}(\bx) \geq \bv^{R}_{j,a}(\bx')$.
	That is, the vector $\bv^{R}_{j,a}$ is sorted decreasing.
\end{claim}
\begin{proof}	
	Since $h$ is semisorted $h(\bx' + a\be_j) = 0$ implies $h(\bx + a\be_j) = 0$.
	Furthermore, since both these are violations, by design $\chi'(\bx', \bx'+a\be_j) = 1$ implies $\chi'(\bx, \bx+a\be_j) = 1$. 
\end{proof}
\begin{claim}\label{clm:L2-permutation}
	Fix a $j \geq i + 1$ and $a > 0$. The vectors $\bv^{R}_{j,a}$ and $\bv^{L}_{j,a}$ are permutations of one another.
\end{claim}
\begin{proof}
	This is precisely how $\chi'$ is defined: it only permutes the colorings on the violations incident on $A$.
\end{proof}
\noindent
In conclusion, using the observation~\eqref{eq:L3R3-piece1}, we conclude that we can write 
\[
\vv{L^{(3)}_A} = \left(\sum_{j=i+1}^d 	\Phi_{h,\chi}(\bx;j)~~:~~\bx \in A\right)
\]
as a weighted sum of Boolean vectors, and the above two claims imply that the vector
\[
\vv{R^{(3)}_A} = \left(\sum_{j=i+1}^d 	\Phi_{h,\chi'}(\bx;j)~~:~~\bx\in A\right)
\]
is the same weighted sum of the {\em sorted decreasing} orders of those Boolean vectors.
Therefore, we can conclude using~\Cref{lem:sum-of-vectors}, 
\begin{mdframed}[backgroundcolor=gray!20,topline=false,bottomline=false,leftline=false,rightline=false]
	\begin{equation}\label{eq:L3R3A}
		\vv{R^{(3)}_A} \majorizes \sortdown{\vv{L^{(3)}_A}}
	\end{equation}
\end{mdframed}
An absolutely analogous argument with $O$'s replacing $A$'s gives us
\begin{mdframed}[backgroundcolor=gray!20,topline=false,bottomline=false,leftline=false,rightline=false]
	\begin{equation}\label{eq:L3R3O}
		\vv{R^{(3)}_O} \majorizes \sortdown{\vv{L^{(3)}_O	}}
	\end{equation}
\end{mdframed}

\paragraph{Proving (c) for $q=3$:}

The picture is similar, but reversed, when we consider the points in $W\cup C$, where $h(\bx) = 0$. Recall the definition of $W$ and $C$ as in the illustration. Fix $j \geq i+1$ and a $i$-line $\ell'' := \ell - a\be_j$. Let $W' := \{\bx\in W~:~ h(\bx-a\be_j) = 1\}$ and $C' := \{\bx\in C~:~ h(\bx-a\be_j) = 1\}$. 
It is not hard to see that $W'$ and $C'$ are suffixes of $W$ and $C$, respectively. 
\begin{claim}
	If $\bx_i < \bx'_i$ in $W$ such that $\bx \in W'$, then $\bx'\in W'$. The same is true for $C$ and $C'$.
\end{claim}
\begin{proof}
	Since $h$ is semisorted, $h(\bx - a\be_j) = 1$ implies $h(\bx'-a\be_j) = 1$.
\end{proof}
Again, observe that our definition of $\chi'$ gives us
\[
(\chi'(\bx-a\be_j,\bx) \colon \bx \in W') = \sortdown{\chi(\bx-a\be_j,\bx) \colon \bx \in W'}
\]
and
\[
(\chi'(\bx-a\be_j,\bx) \colon \bx \in C') = \sortdown{\chi(\bx-a\be_j,\bx) \colon \bx \in C'}
\]

\begin{definition}
	Fix $j \geq i + 1$ and fix an $i$-line $\ell'' := \ell - a\be_j$ for $a > 0$. Define the following two boolean vectors
	\[
	\bv^{R}_{j,a} := \left(\bone{\left(h(\bx - a\be_j) = 1 ~~\textbf{and}~~~\chi'(\bx-a\be_j, \bx) = 0\right)}~~:~\bx\in C\right)
	\]
	and 
	\[
	\bv^{L}_{j,a} := \left(\bone{\left(h(\bx - a\be_j) = 1 ~~\textbf{and}~~~\chi(\bx-a\be_j, \bx) = 0\right) }~~:~\bx\in C\right)
	\]
\end{definition}
Observe, for $\bx\in C$, 
\begin{equation}
	\Phi_{h,\chi'}(\bx;j) = \min\left(1, \sum_a \bv^{R}_{j,a}(\bx)\right) ~~\textrm{and}~~~	\Phi_{h,\chi}(\bx;j) = \min\left(1, \sum_a \bv^{L}_{j,a}(\bx) \right)
\end{equation}

\begin{claim}
	Fix a $j \geq i + 1$ and $a > 0$. For any two $\bx_i > \bx'_i$ in $C$, we have $\bv^{R}_{j,a}(\bx) \geq \bv^{R}_{j,a}(\bx')$.
	That is, the vector $\bv^{R}_{j,a}$ is sorted increasing when considered left to right.
\end{claim}
\begin{proof}	
	Since $h$ is semisorted $h(\bx' - a\be_j) = 1$ implies $h(\bx - a\be_j) = 1$.
	Furthermore, since both these are violations, by design $\chi'(\bx', \bx'+a\be_j) = 0$ implies $\chi'(\bx, \bx+a\be_j) = 0$. 
\end{proof}
\begin{claim}\label{clm:L2-permutation}
	Fix a $j \geq i + 1$ and $a > 0$. The vectors $\bv^{R}_{j,a}$ and $\bv^{L}_{j,a}$ are permutations of one another.
\end{claim}
A similar argument to the one given above now implies 
$\vv{L^{(3)}_C}$ 
is a sum of Boolean vectors, and 
$\vv{R^{(3)}_C}$ 
is the sum of the {\em sorted increasing} orders of those Boolean vectors.
Using~\Cref{lem:sum-of-vectors}, we can conclude
%
\begin{mdframed}[backgroundcolor=gray!20,topline=false,bottomline=false,leftline=false,rightline=false]
	\begin{equation}\label{eq:L3R3C}
		\vv{R^{(3)}_C} \majorizes \sortup{\vv{L^{(3)}_C}}
	\end{equation}
\end{mdframed}
And an absolutely analogous argument gives
\begin{mdframed}[backgroundcolor=gray!20,topline=false,bottomline=false,leftline=false,rightline=false]
	\begin{equation}\label{eq:L3R3W}
		\vv{R^{(3)}_W} \majorizes \sortup{\vv{L^{(3)}_W}}
	\end{equation}
\end{mdframed}
This finishes the proofs of $q=3$ for (a) and (c). \medskip

\subsubsection{Proving (a) and (c) for $q=1$}

\paragraph{Defining $\xi'_{\bx}(S,S\oplus j)$ for $S \subseteq [i-1]$ and $j \leq i-1$:}

We now define the partial coloring $\xi'_{\bx} := \xi^{(i)}_{\bx}$ on all edges $(S,S\oplus j)$ where $S \subseteq [i-1]$ and $j\leq i-1$ for all $\bx \in [n]^d$. These are exactly the relevant edges for the proof of parts (a) and (c) for $q=1$. Note that the partial coloring $\xi_{\bx} := \xi_{\bx}^{(i-1)}$ is defined over precisely these edges for each $\bx \in [n]^d$. The color of $\xi'_{\bx}$ on the edges $(S,S+i)$ for $S \subseteq [i-1]$ will be defined when we prove parts (b) and (d). The color of $\xi'_{\bx}$ on the edges $(S+i,S+i \oplus j)$ for $S \subseteq [i-1]$ and $j \leq i-1$ will be defined when we prove $\LHS \geq \RHS_2$.

Fix $j \leq i-1$, $S \subseteq [i-1]$, and a $i$-line $\ell$. We consider the set of $\bx \in \ell$ such that $(S,S\oplus j)$ is influential in $g_{\bx}$:
\begin{align} \label{eq:V_gx}
    V := \left\{\bx \in \ell \colon g_{\bx}(S) = 1 \text{ and } g_{\bx}(S \oplus j) = 0 \right\} \text{.}
\end{align}

Note that since $f$ is semi-sorted, we have that $(S \circ f)$ and $(S \oplus j \circ f)$ are both semi-sorted. 
Thus, we can write $V = V_L \cup V_R$ where $V_L$ and $V_R$ are intervals contained in the left and right half of $\ell$, respectively. We again obtain $\xi'_{\bx}$ by down-sorting the original coloring on these intervals:
\[
    (\xi'_{\bx}(S,S\oplus j) \colon \bx \in V_L) = \sortdown{\xi_{\bx}(S,S\oplus j) \colon \bx \in V_L}    
\]
and similarly
\[
    (\xi'_{\bx}(S,S\oplus j) \colon \bx \in V_R) = \sortdown{\xi_{\bx}(S,S\oplus j) \colon \bx \in V_R} \text{.}
\]
For all $\bx \in \ell \setminus V$, we define $\xi'_{\bx}(S,S\oplus j) := \xi_{\bx}(S,S\oplus j)$. This completely describes $\xi'_{\bx}(S,S\oplus j)$ for every $\bx \in [n]^d$. 

We provide the following illustration for clarity. Note that the picture is quite similar to the one provided in \Cref{sec:ac3}, when we defined $\chi'$. The key difference is that the bottom and top segments represent the same line $\ell$, but with different functions $S \circ f$ and $(S \oplus j) \circ f$, respectively. The vertical lines are no longer arrows to emphasize that they represent \emph{undirected edges in the hypercube} as opposed to directed edges in the augmented hypergrid.

\hspace*{-1cm}\includegraphics*[clip, scale = .5, trim = 0 300 0 20]{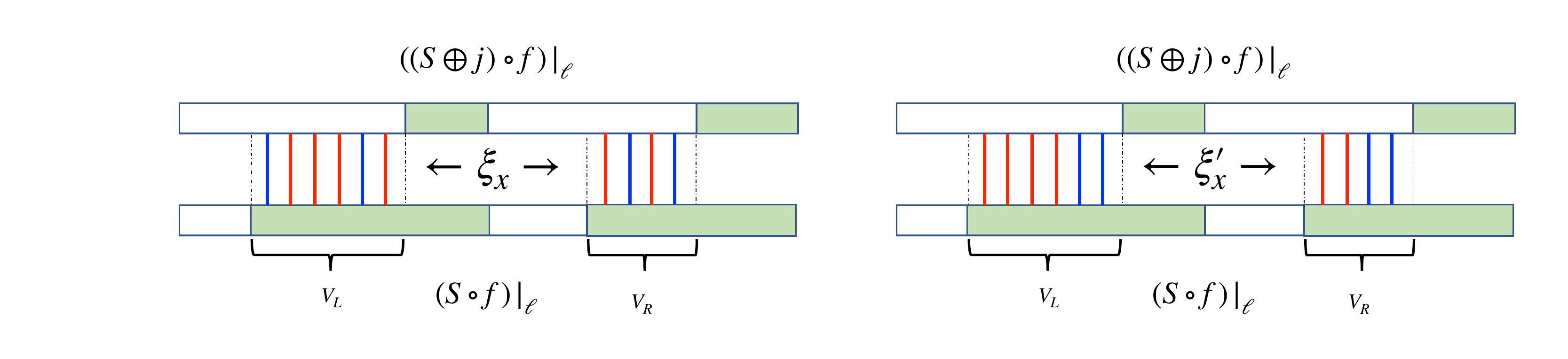}

We now return to our fixed $i$-line $\ell$ and set out to prove parts (a) and (c) for $q=1$, given the colorings $\xi'_{\bx}$. Let's recall our illustration of $\ell$ and our definition of the intervals $W,A,C,O$. Recall that $g_{\bx} = h(\bx)$ and so the definition of these intervals is the same.
\begin{figure}[h!]
	\includegraphics*[trim = 0 320 0 100, clip, scale = 0.5]{figs/semisorted-k-try2}
\end{figure}
\paragraph{Proof of Part (a) for $q=1$:}


Fix $j \leq i-1$ and let $A' = \{\bx \in A \colon g_{\bx}(S\oplus j) = 0\}$ and $O' = \{\bx \in O \colon g_{\bx}(S\oplus j) = 0\}$, which are prefixes of $A$ and $O$, respectively. From our definition of $\xi'_{\bx}(S,S\oplus j)$ above, we have
\[
(\xi'_{\bx}(S,S\oplus j) \colon \bx \in A') = \sortdown{\xi_{\bx}(S,S\oplus j) \colon \bx \in A'}
\]
and similarly
\[
(\xi'_{\bx}(S,S\oplus j) \colon \bx \in O') = \sortdown{\xi_{\bx}(S,S\oplus j) \colon \bx \in O'} \text{.}
\]


\begin{claim} \label{clm:A_sorted}
	$\left(I_{g_\bx,\xi'_\bx}^{=j}(S)~~:~~ \bx\in A\right)$ is a sorted decreasing vector, and is a permutation of 	$\left(I_{g_\bx,\xi_\bx}^{=j}(S)~~:~~ \bx\in A\right)$.
\end{claim}
\begin{proof}
	Take $\bx_i < \bx'_i$ in $A$. Note that $g_\bx(S) = 1$ for both $\bx, \bx'$. Thus,
	\[
	I_{g_\bx,\xi'_\bx}^{=j}(S) = \bone\left(g_\bx(S\oplus j) = 0 ~\textbf{and}~\xi'_\bx(S,S\oplus j) = 1 \right)
	\]
	and
	\[
	I_{g_\bx,\xi_\bx}^{=j}(S) = \bone\left(g_\bx(S\oplus j) = 0 ~\textbf{and}~\xi_\bx(S,S\oplus j) = 1 \right)
	\]
	The two vectors are Boolean vectors with number of ones equal to the number of ones in $(\xi_\bx(S, S\oplus j) ~:~\bx\in A')$ which equals 
	the number of ones in $(\xi'_\bx(S, S\oplus j)~:~\bx\in A')$. Thus, they are permutations.
	By design of $\xi'_\bx$'s, this vector is sorted decreasing on $A'$, and all zeros in $A\setminus A'$ (which come to the right of $A'$).
\end{proof}
Observing that 
\[
\vv{L^{(1)}_A} = \left(\sum_{j=1}^{i-1} I_{g_\bx,\xi_\bx}^{=j}(S)~:~\bx\in A\right) ~~\textrm{and}~~\vv{R^{(1)}_A} = \left(\sum_{j=1}^{i-1} I_{g_\bx,\xi'_\bx}^{=j}(S)~:~\bx\in A\right)
\]
using~\Cref{lem:sum-of-vectors} and the claim above, we get

\begin{mdframed}[backgroundcolor=gray!20,topline=false,bottomline=false,leftline=false,rightline=false]
	\begin{equation}\label{eq:L1R1A}
		\vv{R^{(1)}_A} \majorizes \sortdown{\vv{L^{(1)}_A}}
	\end{equation}
\end{mdframed}
Absolutely analogously, we get
\begin{mdframed}[backgroundcolor=gray!20,topline=false,bottomline=false,leftline=false,rightline=false]
	\begin{equation}\label{eq:L1R1O}
		\vv{R^{(1)}_O} \majorizes \sortdown{\vv{L^{(1)}_O}}
	\end{equation}
\end{mdframed}

\paragraph{Proof of Part (c) for $q=1$:}

The picture is similar, but reversed when we consider the points in $W \cup C$, where $g_{\bx}(S) = 0$. Fix $j \leq i-1$ and define $W' := \{\bx \in W \colon g_{\bx}(S \oplus j) = 1\}$ and $C' := \{\bx \in C \colon g_{\bx}(S \oplus j) = 1\}$ which are suffixes of $W$ and $C$, respectively. From our definition of $\xi'_{\bx}(S,S\oplus j)$ above, made from the perspective of the set $S \oplus j$, 
we have
\[
(\xi'_{\bx}(S,S\oplus j) \colon \bx \in W') = \sortdown{\xi_{\bx}(S,S\oplus j) \colon \bx \in W'}
\]
and similarly
\[
(\xi'_{\bx}(S,S\oplus j) \colon \bx \in C') = \sortdown{\xi_{\bx}(S,S\oplus j) \colon \bx \in C'} \text{.}
\]


Analogous to \Cref{clm:A_sorted}, we have the following claim.
\begin{claim}
	$\left(I_{g_\bx,\xi'_\bx}^{=j}(S)~~:~~ \bx\in W\right)$ is a sorted increasing vector, and is a permutation of 	$\left(I_{g_\bx,\xi_\bx}^{=j}(S)~~:~~ \bx\in W\right)$.
\end{claim}
Arguing similarly to the proof of \Cref{eq:L1R1A} we get
\begin{mdframed}[backgroundcolor=gray!20,topline=false,bottomline=false,leftline=false,rightline=false]
	\begin{equation}\label{eq:L1R1W}
		\vv{R^{(1)}_W} \majorizes \sortup{\vv{L^{(1)}_W}}
	\end{equation}
\end{mdframed}
and absolutely analogously, we get
\begin{mdframed}[backgroundcolor=gray!20,topline=false,bottomline=false,leftline=false,rightline=false]
	\begin{equation}\label{eq:L1R1C}
		\vv{R^{(1)}_C} \majorizes \sortup{\vv{L^{(1)}_C}}
	\end{equation}
\end{mdframed}
\eqref{eq:L3R3A}, \eqref{eq:L3R3O}, \eqref{eq:L3R3C}, \eqref{eq:L3R3W}, and \eqref{eq:L1R1A}, \eqref{eq:L1R1O}, \eqref{eq:L1R1W}, \eqref{eq:L1R1C}
establish (a) and (c). \medskip

\subsubsection{Proving (b) and (d):}

Finally, we need to establish (b) and (d). Let us recall these and also draw the picture of $\ell$ that we have been using.
\begin{itemize}
	\item[(b)] $\exists \vv{L'^{(2)}_A}$ such that $\vv{L'^{(2)}_A} \majorizes \sortdown{\vv{L^{(2)}_A}}$ and $\vv{L'^{(2)}_A} \coordom \vv{R^{(2)}_A}$.
	\item[(d)] $\exists \vv{L'^{(2)}_C}$ such that $\vv{L'^{(2)}_C} \majorizes \sortup{\vv{L^{(2)}_C}}$ and $\vv{L'^{(2)}_C} \coordom \vv{R^{(2)}_C}$.
\end{itemize}
\begin{figure}[h!]
	\includegraphics*[trim = 0 320 0 100, clip, scale = 0.5]{figs/semisorted-k-try2}
\end{figure}
\noindent
We remind the reader that 
$\vv{L^{(2)}}(\bx) = \Phi_{h,\chi}(\bx;i)$ for all $\bx\in \ell$. 
We begin with an observation which strongly uses the ``thresholded'' nature of the definition of $\Phi$.
\begin{claim}
	No matter how $\chi$ is defined, either $\vv{L^{(2)}_A}$ is the all $1$s vector, or $\vv{L^{(2)}_C}$ is the all $1$s vector. 
\end{claim}
\begin{proof}
	Suppose for the sake of contradiction, there exists $\bx \in A$ and $\by \in C$ such that $\Phi_{h,\chi}(\bx;i) = \Phi_{h,\chi}(\by;i) = 0$.
	But the edge $(\bx, \by)$ is a violation, and if $\chi(\bx, \by) = 1$ then $\Phi_{h,\chi}(\bx;i) = 1$, otherwise $\Phi_{h,\chi}(\by;i) = 1$. Contradiction.
\end{proof}
Next we remind the reader that $\vv{R^{(2)}}(\bx) = I^{=i}_{g_\bx, \xi'_\bx}(S)$. 
We now define the $\xi'_\bx(S, S+i)$ colorings for $\bx \in A\cup C$ using the above claim in the following simple manner.
\begin{equation}\label{eq:xi1}
\textrm{If}~~\vv{L^{(2)}_A} \equiv \bone, ~~\textrm{then}~~ \xi'_\bx(S, S+i) = 1~~\forall \bx \in A \cup C
\end{equation}
otherwise, 
\begin{equation}\label{eq:xi2}
\textrm{we have}~~\vv{L^{(2)}_C} \equiv \bone, ~~\textrm{and so we define}~~ \xi'_\bx(S, S+i) = 0~~\forall \bx \in A \cup C
\end{equation}
In the former case, we have $\vv{R^{(2)}_A} = (\underbrace{111\cdots1}_{k~\text{many}}0000)$ and $\vv{L^{(2)}_A} \equiv \bone$ and so we pick $\vv{L'^{(2)}_A} = \vv{L^{(2)}_A}$.
Also note that we have $\vv{R^{(2)}_C}$ as the all zeros vector, and so we pick $\vv{L'^{(2)}_C} = \sortup{\vv{L^{(2)}_C}}$. These satisfy (b) and (d).
In the latter case the argument is analogous. Thus, in either case we have established (b) and (d), and this completes the proof of $\LHS \geq \RHS_1$. \medskip

We remind the reader that we have now defined $\xi'_\bx(S, S\oplus j)$ for all subsets $S\subseteq [i-1]$ and $1\leq j\leq i$. In the next subsection, when we prove $\LHS \geq \RHS_2$, we will need to define $\xi'_\bx(S+i, S+i\oplus j)$ for all $j \leq i-1$. Note that for $j=i$, we have $(S+i, S+i\oplus j) = (S+i,S)$ and the coloring $\xi_{\bx}'$ has already been defined for these edges in \eqref{eq:xi1} or \eqref{eq:xi2}.


\subsection{Proving $\LHS \geq \RHS_2$} This inequality is a bit trickier to establish because the function $h$ itself now changes to $i\circ h$ in $\RHS_2$.
For instance, focusing on the illustration we have been using, upon sorting the picture looks like this.

\begin{figure}[ht!]
	\includegraphics*[trim = 0 210 0 100, clip, scale = 0.5]{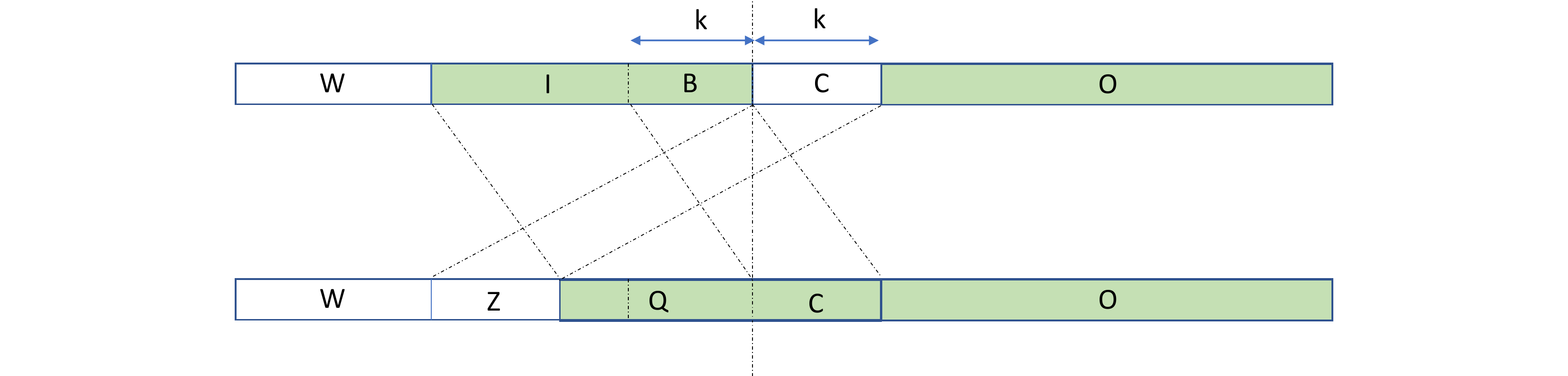}
\end{figure}

We have now partitioned the interval $A$ into $I \cup B$ where $B$ is the $k$-ones closest to the semi-sorting boundary.
After sorting, we think of the ones in $B$ moving into $C$, and the ones in $I$ shifting and moving to $Q \subseteq A$. 
The first $k$ entries of $A$, which we call $Z$, takes the value $0$ after sorting this line. 

To argue $\LHS \geq \RHS_2$, we break the vector $\vv{L}$ as
\[
\norm{\vv{L}}_{1/2} = \norm{\vv{L}_W}_{1/2} + \norm{\vv{L}_{I \cup B\cup O}}_{1/2} + \norm{\vv{L}_C}_{1/2} 
\]
and the vector $\vv{M}$ as 
\[
\norm{\vv{M}}_{1/2} = \norm{\vv{M_W}}_{1/2} + \norm{\vv{M_{Q\cup C\cup O}}}_{1/2} + \norm{\vv{M_Z}}_{1/2} 
\]
and argue vector-by-vector. The plan of the proof is similar to the previous case. We want to find recolorings $\chi''$ and $\xi'_\bx$ such that

\begin{mdframed}[backgroundcolor=gray!20,topline=false,bottomline=false,leftline=false,rightline=false]
	\begin{center}
		\begin{enumerate}
			\item[(e)] $\exists \vv{M'^{(q)}_{QCO}}$ such that $\vv{M'^{(q)}_{QCO}} \majorizes \sortdown{\vv{L^{(q)}_{IBO}}}$ and $\vv{M'^{(q)}_{QCO}} \coordom \vv{M^{(q)}_{QCO}}$, for $q\in \{1,3\}$.
			\item[(f)] $\exists \vv{L'^{(2)}_{IBO}}$ such that $\vv{L'^{(2)}_{IBO}} \majorizes \sortdown{\vv{L^{(2)}_{IBO}}}$ and $\vv{L'^{(2)}_{IBO}} \coordom \vv{M^{(2)}_{QCO}}$.
			\item[(g)] $\exists \vv{M'^{(q)}_{WZ}}$ such that $\vv{M'^{(q)}_{WZ}} \majorizes \sortup{\vv{L^{(q)}_{WC}}}$ and $\vv{M'^{(q)}_{WZ}} \coordom \vv{M^{(q)}_{WZ}}$, for $q\in \{1,3\}$.
			\item[(h)] $\exists \vv{L'^{(2)}_{WC}}$ such that $\vv{L'^{(2)}_{WC}} \majorizes \sortup{\vv{L^{(2)}_{WC}}}$ and $\vv{L'^{(2)}_{WC}} \coordom \vv{M^{(2)}_{WZ}}$.
		\end{enumerate}
	\end{center}
\end{mdframed}

Let us see why the above conditions suffice to prove the inequality. 
The second part of (f) implies that $\norm{\vv{M_{QCO}}}_{1/2} \leq \norm{\vv{M^{(1)}_{QCO}} + \vv{M^{(3)}_{QCO}} + \vv{L'^{(2)}_{IBO}}}_{1/2}$.
Part (e) and the first part of (f), along with~\Cref{lem:sum-of-vectors}, implies $\vv{M^{(1)}_{QCO}} + \vv{M^{(3)}_{QCO}} + \vv{L'^{(2)}_{IBO}} \majorizes \sortdown{\vv{L_{IBO}}}$.
And so, $\norm{L_{IBO}}_{1/2} \geq \norm{M_{IBO}}_{1/2}$. Now, by the second part of (g) and the second part of (h) we have $\norm{\vv{M_{WC}}}_{1/2} \leq \norm{\vv{M'^{(1)}_{WC}} + \vv{M'^{(3)}_{WC}} + \vv{L'^{(2)}_{WC}}}_{1/2}$ and by the first part of (g) and (h) we have $\norm{\vv{M'^{(1)}_{WC}} + \vv{M'^{(3)}_{WC}} + \vv{L'^{(2)}_{WC}}}_{1/2} \majorizes \sortup{\vv{L_{WC}}}$. Thus, $\norm{L_{WZ}}_{1/2} \geq \norm{M_{WZ}}_{1/2}$.



\subsubsection{Proving (e) and (g) for $q=3$} \label{sec:eg3}

\paragraph{Defining the Coloring $\chi''$:} We now describe the coloring $\chi''$ on all edges of the form $(\bx,\bx + a\be_j)$ where $j \geq i+1$, $(i\circ h)(\bx) = 1$ and $(i \circ h)(\bx + a\be_j) = 0$. For all other edges $e$, we simply define $\chi''(e) = \chi(e)$. 

Given a pair of $i$-lines $\ell$ and $\ell' = \ell + a\be_j$ for $j \geq i+1$ and $a > 0$ we consider the set of violations from $\ell$ to $\ell'$ in $h$ and in $i \circ h$. As before, the violations in $h$ form two a union of two intervals $V = V_L \cup V_R$. Recall the definition of $V$ in \Eqn{violations}. Since $(i \circ h)$ is sorted in dimension $i$, the violations from $\ell$ to $\ell'$ in $(i \circ h)$ form a single interval which we will call $U$:
\[
U := \left\{(\bx,\bx+a\be_j) \colon \bx \in \ell \text{, } (i\circ h)(\bx) = 1)\text{, and } (i\circ h)(\bx + a\be_j) = 0\right\} \text{.}
\]
Since the sort operator can only reduce the number of violations in a dimension, we have $|U| \leq |V|$ (\Cref{clm:sort-violations} applied to $h|_{\ell}$ and $h|_{\ell'}$). We define $J$ to be the interval of $|V|-|U|$ points directly to the right of $U$ so that $U \cup J$ is an interval of size $|V|$. We then define 
\[
(\chi''(e) \colon \bx \in U \cup J) = \sortdown{\chi(e)  \colon e \in V} \text{.}
\]

We now have a complete description of $\chi''$. We provide the following illustration for clarity. The white and green intervals represent where $h = 0$ and $h=1$, respectively. The vertical arrows represent violated edges. Blue edges have color $0$ and red edges have color $1$. The left picture depicts the original coloring, $\chi$, and the original function, $h$. The right picture depicts the recoloring, $\chi''$, and the function after sorting, $i \circ h$. 

\includegraphics*[trim = 60 260 0 20, clip, scale = 0.5]{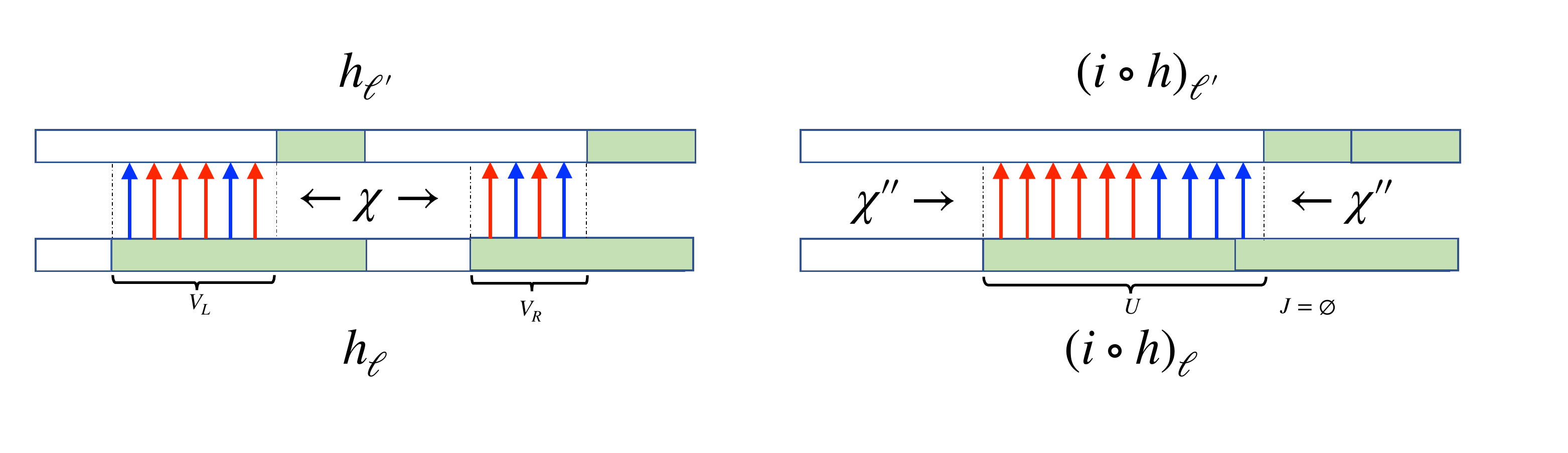}



We now return to our fixed $i$-line $\ell$ and set out to prove (e) and (g) for $q=3$, given this coloring $\chi''$. Let's recall our illustration of $h$ and $(i \circ h)$ restricted to $\ell$ and our definition of the intervals $W,I,B,C,O,Z,Q$. 

\hspace{1cm}\includegraphics*[trim = 0 210 0 80, clip, scale = 0.4]{figs/semisorted-to-sorted-try2}

\paragraph{Proving (e) for $q=3$:} Recall the definition of $A = I \cup B$, $O$, and $Q \cup C \cup O$ as in the illustration. Fix $j \geq i+1$ and a $i$-line $\ell' = \ell + a\be_j$. Let $A' := \{\bx \in A \colon h(\bx + a\be_j) = 0\}$, $O' := \{\bx \in O \colon h(\bx + a\be_j) = 0\}$, and $U := \{\bx \in Q \cup C \cup O \colon (i\circ h)(\bx + a\be_j) = 0\}$. Again, applying \Cref{clm:sort-violations} to $h|_{\ell}$ and $h|_{\ell'}$, we have $|U| \leq |A'| + |O'|$. Let $J$ denote the interval of size $|A'| + |O'| - |U|$ directly to the right of $U$ so that $U \cup J$ is an interval of size $|A'| + |O'|$. Observe that by our definition of $\chi''$ above, we have
\[
(\chi''(\bx,\bx+a\be_j) \colon \bx \in U \cup J) = \sortdown{\chi(\bx,\bx+a\be_j) \colon \bx \in A' \cup O'} \text{.}
\]
Let's see what this leads to.
\begin{definition}
	Fix $j \geq i + 1$ and fix an $i$-line $\ell' := \ell + a\be_j$ for $a > 0$. Define the following two boolean vectors:
	\[
	\bv^{M}_{j,a} := \left(\bone{\left((i\circ h)(\bx + a\be_j) = 0 ~~\textbf{and}~~~\chi''(\bx, \bx+a\be_j) = 1\right) }~~:~\bx\in Q\cup C\cup O \right)
	\]
	and 
	\[
	\bv^{L}_{j,a} := \left(\bone{\left(h(\bx + a\be_j) = 0 ~~\textbf{and}~~~\chi(\bx, \bx+a\be_j) = 1\right) }~~:~\bx\in I\cup B\cup O\right)\text{.}
	\]
\end{definition}
\noindent
Observe, for $\bx\in Q\cup C\cup O$, 
\begin{equation}\label{eq:M3R3-piece1.1}
	\Phi_{i\circ h,\chi''}(\bx;j) = \min\left(1, \sum_a \bv^{M}_{j,a}(\bx)\right) 
\end{equation}
and for $\bx \in I\cup B\cup O$, 
\begin{equation}\label{eq:M3R3-piece1.2}
	\Phi_{h,\chi}(\bx; j) = \min\left(1, \sum_a \bv^{L}_{j,a}(\bx) \right) \text{.}
\end{equation}
\begin{claim}
	Fix $j \geq i + 1$ and $a > 0$. For any two $\bx_i < \bx'_i$ in $Q\cup C\cup O$, we have $\bv^{M}_{j,a}(\bx) \geq \bv^{M}_{j,a}(\bx')$.
	That is, the vector $\bv^{M}_{j,a}$ is sorted decreasing.
\end{claim}
\begin{proof}	
	Since $(i\circ h)$ is semisorted $(i\circ h)(\bx' + a\be_j) = 0$ implies $(i\circ h)(\bx + a\be_j) = 0$.
	Furthermore, since both these are violations, by design $\chi''(\bx', \bx'+a\be_j) = 1$ implies $\chi''(\bx, \bx+a\be_j) = 1$. 
\end{proof}
\begin{claim}\label{clm:L3-permutation}
	Fix $j \geq i + 1$ and $a > 0$. The vector $\bv^{M}_{j,a}$ has at most as many $1$s as $\bv^{L}_{j,a}$ and thus $ \sortdown{\bv^{L}_{j,a}}\coordom \bv^{M}_{j,a}$.
\end{claim}
\begin{proof}
	This is precisely how $\chi''$ is defined: it only permutes the colorings on the violations incident on $I\cup B\cup O$, and this number can only decrease upon sorting (\Cref{clm:sort-violations} applied to $\chi$ restricted to the edges going from $\ell$ to $\ell'$). 
\end{proof}
In conclusion, we can write 
\[
\vv{L^{(3)}_{IBO}} = \left(\sum_{j=i+1}^d 	\Phi_{h,\chi}(\bx;j)~~:~~\bx \in I\cup B\cup O\right)
\]
as a sum of Boolean vectors, and the above two claims imply that the vector
\[
\vv{M^{(3)}_{QCO}} = \left(\sum_{j=i+1}^d 	\Phi_{(i\circ h),\chi''}(\bx;j)~~:~~\bx\in Q\cup C\cup O\right)
\]
is {\em coordinate wise} dominated by the 
sum of the {\em sorted decreasing} orders of those Boolean vectors. Defining $\vv{M'^{(3)}_{QCO}}$ to be the sum of the sorted decreasing orders, using~\Cref{lem:sum-of-vectors}, we establish part (e) for $q=3$. Namely, we get
\begin{mdframed}[backgroundcolor=gray!20,topline=false,bottomline=false,leftline=false,rightline=false]
	\begin{equation}\label{eq:M3R3-QCO}
		\exists \vv{M'^{(3)}_{QCO}}:~~\vv{M'^{(3)}_{QCO}} \majorizes \sortdown{\vv{L^{(3)}_{IBO}}}~~\text{and}~~ \vv{M'^{(3)}_{QCO}} \coordom \vv{M^{(3)}_{QCO}}
	\end{equation}
\end{mdframed}

\paragraph{Proving (g) for $q=3$:}

A similar argument but working with the zeros establishes part (g) for $q=3$. The picture is similar, but reversed, when we consider the points in $W\cup C$, where $h(\bx) = 0$. Fix a dimension $j \geq i+1$ and some $\ell'' = \ell - ae_j$. Let $W' = \{\bx \in W \colon h(\bx-a\be_j) = 1\}$, $C' = \{\bx \in C \colon h(\bx-a\be_j) = 1\}$, and $U = \{\bx \in W \cup Z \colon (i \circ h)(\bx - a\be_j) = 1\}$. Note that $|U| \leq |W'| + |C'|$ (\Cref{clm:sort-violations} applied to $h|_{\ell''}$ and $h|_{\ell}$). Let $J$ denote the interval of $|W'| + |C'|$ directly to the right of $|U|$ so that $U \cup J$ is an interval of size $|W'| + |C'|$. Observe that by our definition of $\chi''$ above, we have
\[
(\chi''(\bx-a\be_j,\bx) \colon \bx \in U \cup J) = \sortdown{\chi(\bx-a\be_j,\bx) \colon \bx \in W' \cup C'} \text{.}
\]
Let's see what this leads to.
\begin{definition}
	Fix $j \geq i + 1$ and fix an $i$-line $\ell'' := \ell - a\be_j$ for $a > 0$. Define the following two boolean vectors:
	\[
	\bv^{M}_{j,a} := \left(\bone{\left((i\circ h)(\bx - a\be_j) = 1 ~~\textbf{and}~~~\chi''(\bx-a\be_j, \bx) = 0\right) }~~:~\bx\in W\cup Z \right)
	\]
	and 
	\[
	\bv^{L}_{j,a} := \left(\bone{\left(h(\bx - a\be_j) = 1 ~~\textbf{and}~~~\chi(\bx-a\be_j, \bx) = 0\right) }~~:~\bx\in W \cup C\right)\text{.}
	\]
\end{definition}

\noindent
Observe, for $\bx\in W \cup Z$, 
\begin{equation}\label{eq:M3R3-piece1.1}
	\Phi_{i\circ h,\chi''}(\bx;j) = \min\left(1, \sum_a \bv^{M}_{j,a}(\bx)\right) 
\end{equation}
and for $\bx \in W \cup C$, 
\begin{equation}\label{eq:M3R3-piece1.2}
	\Phi_{h,\chi}(\bx; j) = \min\left(1, \sum_a \bv^{L}_{j,a}(\bx) \right) \text{.}
\end{equation}
\begin{claim}
	Fix $j \geq i + 1$ and $a > 0$. For any two $\bx_i < \bx'_i$ in $W \cup Z$, we have $\bv^{M}_{j,a}(\bx) \leq \bv^{M}_{j,a}(\bx')$.
	That is, the vector $\bv^{M}_{j,a}$ is sorted increasing.
\end{claim}
\begin{proof}	
	Since $(i\circ h)$ is sorted in dimension $i$, we have $(i\circ h)(\bx - a\be_j) = 1$ implies $(i\circ h)(\bx'-a\be_j) = 1$.
	Furthermore, since both these are violations, by design $\chi''(\bx-a\be_j, \bx) = 0$ implies $\chi''(\bx'-a\be_j, \bx') = 0$. 
\end{proof}
\begin{claim}\label{clm:L3-permutation}
	Fix $j \geq i + 1$ and $a > 0$. The vector $\bv^{M}_{j,a}$ has at most as many $1$s as $\bv^{L}_{j,a}$ and thus $ \sortup{\bv^{L}_{j,a}}\coordom \bv^{M}_{j,a}$.
\end{claim}
\begin{proof}
	This is precisely how $\chi''$ is defined: it only permutes the colorings on the violations incident on $W \cup C$, and this number can only decrease upon sorting (\Cref{clm:sort-violations} applied to $\chi$ restricted to the edges going from $\ell''$ to $\ell$). 
\end{proof}
In conclusion, 
we can write 
\[
\vv{L^{(3)}_{WC}} = \left(\sum_{j=i+1}^d 	\Phi_{h,\chi}(\bx;j)~~:~~\bx \in W\cup C\right)
\]
as a sum of Boolean vectors, and the above two claims imply that the vector
\[
\vv{M^{(3)}_{WZ}} = \left(\sum_{j=i+1}^d 	\Phi_{(i\circ h),\chi''}(\bx;j)~~:~~\bx\in W \cup Z\right)
\]
is {\em coordinate wise} dominated by the 
sum of the {\em sorted increasing} orders of those Boolean vectors. Defining $\vv{M'^{(3)}_{WZ}}$ to be the sum of the sorted decreasing orders, using~\Cref{lem:sum-of-vectors}, we establish part (g) for $q=3$. Namely, we get
\begin{mdframed}[backgroundcolor=gray!20,topline=false,bottomline=false,leftline=false,rightline=false]
	\begin{equation}\label{eq:M3R3-QCO}
		\exists \vv{M'^{(3)}_{MZ}}:~~\vv{M'^{(3)}_{MZ}} \majorizes \sortup{\vv{L^{(3)}_{WC}}}~~\text{and}~~ \vv{M'^{(3)}_{WZ}} \coordom \vv{M^{(3)}_{WZ}}
	\end{equation}
\end{mdframed}


\subsubsection{Proving (e) and (g) for $q=1$}


\paragraph{Defining $\xi'_{\bx}(S+i,S+i\oplus j)$ for $S\subseteq [i-1]$ and $j \leq i-1$:} We now define the partial coloring $\xi'_{\bx} := \xi^{(i)}_{\bx}$ on all edges $(S+i,S+i\oplus j)$ where $S \subseteq [i-1]$ and $j\leq i-1$ for all $\bx \in [n]^d$. These are exactly the relevant edges for the proof of parts (e) and (g) for $q=1$. Note that the partial coloring $\xi_{\bx} := \xi_{\bx}^{(i-1)}$ is undefined over these edges. 

Fix $S\subseteq [i-1]$, $j \leq i-1$, and a $i$-line $\ell$. We consider the set of $\bx \in \ell$ such that $(S,S\oplus j)$ is influential in $g_{\bx}$ and the set of edges where $(S+i,S+i\oplus j)$ is influential in $g_{\bx}$. As before, the former is a union of two intervals $V = V_L \cup V_R$. Recall the definition of $V$ in \Eqn{V_gx}. Since $(S+i) \circ f$ and $(S+i \oplus j) \circ f$ are both sorted in dimension $i$, the set of $\bx \in \ell$ such that $(S+i,S+i\oplus j)$ is influential forms a single interval which we will call $U$: 
\[
U := \left\{\bx \in \ell \colon g_\bx(S + i) = 1 \text{ and } g_{\bx}(S+i \oplus j) = 0\right\}\text{.}
\]

Again, we have $|U| \leq |V|$ (\Cref{clm:sort-violations} applied to $(S \circ f)|_{\ell}$ and $((S \oplus j) \circ f)|_{\ell}$) and we let $J$ denote the $|V| - |U|$ points directly right of $U$, so that $U \cup J$ is an interval of length $|V|$. We then define 
\[
(\xi'_{\bx}(S+i,S+i\oplus j) \colon \bx \in U \cup J) = \sortdown{\xi_{\bx}(S,S\oplus j) \colon \bx \in V} \text{.}
\]
For all $x \in \ell \setminus (U \cup J)$ we define $\xi'_{\bx}(S+i,S+i\oplus j) = 1$. Note that this is an arbitrary choice since such edges are not influential and so they do not come in to play in the rest of the proof. 

We now have a complete description of $\xi'_{\bx}$ on $(S+i,S+i\oplus j)$ for all $\bx \in [n]^d$. We provide the following illustration for clarity, which is quite similar to the illustration provided in \Sec{eg3} when we defined $\chi''$. The left picture depicts the original colorings, $\xi_{\bx}$, and the relevant functions before applying the sort operator in dimension $i$. The right picture depicts the recoloring, $\xi'_{\bx}$, and the relevant functions after applying the sort operator in dimension $i$. 

\includegraphics*[trim = 50 280 0 20, clip, scale = 0.5]{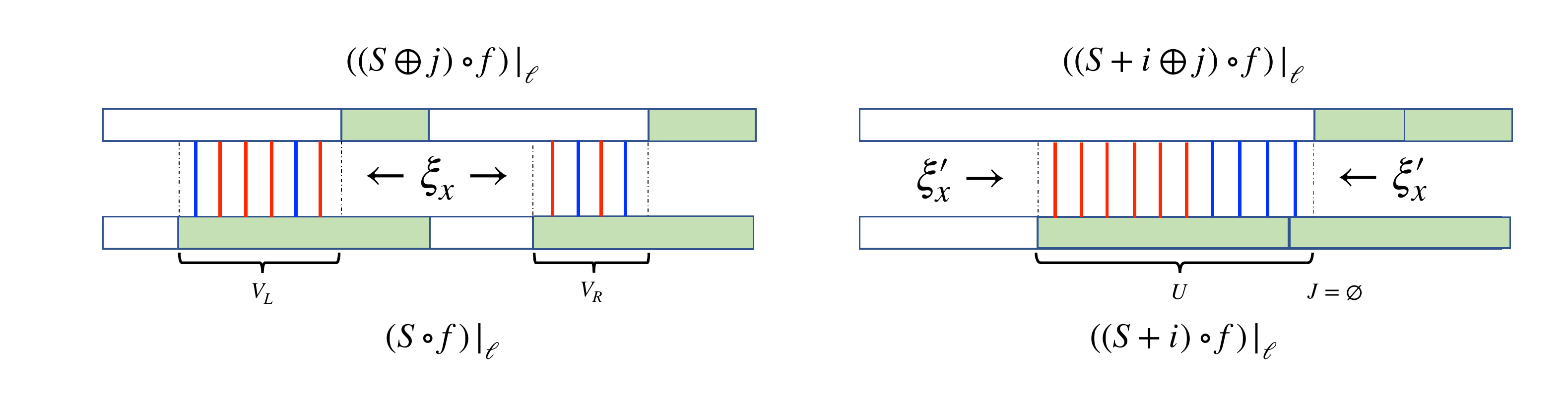}

We now return to our fixed $i$-line $\ell$ and set out to prove (e) and (g) for $q=1$, given the colorings $\xi'_{\bx}$. Recall $g_{\bx}(S) = h(\bx)$ and $g_{\bx}(S+i) = (i \circ h)(\bx)$ and so we can reference the same illustration and our definition of the intervals $W,I,B,C,O,Z,Q$. 

\hspace{1cm}\includegraphics*[trim = 0 210 0 80, clip, scale = 0.4]{figs/semisorted-to-sorted-try2}

\paragraph{Proving (e) for $q=1$:}



Fix $j \leq i-1$ and let $A' := \{\bx \in A \colon g_{\bx}(S \oplus j) = 0\}$, $O' := \{\bx \in O \colon g_{\bx}(S \oplus j) = 0$, and $U := \{\bx \in Q \cup C \cup O \colon g_{\bx}(S+i\oplus j) = 0\}$. As before, $|U| \leq |A'| + |O'|$ (applying \Cref{clm:sort-violations} to $(S \circ f)|_{\ell}$ and $((S \oplus j) \circ f)|_{\ell}$) and we define $J$ to be the $|A'| - |O'|$ points directly to the right of $U$ so that $U \cup J$ is a prefix of $Q\cup C \cup O$ of size $|A'| + |O'|$. From our definition of $\xi'_{\bx}$ from above we have
\[
(\xi'_{\bx}(S+i,S+i\oplus j) \colon \bx \in U \cup J) = \sortdown{\xi_{\bx}(S,S\oplus j) \colon \bx \in A' \cup O'} \text{.}
\]

We now get the following claim.
\begin{claim}
	$\left(I_{g_\bx,\xi'_\bx}^{=j}(S+i)~~:~~ \bx\in Q\cup C\cup O\right)$ is a sorted decreasing vector, and has at most as many ones as the vector $\left(I_{g_\bx,\xi_\bx}^{=j}(S)~~:~~ \bx\in I\cup B\cup O\right)$.
\end{claim}
\begin{proof}
	Take $\bx_i < \bx'_i$ in $Q\cup C\cup O$. Note that $g_\bx(S+i) = g_{\bx'}(S+i) = 1$ by definition $Q\cup C\cup O$.
	Thus,
	\[
	I_{g_\bx,\xi'_\bx}^{=j}(S + i) = \bone\left(g_\bx(S + i\oplus j) = 0 ~\textbf{and}~\xi'_\bx(S+i,S+i\oplus j) = 1 \right)
	\]
	and
	\[
	I_{g_\bx,\xi_\bx}^{=j}(S) = \bone\left(g_\bx(S\oplus j) = 0 ~\textbf{and}~\xi_\bx(S,S\oplus j) = 1 \right)
	\]
	By design of the $\xi'_\bx$'s, the first vector is sorted decreasing on $Q\cup C\cup O$ (it takes value $0$ after $U$).
Also by design, the number of ones in the latter vector can only be larger since we obtain $\xi'$ by taking a permutation and possibly discarding some ones (the ones corresponding to $J$).
\end{proof}
Observing that 
\[
\vv{L^{(1)}_{IBO}} = \left(\sum_{j=1}^{i-1} I_{g_\bx,\xi_\bx}^{=j}(S)~:~\bx\in I\cup B\cup O\right) ~~\textrm{and}~~\vv{M^{(1)}_{QCO}} = \left(\sum_{j=1}^{i-1} I_{g_\bx,\xi'_\bx}^{=j}(S+i)~:~\bx\in Q\cup C\cup O\right)
\]
we see that the latter vector is coordinate-wise dominated by a vector which is a sum of \emph{sorted decreasing} versions of Boolean vectors which add up to the former one. Defining $\vv{M'^{(1)}_{QCO}}$ to be the sum of the sorted decreasing orders, using~\Cref{lem:sum-of-vectors}, we establish part (e) for $q=3$. Namely, we get 
\begin{mdframed}[backgroundcolor=gray!20,topline=false,bottomline=false,leftline=false,rightline=false]
	\begin{equation}\label{eq:M1R1-QCO}
		\exists \vv{M'^{(1)}_{QCO}}:~~\vv{M'^{(1)}_{QCO}} \majorizes \sortdown{\vv{L^{(1)}_{IBO}}}~~\text{and}~~ \vv{M'^{(1)}_{QCO}} \coordom \vv{M^{(1)}_{QCO}} \text{.}
	\end{equation}
\end{mdframed}

\paragraph{Proof of Part (g) for $q=1$:}


A similar argument but working with the zeros establishes part (g) for $q=1$. Recall the definition of the sets $W$, $C$, and $Z$. Let $W' = \{\bx \in W \colon g_{\bx}(S\oplus j) = 1\}$, $C' = \{\bx \in C \colon g_{\bx}(S \oplus j) = 1\}$, and $U = \{x \in W \cup Z \colon g_{\bx}(S+i\oplus j) = 1\}$. As before $|U| \leq |W'| + |C'|$ (applying \Cref{clm:sort-violations} to $((S \oplus j) \circ f)|_{\ell}$ and $(S \circ f)|_{\ell}$) and we define $J$ to be the set of $|W'|+|C'| - |U|$ points directly to the right of $U$ so that $U \cup J$ is an interval of size $|W'| + |C'|$. Note that $U$ is a suffix of $W \cup Z$ and $J$ is a prefix of $Q \cup C \cup O$. 

From our definition of $\xi'_{\bx}$ above, made with the set $S \oplus j$,  we have
\[
(\xi'_{\bx}(S+i,S+i\oplus j) \colon \bx \in U \cup J) = \sortdown{\xi_{\bx}(S,S\oplus j) \colon \bx \in W' \cup C'} \text{.}
\]

\begin{claim}
	$\left(I_{g_\bx,\xi'_\bx}^{=j}(S+i)~~:~~ \bx\in W \cup Z\right)$ is a sorted increasing vector, and has at most as many ones as the vector $\left(I_{g_\bx,\xi_\bx}^{=j}(S)~~:~~ \bx\in W \cup C\right)$.
\end{claim}
\begin{proof}
	Take $\bx_i < \bx'_i$ in $W \cup Z$. Note that $g_\bx(S+i) = g_{\bx'}(S+i) = 0$ by definition $W \cup Z$.
	Thus,
	\[
	I_{g_\bx,\xi'_\bx}^{=j}(S + i) = \bone\left(g_\bx(S + i\oplus j) = 1 ~\textbf{and}~\xi'_\bx(S+i,S+i\oplus j) = 0 \right)
	\]
	and
	\[
	I_{g_\bx,\xi_\bx}^{=j}(S) = \bone\left(g_\bx(S\oplus j) = 1 ~\textbf{and}~\xi_\bx(S,S\oplus j) = 0 \right)
	\]
	By design of the $\xi'_\bx$'s, the first vector is sorted increasing on $W \cup Z$.
Also by design, the number of ones in the latter vector can only be larger since we obtain $\xi'$ by taking a permutation and possibly discarding some ones (the ones corresponding to $J$).
\end{proof}
Observing that 
\[
\vv{L^{(1)}_{WC}} = \left(\sum_{j=1}^{i-1} I_{g_\bx,\xi_\bx}^{=j}(S)~:~\bx\in W \cup C\right) ~~\textrm{and}~~\vv{M^{(1)}_{WZ}} = \left(\sum_{j=1}^{i-1} I_{g_\bx,\xi'_\bx}^{=j}(S+i)~:~\bx\in W \cup Z\right)
\]
we see that the latter vector is coordinate-wise dominated by a vector which is a sum of \emph{sorted increasing} versions of Boolean vectors which add up to the former one. Defining $\vv{M'^{(1)}_{WZ}}$ to be the sum of the sorted increasing orders, using~\Cref{lem:sum-of-vectors}, we establish part (e) for $q=3$. Namely, 
\begin{mdframed}[backgroundcolor=gray!20,topline=false,bottomline=false,leftline=false,rightline=false]
	\begin{equation}\label{eq:M1R1-QCO}
		\exists \vv{M'^{(1)}_{WZ}}:~~\vv{M'^{(1)}_{WZ}} \majorizes \sortdown{\vv{L^{(1)}_{WC}}}~~\text{and}~~ \vv{M'^{(1)}_{WZ}} \coordom \vv{M^{(1)}_{WZ}} \text{.}
	\end{equation}
\end{mdframed}

\subsubsection{Proving (f) and (h):}

Let us now prove part (f) and (h). Note, at this point, $\xi'_\bx$ is fully defined on all pairs $(S, S\oplus j)$ for $S\subseteq [i]$ and $j \leq i$. We don't have the freedom to redefine.
However, we see that the definition we made in~\eqref{eq:xi1} and~\eqref{eq:xi2} suffices. Let us recall what we want to establish.

\begin{enumerate}
			\item[(f)] $\exists \vv{L'^{(2)}_{IBO}}$ such that $\vv{L'^{(2)}_{IBO}} \majorizes \sortdown{\vv{L^{(2)}_{IBO}}}$ and $\vv{L'^{(2)}_{IBO}} \coordom \vv{M^{(2)}_{QCO}}$.
		\item[(h)] $\exists \vv{L'^{(2)}_{WC}}$ such that $\vv{L'^{(2)}_{WC}} \majorizes \sortup{\vv{L^{(2)}_{WC}}}$ and $\vv{L'^{(2)}_{WC}} \coordom \vv{M^{(2)}_{WZ}}$.
\end{enumerate}

We remind the reader that  $\vv{L^{(2)}}(\bx) = \Phi_{h,\chi}(\bx;i)$ for all $\bx\in \ell$ and the coloring was defined as follows:
\begin{equation}
	\textrm{If}~~\vv{L^{(2)}_{IB}} \equiv \bone, ~~\textrm{then}~~ \xi'_\bx(S, S+i) = 1~~\forall \bx \in I\cup B\cup C \notag
\end{equation}
otherwise, 
\begin{equation}
	\textrm{we have}~~\vv{L^{(2)}_C} \equiv \bone, ~~\textrm{and so}~~ \xi'_\bx(S, S+i) = 0~~\forall \bx \in I\cup B \cup C \notag
\end{equation}

\noindent
We remind the reader that $\vv{M^{(2)}}(\bx) = I^{=i}_{g_\bx, \xi'_\bx}(S+i)$ and therefore this is $1$
iff $g_\bx(S + i) \neq g_\bx(S)$ and $\xi'_\bx(S, S+i) = g_\bx(S+i)$. The former implies $\bx \in Z\cup C$. 

Suppose we are in the first case. Then, $\vv{M^{(2)}}(\bx) = 1$ if and only if $\bx \in C$.
Since $\vv{L^{(2)}_{IB}} \equiv \bone \coordom \vv{M^{(2)}_{QC}}$, we can set $\vv{L'^{(2)}_{IBO}}$ to be the vector that is $1$s in $I\cup B$ and $0$'s in $O$. This establishes (f).
To establish (h), we observe that $\vv{M^{(2)}_{WZ}}$ is the zero vector, and thus we can choose $\vv{L'^{(2)}_{WC}}$ to be $\sortup{\vv{L^{(2)}_{WC}}}$.

Suppose we are in the second case. Then, $\vv{M^{(2)}}(\bx) = 1$ if and only if $\bx \in Z$.
Since $\vv{L^{(2)}_{C}} \equiv \bone \coordom \vv{M^{(2)}_{Z}}$, we can set $\vv{L'^{(2)}_{WC}}$ to be the vector that is $1$s in $C$ and $0$'s in $W$. This establishes (h).
To establish (f), we observe that $\vv{M^{(2)}_{QCO}}$ is the zero vector, and thus we can choose $\vv{L'^{(2)}_{IBO}}$ to be $\sortdown{\vv{L^{(2)}_{IBO}}}$.

In either case, we have established (f) and (h), and thus completed the proof.

%% file: tester_analysis.tex
\section{The Tester and it's Analysis: Proof of~\Cref{thm:mono-testing}} \label{sec:tester}

With the isoperimetric theorem of \Thm{dir-tal} in place, we can now design and analyze
the monotonicity tester for Boolean hypergrid functions. This section closely
follows the analogous analysis in~\cite{KMS15}, and will lift certain notions from that paper. 
We do have to make slight adaptations to various
arguments therein to account for the hypergrid domain. 

We first describe the path tester for the hypergrid.

\begin{figure}[h]
\begin{framed}
    \noindent \textbf{Input:} A Boolean function $f: [n]^d \to \{0,1\}$. 

    \smallskip

    \begin{asparaenum}
        \item Choose $k\in_R \{0,1,2,\ldots,\ceil{\log d}\}$ uniformly at random. Set $\tau := 2^k$.
        \item Choose $\bx\in [n]^d$ uniformly at random. Denote $\bx = (\bx_1, \bx_2, \ldots, \bx_d)$.
        \item Pick a uniform random subset $R\subseteq [d]$ of $\tau$ coordinates.
        \item For each $r \in R$, pick uar value $c_r \in [n] \setminus \{\bx_r\}$. 
        \item Generate $\bz$ as follows. For every $r \in [d]$, if $r \in R$ \emph{and} $c_r > \bx_r$, set $\bz_r = c_r$.
        Else, set $\bz_r = \bx_r$.
        \item If $f(\bz) > f(\bx)$, REJECT.
    \end{asparaenum} 
\end{framed}
\caption{\small{\textbf{Path Tester for Hypergrid Functions}}}
\label{fig:alg}
\end{figure}
Clearly, the tester doesn't reject any monotone function.
Our main theorem regarding the tester follows. 
 A standard boosting argument gives us the $\widetilde{O}(\eps^{-2} n\sqrt{d})$-query Boolean monotonicity tester on $[n]^d$, proving~\Cref{thm:mono-testing}.

\begin{theorem} \label{thm:tester} If $f$ is $\eps$-far from monotone,
then the path tester for hypergrid functions rejects with probability $\Omega(\frac{\eps^2}{n\sqrt{d}\log^5(nd)})$.
\end{theorem}

\subsection{Directed random walks, influences, and persistence} \label{sec:dirwalk}

One can think of the above tester as obtaining $\bz$ by a lazy directed random walk of $\tau$ steps from a uniform random $\bx$.
Note that in some steps we may not move at all; these correspond to coordinates $r\in R$ such that $c_r \leq \bx_r$.
It will be convenient to define an alternate, but equivalent process that generates the pair $(\bx,\bz)$. 

\begin{figure}[h]
\begin{framed}
    \noindent \textbf{Input:} A length parameter $\tau$.

    \smallskip

    \begin{asparaenum}
        \item In each dimension $i$, sample a uniform random pair $a_i < b_i$ from $[n]$.
        \item Let $H$ be the hypercube formed by $\prod_{i=1}^d \{a_i, b_i\}$.
        \item Pick a uar point $\bx$ from $H$.
        \item Pick a uar subset $R$ of $\tau$ coordinates, permuted randomly.
        \item Generate $\bz$ as follows. For every $r \in [d]$, if $r \in R$ and $\bx_r = a_r$, set $\bz_r = b_r$. Else, set $\bz_r = \bx_r$.
    \end{asparaenum} 
\end{framed}
\caption{\small{\textbf{Directed random walks, by sampling hypercubes}}}
\label{fig:dirwalk}
\end{figure}

This process generates walks by first sampling a random hypercube $H$, and then
doing a lazy directed walk on $H$. We first observe that conditioned on the walk length $\tau$, the distribution
of $(\bx, \bz)$ pairs
generated by the path tester and the above process are identical. 

\begin{observation} \label{obs:walk-distributions} Fix $\tau \in [d]$. Let $\cD_1,\cD_2$ denote the distributions over $[n]^d \times [n]^d$ described in \Fig{alg} and \Fig{dirwalk}, respectively, conditioned on the walk length $\tau$. Then, for any pair $(\bx,\bz)$ where $\bx \preceq \bz$, we have $\Pr[(\bx,\bz) \sim \cD_1] = \Pr[(\bx,\bz) \sim \cD_2]$. \end{observation}

\begin{proof} 
Let $S = \{i \in [d] \colon \bz_i > \bx_i\}$. The probability from $\cD_1$ is given by
\begin{alignat}{4}
\Pr[(\bx, \bz) \sim \cD_1] &&~=~& \frac{1}{n^d} \cdot \sum_{R\supseteq S~:~|R|=\tau} {d \choose \tau}^{-1} \prod_{i\in S} \Pr[c_i = \bz_i] \prod_{i\in R\setminus S} \Pr[c_i \leq \bx_i] \notag \\
&&~=~& \frac{1}{n^d} \cdot {d \choose \tau}^{-1} \cdot \sum_{R\supseteq S~:~|R|=\tau} \prod_{i\in S} \frac{1}{n-1} \prod_{i\in R\setminus S} \frac{\bx_i-1}{n-1} \text{.}\notag
\end{alignat}

We now compute the probability for $\cD_2$. For $i \in [d]$, let $\cE_i$ be the event that $a_i = \bx_i$ or $b_i = \bx_i$. Note
\[
\Pr[\neg \cE_i] = \frac{\binom{n-1}{2}}{\binom{n}{2}} = \frac{n-2}{n	} ~~\Rightarrow \Pr[\cE_i] = \frac{2}{n} \text{.}
\]
Let $\cE_{\bx}$ denote the event that the first point sampled according to $\cD_2$ is $\bx$. We have
\[
\Pr[\cE_{\bx}] = \prod_{i=1}^d \Pr[\cE_i]\cdot \frac{1}{2^d} = \left(\frac{2}{n}\right)^d \frac{1}{2^d} = \frac{1}{n^d} \text{.}
\]
Let $\cE_\bz$ denote the event that the second point is $\bz$. We have 
\[
\Pr\left[\cE_{\bz} | \cE_{\bx}\right] = \sum_{R\supseteq S~:~|R|=\tau} {d \choose \tau}^{-1} \prod_{i\in S} \Pr[a_i = \bx_i ~\textrm{and}~b_i = \bz_i~|~\cE_\bx] \cdot \prod_{i\in R\setminus S} \Pr[b_i = \bx_i~|~\cE_\bx]
\]
Fix an $i\in S$. We have
\[
\Pr[a_i = \bx_i ~\textrm{and}~b_i = \bz_i~|~\cE_\bx] = \Pr[a_i = \bx_i ~\textrm{and}~b_i = \bz_i~|~\cE_i] = \frac{\Pr[a_i = \bx_i ~\textrm{and}~b_i = \bz_i]}{\Pr[\cE_i]} = \frac{1/{n \choose 2}}{2/n} = \frac{1}{n-1} \text{.}
\]
Now fix an $i\in R\setminus S$. We have 
\[
\Pr[b_i = \bx_i~|~\cE_{\bx}] = \Pr[b_i = \bx_i~|~\cE_i] = \frac{\Pr[b_i = \bx_i]}{\Pr[\cE_i]} = \frac{\left(\frac{\bx_i-1}{{n \choose 2}}\right)}{2/n} = \frac{\bx_i-1}{n-1} \text{.}
\]
Therefore, $\Pr[(\bx,\bz) \sim \cD_2] = \Pr[\cE_\bx] \cdot \Pr[\cE_{\bz} | \cE_{\bx}] = \Pr[(\bx,\bz) \sim \cD_1]$. \end{proof}



\noindent

We now have that the random walk distributions described in \Fig{alg} and \Fig{dirwalk} are equivalent. The former is more convenient to analyze for the final rejection probability, but the latter
perspective allows us to prove various influence bounds by piggybacking on the \cite{KMS15} hypercube analysis.

\begin{definition} \label{def:totalinf} We define the total influence and total negative
influence of $f$ as follows.
$$ I_f = n^{-d} \sum_{\bx \in [n]^d} \sum_{i=1}^d \sum_{c = 1}^n \bone(f(\bx) \neq f(\bx_1, \bx_2, \ldots, \bx_{i-1}, c, \bx_{i+1}, \ldots, \bx_d))$$
$$ I^-_f = n^{-d} \sum_{\bx \in [n]^d: f(\bx) = 1} \sum_{i=1}^d \sum_{c > \bx_i} \bone(f(\bx) \neq f(\bx_1, \bx_2, \ldots, \bx_{i-1}, c, \bx_{i+1}, \ldots, \bx_d))$$
Note these are different from the thresholded influences and for most functions will grow as $n$ grows.
\end{definition}

We can analogously define these influences on the hypercubes $H$ sampled by the process described in \Fig{dirwalk}.
Abusing notation, we will denote these influences as $I_H := I_{f|_{H}}$ and $I_H^- := I^-_{f|_H}$. A simple, yet important claim
relates the expected influence on $H$ to the total influence on the hypergrid. All expectations over $H$ are taken with respect to the distribution described in steps 1-2 of \Fig{dirwalk}.

\begin{claim} \label{clm:inf-hyp} $\EX_H[I_H] = I_f/(n-1)$ and $\EX_H[I^-_H] = I^-_f/(n-1)$. 
\end{claim}

\begin{proof} Consider any $i$-edge $(\bx,\by)$ of the fully augmented hypergrid.
The probability that this edge is present in $H$ can be computed as follows.
Firstly, we need $a_i = \bx_i$ and $b_i = \by_i$. This happens with probability ${n \choose 2}^{-1}$. Then, for all $j \neq i$,
one of $\{a_j,b_j\}$ needs to be $\bx_j$ (note that $\bx_j = \by_j$ since $(\bx,\by)$ is an $i$-edge). This happens with probability $(1-({n-1 \choose 2}/{n \choose 2}))^{d-1} = (2/n)^{d-1}$.

The total number of influential edges in $f$ is $n^d I_f$. Thus, by linearity
of expectation, the expected number of influential edges in $H$
is ${n \choose 2}^{-1} \times (2/n)^{d-1} \times n^d I_f = 2^d I_f/(n-1)$. 

An analogous proof holds for the negative influence. \end{proof}

By using a lemma
of \cite{KMS15}, we can prove that if the total influence of $f$ is too large, then the negative influence
is also large. 

\begin{claim} \label{clm:tot-neg-inf} If $I_f > 9(n-1)\sqrt{d}$, then $I^-_f > (n-1)\sqrt{d}$.
\end{claim}

\begin{proof} Theorem 9.1 of \cite{KMS15} asserts that, if $I_H > 6\sqrt{d}$, then $I^-_H > I_H/3$.
(This holds for any Boolean hypercube function.) If $I_f > 9(n-1)\sqrt{d}$, then
by \Clm{inf-hyp}, $\EX_H[I_H] > 9\sqrt{d}$. Hence,
\begin{eqnarray}
9\sqrt{d} < \EX_H[I_H] & = & \Pr[I_H \leq 6\sqrt{d}] \ \EX_H[I_H | I_H \leq 6\sqrt{d}] \ + \ \Pr[I_H > 6\sqrt{d}] \ \EX_H[I_H | I_H > 6\sqrt{d}] \nonumber \\
& < & 6\sqrt{d} + \Pr[I_H > 6\sqrt{d}] \EX_H[3 I^-_H | I_H > 6\sqrt{d}] \leq 6\sqrt{d} + 3\EX_H[I^-_H]
\end{eqnarray}
Hence, $\EX_H[I^-_H] > \sqrt{d}$. By \Clm{inf-hyp}, $I^-_f > (n-1)\sqrt{d}$.
\end{proof}

One of the crucial definitions is that of \emph{persistence}.

\begin{definition} \label{def:persist} A point $\bx \in [n]^d$ is called $\tau$-persistent
if $\Pr_{\bz}[f(\bx) = f(\bz)] \geq 1/2$ where $\bz$ is chosen by a $\tau$-length directed random walk from $\bx$. \end{definition}


\begin{lemma} \label{lem:persist} If $I_f \leq 9(n-1)\sqrt{d}$, then the fraction of vertices that are not $\tau$-persistent is at most $\cper \tau/\sqrt{d}$ (where $\cper$ is an absolute constant). \end{lemma}

\begin{proof} We will analyze the random walk using the distribution described in \Fig{dirwalk} and leverage the analysis that \cite{KMS15} use to prove their Lemma 9.3. Let $\alpha$ denote the fraction of $\tau$-non-persistent vertices in the fully augmented hypergrid with respect to our function $f$. Again, let $\cD_1$, $\cD_2$ denote the distributions over pairs $(\bx,\bz)$ described in \Fig{alg}, \Fig{dirwalk}, respectively, conditioned on the walk length $\tau$. For a fixed $H$ sampled in steps 1-2 of \Fig{dirwalk}, let $\cD_{2,H}$ denote the distribution over $(\bx,\bz)$ described in steps 3-5. Using the definition of persistence and \Obs{walk-distributions}, we have
\begin{align} \label{eq:persistence1}
    \frac{\alpha}{2} \leq \Pr_{(\bx,\bz)\sim \cD_1}(f(\bx) \neq f(\bz)) = \EX_H\left[\Pr_{(\bx,\bz)\sim \cD_{2,H}}\left(f(\bx) \neq f(\bz)\right)\right] \text{.}
\end{align}
Let $\widehat{\cD}_{2,H}$ denote the same distribution as $\cD_{2,H}$ except with $R$ being a uar subset of the $0$-coordinates of $\bx$ (recall step 4 of \Fig{dirwalk}). I.e. $\widehat{\cD}_{2,H}$ is the \emph{non-lazy} walk distribution on $H$. Let $\bx = \bx^0, \bx^1, \ldots, \bx^{\tau} = \bz$ be the $\tau$ steps taken on the walk sampled by $\cD_{2,H}$ and let $\bx = \widehat{\bx}^0, \widehat{\bx}^1, \ldots, \widehat{\bx}^{\tau} = \bz$ be the $\tau$ steps taken on the walk sampled by $\widehat{\cD}_{2,H}$. For a fixed $H$ we have
\begin{align} \label{eq:persistence2}
    \Pr_{(\bx,\bz)\sim \cD_{2,H}}\left(f(\bx) \neq f(\bz)\right) \leq \sum_{\ell=0}^{\tau - 1} \Pr\left(f(\bx^{\ell}) \neq f(\bx^{\ell+1})\right) \leq \sum_{\ell=0}^{\tau - 1} \Pr\left(f(\widehat{\bx}^{\ell}) \neq f(\widehat{\bx}^{\ell+1})\right) \text{.}
\end{align}
The first inequality is by a union bound and the second inequality holds because the first walk is lazy and the second is not. More precisely, we can couple the $\tau' \leq \tau$ steps 
of the lazy-random walk where the point actually moves to the first $\tau'$ steps of the second non-lazy walk, and the remaining $\tau-\tau'$ terms of the non-lazy walk can only increase the RHS.

By Lemma 9.4 of \cite{KMS15}, the edge $(\widehat{\bx}^{\ell},\widehat{\bx}^{\ell+1})$ is distributed approximately as a uniform random edge in $H$. In particular, this implies $\Pr\left(f(\widehat{\bx}^{\ell}) \neq f(\widehat{\bx}^{\ell+1})\right) \leq C \cdot 2I_H/d$ for an absolute constant $C$. (Note $2I_H/d$ is the probability of a uniform random edge in $H$ being influential.) Putting \Eqn{persistence1} and \Eqn{persistence2} together yields $\alpha \leq \frac{4 C \tau }{d}\EX_H[I_H]$. By \Clm{inf-hyp} we have $\EX_H[I_H] \leq 9\sqrt{d}$ and so setting $\cper := 36C$ completes the proof. \end{proof}



\subsection{The good subgraph and capturing violations} \label{sec:good-sub}

We now use the isoperimetric theorem of \Thm{dir-tal} to construct a \emph{good subgraph},
in the parlance of~\cite{KMS15}. 

\begin{theorem} \label{thm:good-sub} There exists a bipartite subgraph $G = (X,Y,E)$
of the fully augmented hypergrid with the following properties. 
\begin{asparaitem}
    \item $|X| = \sigma n^d$ or $|Y| = \sigma n^d$. 
    \item Every vertex has degree at most $k$.
    \item For all $\bx \in X$, $f(\bx) = 1$. For all $\by \in Y$, $f(\by) = 0$.
    \item $|E| \geq \sigma k n^d/2$.
    \item $\sigma\sqrt{k} = \Theta(\eps/\log^2(nd))$.
\end{asparaitem}
\end{theorem}

\begin{proof} Consider the bipartite subgraph consisting of all violations
of the fully augmented hypergrid. Consider any bi-coloring $\chi$ of the edges
of this subgraph. Let $\deg_{\chi}(\bx)$ denote the number of violating edges $e$ incident to $\bx$ for which $\chi(e) = f(\bx)$. Note that $\deg_{\chi}(\bx) \geq \Phi_{f,\chi}(\bx)$.
Hence, \Thm{dir-tal} asserts that $\sum_\bx \sqrt{\deg_{\chi}(\bx)} \geq C' n^d \eps/\log n$,
for some absolute constant $C'$. 
According to Def. 6.4 of \cite{KMS15}, the bipartite graph is $C' n^d \eps/\log n$-robust.

By Lemma 6.5 of \cite{KMS15}, any robust bipartite graph has a ``good subgraph" satisfying 
the following bound. (Below, $C''$ is a constant.)
$$ \sigma n^d \times \sqrt{k} \geq \frac{C' \eps n^d}{8\log(nd)\log n} \ \ \ \Longrightarrow \ \ \ \sigma\sqrt{k} \geq \frac{C''\eps}{\log^2(nd)} $$
One can remove vertices from this good subgraph to ensure that $\sigma \sqrt{k} = \Theta(\eps/\log^2(nd))$.
\end{proof}

For the rest of the analysis we will assume $|X| = \sigma n^d$, without loss of generality. The edges of the good subgraph of \Thm{good-sub} are central to the tester analysis. We will need
to choose $\tau$ carefully to ensure that the analysis carries through. Towards that choice,
we will set a convenient bound on $\sigma$. We will use $\clarge$ to denote a sufficiently
large constant that is at least $100\cper$ and the constants of \Thm{dir-tal}. (The constant $\cper$ is from \Lem{persist}.)

\begin{claim} \label{clm:small-sigma} If $\sigma < \clarge /\sqrt{d}$, then $I^-_f = \Omega(\eps^2\sqrt{d}/\log^4(nd))$.
\end{claim}

\begin{proof} By the good subgraph properties in \Thm{good-sub}, $\sqrt{k} = \Omega(\eps\sqrt{d}/\log^2(nd))$.
Hence, the number of edges of the good subgraph is at least $\sigma k n^d/2 = (\sigma \sqrt{k}) \times \sqrt{k} \times n^d/2
= \Omega(\eps^2 \sqrt{d} n^d/\log^4(nd))$. We divide this bound by $n^d$ to bound $I^-_f$.
\end{proof}

Essentially, for the analysis, we can ignore the case when $\sigma$ is too small. With this bound in place,
we can now set the right choice of $\tau$ based on the good subgraph parameters.

\begin{definition} \label{def:tau-sigma} For any $\sigma \geq \clarge/\sqrt{d}$, define $\tau_\sigma$ to be the power of $2$
in the range $[\sigma \sqrt{d}/\clarge, \sigma 2\sqrt{d}/\clarge]$. (Since $\sigma \sqrt{d}/\clarge \geq 1$,
the choice of $\tau_\sigma$ exists.)
\end{definition}

We will now define a particular ``edge capturing event" that ensures that the tester
finds a violation to monotonicity. The crucial property is that these events
are uniquely associated with edges of the good subgraph, and are all disjoint.
So, we can lower bound the probability of this event and multiply by the
number of edges of the good subgraph.

\begin{definition} \label{def:capture} Let $\bx \in X$. We call an edge $(\bx,\by)$ of the good
subgraph \emph{viable} if $\by$ is $\tau_\sigma$-persistent. The set of 
\emph{viable coordinates} of $\bx$ are the dimensions containing the viable
edges incident to $\bx$. 

For a viable $i$-edge $e = (\bx,\by)$, the \emph{capturing event} $\cC_e$ is defined as follows.
Consider the sampling process of the tester, and condition on $\tau \eqdef \tau_\sigma$. We define $\cC_e = \cE_1 \wedge \cE_2 \wedge \cE_3 \wedge \cE_4 \wedge \cE_5$ where:
\begin{asparaitem}
    \item $\cE_1$: The point $\bx$ is chosen (as the first point).
    \item $\cE_2$: The coordinate set $R$ contains $i$.
    \item $\cE_3$: The coordinate $c_i$ is $\by_i$.
    \item $\cE_4$: $R\setminus i$ does not contain any viable coordinates of $\bx$.
    \item $\cE_5$: $f(\bz) = 0$.
\end{asparaitem}
\end{definition}

The main calculation is to lower bound the probability of the event $\cC_e$,
for any viable edge.

\begin{lemma} \label{lem:viable} For any viable edge $e = (\bx,\by)$, $\Pr[\cC_e | \tau = \tau_\sigma] = \Omega(n^{-d} \times n^{-1} \times \tau_\sigma/d)$.
\end{lemma}

\begin{proof} The probability of choosing $\bx$ is $n^{-d}$. The probability that $i$ lies in $R$
is ${{d-1}\choose {{\tau_\sigma}-1}}/{d\choose {\tau_\sigma}} = \tau_\sigma/d$. The probability
that $c_i$ equals $\by_i$ is $\frac{1}{n-1}$. That is, $\Pr[\cE_1 \wedge \cE_2 \wedge \cE_3] = \Omega(n^{-d} \times n^{-1} \times \tau_\sigma/d)$. Thus, it remains to show that $\Pr[\cE_4 \wedge \cE_5 ~|~ \cE_1 \wedge \cE_2 \wedge \cE_3] = \Omega(1)$.

Let $T$ denote the set of viable coordinates of $\bx$. Since the maximum degree in the good subgraph $G$ is $k$, we have $|T| \leq k$. Thus,
\begin{align}
    \Pr[(R \setminus i) \cap T = \emptyset] \geq \frac{{d - k \choose \tau_{\sigma} - 1}}{{d \choose \tau_{\sigma} - 1}} = \prod_{\ell=0}^{k-1} \left(1-\frac{\tau_{\sigma}-1}{d-\ell}\right) \geq \left(1-\frac{\tau_{\sigma}}{d-k}\right)^k \geq \exp\left(-\frac{k \cdot \tau_{\sigma}}{d-k}\right)
\end{align}
Recall from \Thm{good-sub} that $k = \Theta(\frac{\eps^2}{\log^4(nd)} \cdot \frac{1}{\sigma^2})$ and from \Clm{small-sigma} we may assume that $\sigma \geq \clarge/\sqrt{d}$. Thus, we have $k = O(d \cdot \frac{\eps^2}{\log^4(nd)})$ and in particular, $k = o(d)$ and so
\begin{align}
    \Pr[(R \setminus i) \cap T = \emptyset] \geq \exp\left(-O\left(\frac{k \cdot \tau_\sigma}{d}\right)\right) = \exp\left(-O\left(\frac{k}{d} \cdot \sigma\sqrt{d}\right)\right) = \exp(-O(\sigma\sqrt{k}))
\end{align}
where the second step used our definition of $\tau_\sigma$ and the last step simply used the fact that $k \leq d$. By \Thm{good-sub} we have $\sigma\sqrt{k} = \Theta(\eps/\log^2(nd))$ and so we have $\Pr[(R \setminus i) \cap T = \emptyset] \geq 9/10$ as long as $nd$ is at least some constant.

Finally, since $\by$ is $\tau_\sigma$-persistent, the probability that this random walk from $\by$ ends at $\bz$ where $f(\bz) = 1$ is at most $1/2$. Thus, by the union bound
\[
\Pr[\cE_4 \wedge \cE_5 ~|~ \cE_1 \wedge \cE_2 \wedge \cE_3] \geq 1-(1/10 + 1/2) = 2/5
\]
and this completes the proof. \end{proof}

\subsection{Wrapping it all up} \label{sec:tester-wrap}

We combine all the bounds and calculations to prove that the path tester
has an $\widetilde{\Omega}(\eps^2/n\sqrt{d})$ probability of success.
%

\begin{proof} (of \Thm{tester}) We first take care of some edge cases.

{\em Case 1, $I_f > 9(n-1)\sqrt{d}$:} By \Clm{tot-neg-inf}, $I^-_f > (n-1)\sqrt{d}$.
Thus, the total number of violated edges of the augmented hypergrid  is at least $(n-1)\sqrt{d} n^d$. The total number of edges is $\frac{1}{2}(n-1)dn^d$. Hence,
a uniform random edge is a violation with probability at least $1/(2\sqrt{d})$.
The path tester selects $\tau = 1$ with probability at least $1/\log d$,
so the rejection probability is at least $1/(2\sqrt{d}\log d)$.

{\em Case 2, $\sigma < \clarge/\sqrt{d}$:} By \Clm{small-sigma},
$I^-_f = \Omega(\eps^2\sqrt{d}/\log^4(nd))$. Thus the probability
that a uniform random edge is a violation is $\Omega(\eps^2/(n\sqrt{d}\log^4(nd)))$.
Similar to the above case, the rejection probability is $\Omega(\eps^2/(n\sqrt{d}\log^5(nd)))$.

{\em Case 3, $I_f \leq 9(n-1)\sqrt{d}$ and $\sigma \geq \clarge/\sqrt{d}$:} 
This is the interesting case, where all the previous claims and lemmas are used.
Since $\sigma \geq \clarge/\sqrt{d}$, we can define $\tau_\sigma$ using \Def{tau-sigma}.
By \Lem{persist}, since $I_f \leq 9(n-1)\sqrt{d}$, 
the fraction of vertices that are not $\tau_\sigma$-persistent
is at most 
\[
\cper \tau_\sigma/\sqrt{d} \leq (2\cper/\clarge)\cdot \sigma < \sigma/4\text{.}
\]
Let us now count the number of viable edges in the good subgraph promised by \Thm{good-sub}. There are at least $\sigma k n^d/2$ edges in the good subgraph. There are at most $(\sigma/4) \cdot n^d$ non-persistent vertices in the good subgraph, each of which has degree at most $k$. Thus, removing all non-persistent vertices leaves us with 
\[
\sigma k n^d/2 - (\sigma/4) \times n^d \times k = \sigma k n^d/4
\]
edges, all of which are viable. Let us now lower bound the probability of the tester
rejecting using \Lem{viable}. Recall that all the events $\cC_e$ are disjoint and $\tau_\sigma = \Theta(\sigma\sqrt{d})$.
\begin{eqnarray*}
\Pr\left[\bigcup_{e \ \textrm{viable}} \cC_e | \tau = \tau_\sigma\right] & = & \sum_{e \ \textrm{viable}} \Pr[\cC_e | \tau = \tau_\sigma]\\
& \geq & (\sigma k/4) n^d \times n^{-d} \times n^{-1} \times (\tau_\sigma/d) \\
& = & \Omega(\sigma^2 k/(n\sqrt{d})) = \Omega(\eps^2/(n\sqrt{d}\log^4(nd))) \text{.}
\end{eqnarray*}
The probability of setting $\tau$ to be $\tau_\sigma$
is $1/\log d$, so we multiply the bound above by $1/\log d$ to complete the proof. \end{proof}

%% file: thoughts-on-psi.tex
\section{Towards a $\widetilde{O}(\sqrt{d})$ tester} \label{sec:no-n}

In this section we describe a different notion of influence of Boolean functions on hypergrids. We conjecture a Talagrand style isoperimetric theorem for the colored version of this influence is true. If so, then we can design a tester whose query complexity has no polynomial dependence on $n$. More precisely, the dependence on $n$ is only polylogarithmic, and since by results of~\cite{BlackCS20,HY22} one can assume $n = \poly(d)$, the final tester's query complexity is $\widetilde{O}(\sqrt{d})$.

\begin{definition}[Weighted Influence]\label{def:psi-f}
		Fix $f:[n]^d \to \{0,1\}$ and a dimension $i\in [d]$. Fix a point $\bx \in [n]^d$. The {\em weighted} influence of $\bx$ along coordinate $i$ is defined as
		\[
		\Psi_f(\bx;i) := \sum_{\by = \bx\pm a\be_i~:~ \textrm{such that} ~(\bx, \by)~ \textrm{is violating}} ~~\frac{1}{a}
		\]
		The weighted influence of $\bx$ is $\Psi_f(\bx) = \sum_{i=1}^d \Psi_f(\bx;i)$.
	\end{definition}

Consider giving a {\em weight} to every edge $(\bx, \bx \pm a\be_i)$ equal to $1/a$, the reciprocal of the length of the edge. The weighted influence of a point $\bx$ with $f(\bx) = 1$ is
the sum of the weights of {\em out}-edges which are violating. This is another generalization of the notion of influence in hypercubes. Also note that for any $\bx$, the thresholded influence 
of any point can't be much smaller than the weighted influence; indeed,  $\Phi_f(\bx) \geq \frac{1}{H_n} \Psi_f(\bx)$ where $H_n$ is the $n$th Harmonic number and is $\Theta(\log n)$.
Therefore, lower bounding $\Psi$'s also lower bounds $\Phi$'s. On the  other hand, $\Phi_f(\bx;i)$ could be as large as $(n-1)\cdot \Psi_f(\bx;i)$ for a particular $\bx$. 

The Talagrand objective with respect to this notion is defined as the expected {\em square root} of the weighted influence.

\begin{definition}[Talagrand Objective wrt Weighted Influence]
		Given any Boolean function $f:[n]^d \to \{0,1\}$, we define the Talagrand objective with respect to the weighted influence as
		\[
			T_{\Psi}(f) :=  \Exp_{\bx\in [n]^d} \left[ \sqrt{\Psi_f(\bx)}\right]
		\]
		where $\Psi_f$ is as defined in~\Cref{def:psi-f}.
\end{definition}

We also have a notion of colorful weighted influence and the corresponding Talagrand objective.
\begin{definition}[Colorful Weighted Influence]\label{def:psi-f-chi}
		Fix $f:[n]^d \to \{0,1\}$ and $\chi:E\to \{0,1\}$. Fix a  dimension $i\in [d]$. Fix a point $\bx \in [n]^d$. The colorful weighted influence of $\bx$ along coordinate $i$ is defined as
		\[
	\Psi_{f,\chi}(\bx;i) := \sum_{\by = \bx\pm a\be_i~:~ \textrm{such that} ~(\bx, \by)~ \textrm{is violating and}~ \chi(\bx, \by) = f(\bx)} ~~\frac{1}{a}
	\]
	The colorful weighted influence of $\bx$ is $\Psi_{f,\chi}(\bx) = \sum_{i=1}^d \Psi_{f,\chi}(\bx;i)$.
	\end{definition}

\begin{definition}[Colorful Thresholded Talagrand Objective]
		Given any Boolean function $f:[n]^d \to \{0,1\}$ and $\chi:E \to \{0,1\}$, 
		we define the Talagrand objective with respect to the colorful weighted influence as
	\[
	 T_{\Psi_\chi}(f) := \Exp_\bx~\left[\sqrt{\Psi_{f,\chi}(\bx)}  \right]
	\]
		where $\Psi_{f,\chi}$ is as defined in~\Cref{def:psi-f-chi}.
\end{definition}
We are now ready to state our main conjecture. Note that due to the fact that $\Psi_{f,\chi}(\bx) = O(\log n)\Phi_{f,\chi}(\bx)$, the conjecture below generalizes~\Cref{thm:dir-tal} up to $\poly\log n$ factors.
\begin{mdframed}[backgroundcolor=gray!20,topline=false,bottomline=false,leftline=false,rightline=false] 
	\begin{conjecture}\label{conj:dir-tal}~
		Let $f:[n]^d \to \{0,1\}$ which is $\eps$-far from monotone, and let $\chi:E\to \{0,1\}$ be an arbitrary coloring of the edges of the augmented hypergrid. Then there exists a constant $c \geq 0$ such that
		\[
			T_{\Psi_\chi}(f) := \Exp_{\bx\in [n]^d}~\left[\sqrt{\Psi_{f,\chi}(\bx)}  \right] = \Omega\left(\frac{\eps}{\log^c n}\right)
		\]
		where $\Psi_{f,\chi}(\bx)$ is as defined in~\Cref{def:psi-f-chi}.
	\end{conjecture}
\end{mdframed}

If the above conjecture is true, then we can design a $\widetilde{O}(\sqrt{d})$ tester. Before we state this tester, and indeed to motivate it, we first note that the analysis of the tester described in~\Cref{fig:alg} is tight on its dependence on $n$. 
An example is the so-called {\em centrist} function defined in~\cite{BlackCS20} (Section 8) defined as follows for the case $n=d$, that is, $f:[d]^d \to \{0,1\}$.
\[
f(\bx) = \begin{cases}
	0 & \textrm{if}~~\exists i\in [d]~~~ \bx_i = 2 \\
	1 & \textrm{otherwise}
\end{cases}
\]
\cite{BlackCS20} (Claim 8.2) showed that $\eps_f = \Omega(1)$ for the above function. To see this, observe that the probability of a random point having $\bx_i \neq 2$ for all $i \in [d]$ is
$\left(1-\frac{1}{d}\right)^d = \Theta(1)$. Now let's consider the rejection probability of the algorithm in~\Cref{fig:alg}. Note that a violating pair $(\bx, \by)$ must satisfy $\bx_i = 1$
and $\by_i = 2$ for some $i$, and that $i$ needs to be picked in the random set $R$. Even conditioning on picking that $i\in R$, the probability $c_i$ would be set to $2$ is $\frac{1}{n-1} \approx \frac{1}{d}$ in this case. Therefore, the algorithm needs to be modified. 

To fix this, think about the tester as performing a (lazy) directed random walk. And now when one picks a dimension $i$ to move in, one doesn't have a {\em uniform} distribution over the length it moves in this dimension, as currently is the case in~\Cref{fig:alg}, rather one chooses a length using a discrete Pareto distribution. That is, smaller lengths are given more weight than longer lengths. More precisely, we pick a length $a$ to move with probability proportional to $\frac{1}{a}$; note that the constant of proportionality that we need to scale down by is at most $H_n = \Theta(\log n)$. This is the reason we defined the weighted influence as we did in~\Cref{def:psi-f}, and indeed, if~\Cref{conj:dir-tal} is true, then we can prove that this tester is an $\widetilde{O}(\sqrt{d})$ tester; the proof technique is similar to that described in~\Cref{sec:tester} which itself is a modification of the analysis in~\cite{KMS15}. \smallskip

We now show that the {\em uncolored} version of~\Cref{conj:dir-tal} is in fact true, and follows easily using the colorful version~\Cref{thm:dir-tal} for thresholded influence. 
\begin{theorem}
		Let $f:[n]^d \to \{0,1\}$ which is $\eps$-far from monotone. Then,
		\[
		T_{\Psi}(f) := \Exp_{\bx\in [n]^d}~\left[\sqrt{\Psi_{f}(\bx)}  \right] = \Omega\left(\frac{\eps}{\log n}\right)
		\]
		where $\Psi_{f}(\bx)$ is as defined in~\Cref{def:psi-f}.
\end{theorem}
\begin{proof} The key is that we can always define a coloring $\chi$ for which $T_{\Psi}(f) = \Omega(T_{\Phi_{\chi}}(f))$ simply as follows. For every edge $(\bx,\bx+a\be_i)$ of the fully augmented hypergrid, define $\chi(\bx,\bx+a\be_i) := 1$ if the interval $[\bx,\bx+a\be_i]$ is at least half $0$'s and $\chi(\bx,\bx+a\be_i) := 0$ otherwise. This coloring achieves the desired property because of the following simple fact.

\begin{fact} Let $S \subseteq \{1,\ldots,n\}$ be of size $|S| \geq \lceil n/2 \rceil$. Then $\sum_{a \in S} \frac{1}{a} = \Omega(1)$. \end{fact}

\begin{proof} For simplicity let $n$ be even. The sum is minimized when $S = [n] \setminus [n/2]$ and in this case the sum is equal to $H_n - H_{n/2} = \Omega(1)$. \end{proof}

Now consider $\bx$ and $i \in [d]$ such that $f(\bx) = 1$ and $\Phi_{f,\chi}(\bx;i) = 1$. This implies that there exists an $i$-edge $(\bx,\by)$ for which $f(\by) = 0$ and $\chi(\bx,\by) = 1$. By definition of the coloring, the interval $[\bx,\by]$ is at least half $0$'s and so $\Psi_f(\bx;i) = \Omega(1)$. Similarly, consider $\by$ and $i \in [d]$ such that $f(\by) = 0$ and $\Phi_{f,\chi}(\by;i) = 1$. This implies that there exists an $i$-edge $(\bx,\by)$ for which $f(\bx) = 1$ and $\chi(\bx,\by) = 0$. By definition of the coloring, the interval $[\bx,\by]$ is at least half $1$'s and so $\Psi_f(\by;i) = \Omega(1)$. Thus, $T_{\Psi}(f) = \Omega(T_{\Phi_{\chi}}(f))$ and so invoking \Cref{thm:dir-tal} completes the proof. \end{proof}
It is worthwhile to point out the challenge in generalizing the above theorem to the colored version. Note the above was a ``point-by-point'' argument in that we found a coloring $\chi$ of the fully augmented hypergrid edges such that for every $\bx$ and every $i\in [d]$, we could prove $\Psi_f(\bx;i) = \Omega(\Phi_{f,\chi}(\bx;i))$. One would wonder if such a point-by-point analysis is possible even when we have an arbitrary coloring $\chi'$, and the LHS in the previous statement is replaced with the colored version. Unfortunately, this is not possible. One can find examples of $f$ and $\chi'$ such that no matter how you define $\chi$, there will be some point $\bx$ and some dimension $i$ such that $\Phi_{f, \chi}(\bx;i) = 1$, but $\Psi_{f,\chi'}(\bx;i) \approx 1/n$. These examples do not disprove the conjecture since, in these examples, for a constant fraction of $(\bx, i)$ pairs, we do have $\Psi_{f,\chi'}(\bx;i) = \Omega(\Phi_{f,\chi}(\bx))$, but it does point to the need of a new argument. One could also wonder if the proof technique used to prove~\Cref{thm:dir-tal} can bear upon the proof of the conjecture. There are many roadblocks here, one of them primarily being that semisorting can {\em increase} the Talagrand objective with respect to the $\Psi$-influences mainly because it can bring violations closer which bumps up the $\Psi$-influence. Nevertheless, the authors believe that~\Cref{conj:dir-tal} is true, and we leave the resolution of this as a promising direction towards getting rid of the polynomial dependence on $n$ thereby resolving the Boolean monotonicity testing question on hypergrids.